\documentclass[ejsv2,noshowframe]{imsart}

\RequirePackage[authoryear]{natbib}%
\RequirePackage{graphicx}

\arxiv{2412.14391}
\startlocaldefs
\usepackage{amsthm,amsmath,amsfonts}
\usepackage{booktabs}       %
\usepackage{nicefrac}       %
\graphicspath{{Fig/}}     %
\usepackage{enumitem}
\usepackage{mathtools}
\usepackage{bbm}
\usepackage{harpoon}
\usepackage{xspace}
\usepackage[usenames,dvipsnames]{xcolor}
\usepackage{multirow}
\usepackage{algorithm}
\usepackage{algpseudocode}
\usepackage{listings}%
\usepackage{afterpage}

\usepackage{tikz}
\usepackage{tikz-3dplot}
\usetikzlibrary{tikzmark,calc,arrows.meta}

\tdplotsetmaincoords{60}{110}
\pgfmathsetmacro{\radius}{0.75}
\pgfmathsetmacro{\thetavec}{0}
\pgfmathsetmacro{\phivec}{0}

\RequirePackage[colorlinks,citecolor=blue,urlcolor=blue]{hyperref}

\usepackage[capitalize,noabbrev]{cleveref}

\newlist{properties}{enumerate}{10}
\crefname{propertiesi}{Property}{Properties}
\Crefname{propertiesi}{Property}{Properties}
\newlist{tests}{enumerate}{10}
\crefname{testsi}{}{}
\Crefname{testsi}{}{}
\crefname{example}{Example}{Examples}
\Crefname{example}{Example}{Examples}
\crefname{assumption}{Assumption}{Assumptions}
\Crefname{assumption}{Assumption}{Assumptions}

\theoremstyle{plain}

\newtheorem{theorem}{Theorem}[section]
\newtheorem{lemma}[theorem]{Lemma}
\newtheorem{proposition}[theorem]{Proposition}
\newtheorem{corollary}[theorem]{Corollary}
\theoremstyle{definition}
\newtheorem{definition}[theorem]{Definition}

\newtheorem{example}{Example}
\newtheorem{assumption}{Assumption}
\theoremstyle{remark}

\def\bfH{\mathbf{H}}

\def\bfM{\mathbf{M}}

\def\bfR{\mathbf{R}}
\def\bfS{\mathbf{S}}

\def\bfW{\mathbf{W}}
\def\bfX{\mathbf{X}}
\def\bfY{\mathbf{Y}}
\def\bfZ{\mathbf{Z}}

\def\bfe{\mathbf{e}}

\def\calA{\mathcal{A}}

\def\calF{\mathcal{F}}

\def\bbP{\mathbb{P}}

\def\tR{\tilde{R}}

\def\tx{\tilde{x}}
\def\tY{\tilde{Y}}
\def\ty{\tilde{y}}

\def\I{\mathbf{I}}

\def\R{\mathbb{R}}

\def\E{\mathbb{E}}
\def\e{\mathbf{e}}

\def\lip{\langle}
\def\rip{\rangle}
\def\gaussian{\mathsf{N}}

\newcommand{\argdot}{{\,\vcenter{\hbox{\tiny$\bullet$}}\,}}

\def\onevec{1}

\def\indi{\mathbbm{1}} %
\def\condind{{\;\perp\!\!\!\perp\;}} %
\def\equdist{\stackrel{\text{\rm\tiny d}}{=}} 
\def\nequdist{\stackrel{\text{\rm\tiny d}}{\neq}} %
\def\todist{\stackrel{\text{\rm\tiny d}}{\to}} %

\def\equas{\stackrel{\text{\rm\tiny a.s.}}{=}} 
\def\toas{\stackrel{\text{\rm\tiny a.s.}}{\to}} 
\def\simiid{\stackrel{\text{\rm\tiny iid}}{\sim}}
\def\define{\coloneqq}

\def\eps{\varepsilon}

\def\grp{\mathbf{G}}
\def\grpAct{\Phi}
\def\orbit{O}
\def\maxInv{M}
\def\tM{\tilde{\maxInv}}
\def\haar{\lambda}
\def\rhaar{\tilde{\haar}}
\newcommand{\orbSel}{\rho}
\newcommand{\orbRep}{r}
\newcommand{\orbRepRand}{R}
\def\repInv{{\psi^{\circ}}}
\def\repInvRand{\Psi}
\def\repInvRandCond{\psi}
\def\repInvKern{P_{\repInvRand|X}}
\def\id{\textrm{id}}

\newcommand{\invProbs}[1][\bfX]{\mathcal{P}^{\circ}(#1)}

\newcommand{\probs}[1][\bfX]{\mathcal{P}(#1)}

\def\G{\Gamma}
\def\g{\gamma}
\def\GRand{\mathsf{\Gamma}}
\def\GKern{P_{\G|R}}
\def\tYRand{\mathsf{\tY}}
\def\tYKern{P_{\tY|R}}

\def\orbRepKern{P_{R}}
\def\XorbRepKern{P_{X|R}}

\newcommand{\SO}[1]{\mathrm{SO}(#1)}
\newcommand{\Sym}[1]{\mathbb{S}_{#1}}
\newcommand{\Trans}[1]{\mathbb{T}_{#1}}
\def\lorentz{\mathrm{SO}^+(1,3)}

\def\test{\phi}
\def\testStat{T}
\def\hatTestStat{\widehat{\testStat}}
\def\nullHyp{H_0}
\def\altHyp{H_1}

\def\sig{\alpha}

\def\dirac{\delta}

\def\bandwidth{h_m}
\def\kde{k_{\bandwidth}}
\def\kparam{\sigma}
\def\k{k_{\kparam}}
\def\hilbert{\mathcal{H}}

\newcommand{\kme}[1]{\mu_{#1}}
\def\MMD{\mathrm{MMD}_{\kparam}}

\def\supk{\kappa}
\def\nusup{\supk_\lor}
\def\EMMD{\Delta_n}
\def\supk{\kappa}
\def\EH{\xi_n}
\def\EHell{\xi_{n,\ell}}
\def\hatDelta{\hat{\Delta}_n}
\def\barDelta{\bar{\Delta}_n}
\def\supDelta{\Delta^\lor_n}

\def\ZRand{Z}

\def\Pnull{P^0}

\def\nBS{\textsc{BS}\xspace}
\def\nRN{\textsc{RN}\xspace}
\def\nFS{\textsc{FS}\xspace}
\def\nSK{\textsc{SK}\xspace}

\def\nBSFS{\nBS-\nFS}
\def\nBSSK{\nBS-\nSK}
\def\nRNFS{\nRN-\nFS}
\def\nRNSK{\nRN-\nSK}

\makeatletter
\AddToHook{cmd/appendix/before}{%
    \crefalias{section}{appendix}%
    \crefalias{subsection}{appendix}
    \crefalias{subsubsection}{appendix}
}
\makeatother

\allowdisplaybreaks %
\endlocaldefs

\begin{document}
\begin{frontmatter}
\title{Randomization Tests for \\ Conditional Group Symmetry}
\runtitle{Randomization Tests for Conditional Group Symmetry}

\begin{aug}
\author[A]{\fnms{Kenny}~\snm{Chiu}\ead[label=e1]{kenny.chiu@stat.ubc.ca}},
\author[A]{\fnms{Alex}~\snm{Sharp}\ead[label=e2]{alex.sharp@stat.ubc.ca}}
\and
\author[A]{\fnms{Benjamin}~\snm{Bloem-Reddy}\ead[label=e3]{benbr@stat.ubc.ca}}
\address[A]{Department of Statistics,
The University of British Columbia\printead[presep={,\ }]{e1,e2,e3}}
\runauthor{K.\ Chiu, A.\ Sharp, and B.\ Bloem-Reddy}
\end{aug}

\begin{abstract}
Symmetry plays a central role in the sciences, machine learning, and statistics. While statistical tests for the presence of distributional invariance with respect to groups have a long history, tests for conditional symmetry in the form of equivariance or conditional invariance are absent from the literature. This work initiates the study of nonparametric randomization tests for symmetry (invariance or equivariance) of a conditional distribution under the action of a specified locally compact group. We develop a general framework for randomization tests with finite-sample Type I error control and, using kernel methods, implement tests with finite-sample power lower bounds. We also describe and implement approximate versions of the tests, which are asymptotically consistent. We study their properties empirically using synthetic examples and applications to testing for symmetry in two problems from high-energy particle physics.
\end{abstract}



\end{frontmatter}
\tableofcontents

\section{\label{sec:intro}Introduction}

Symmetry plays a central role in science and engineering, and appears in numerous statistical problems, both classical \citep{lehmann_theory_1998,Lehmann:Romano:2005} and modern \citep[e.g.,][]{Cohen:2019,Chen:2020,Huang:2022aa,Bronstein:2021aa}. There is an extensive literature on hypothesis tests for symmetry in the form of distributional invariance under a group of transformations, dating back at least to the work of \citet{Hoeffding_1952}, which generalized older ideas based on testing permutation-invariance \citep{Fisher:1935,Pitman:1937}. More recent developments in the area pertain to randomized versions of tests for invariance \citep{Romano:1988,Romano:1989,hemerik:goeman:2017:fdr,Hemerik2018,Dobriban_2022,Ramdas:2023,Koning:2022} or approximate invariance \citep{canay2017}; see \cite{Ritzwoller:2024} for a recent review. One takeaway from that line of work is that group-based randomization tests have finite-sample Type I error control \citep{Hoeffding_1952}, power against alternatives \citep{Romano:1989,Dobriban_2022}, and achieve various notions of optimality and robustness under appropriate conditions \citep{Romano:1990,Kashlak:2022,Kim:2022:b,Koning:2022}. 

The literature on group-based randomization tests deals with the invariance of an unknown probability distribution $P_X$ from which samples $X_1,\dotsc,X_n$ are obtained. 
In this paper, we initiate the study of hypothesis tests for the symmetry of a \emph{conditional} distribution, $P_{Y|X}$, to be tested from samples $(X_1,Y_1),\dotsc,$ $(X_n,Y_n)$.  
The primary motivation for conditional tests is the central role played by symmetry in modern science \citep{Rosen_1995,Gross_1996} and the increasingly data-driven nature of scientific discovery, particularly in settings in which the relationships between multiple objects are a key feature \citep{Kofler:2012,atlas:2017:search,Collins2018anomaly,Cristadoro:2018,Chmiela:2018,Karagiorgi_2022,Birman2022,yim2023SE,watson2023novo}.
A secondary motivation is the prominence of symmetry in certain prediction problems from machine learning;  \citet{Bronstein:2021aa} provide a recent review. 

\begin{figure}[t]
    \centering
    \begin{tikzpicture}
        \node (fig) at (0,0) {\includegraphics[trim={27 26 0 0},clip,width=\textwidth]{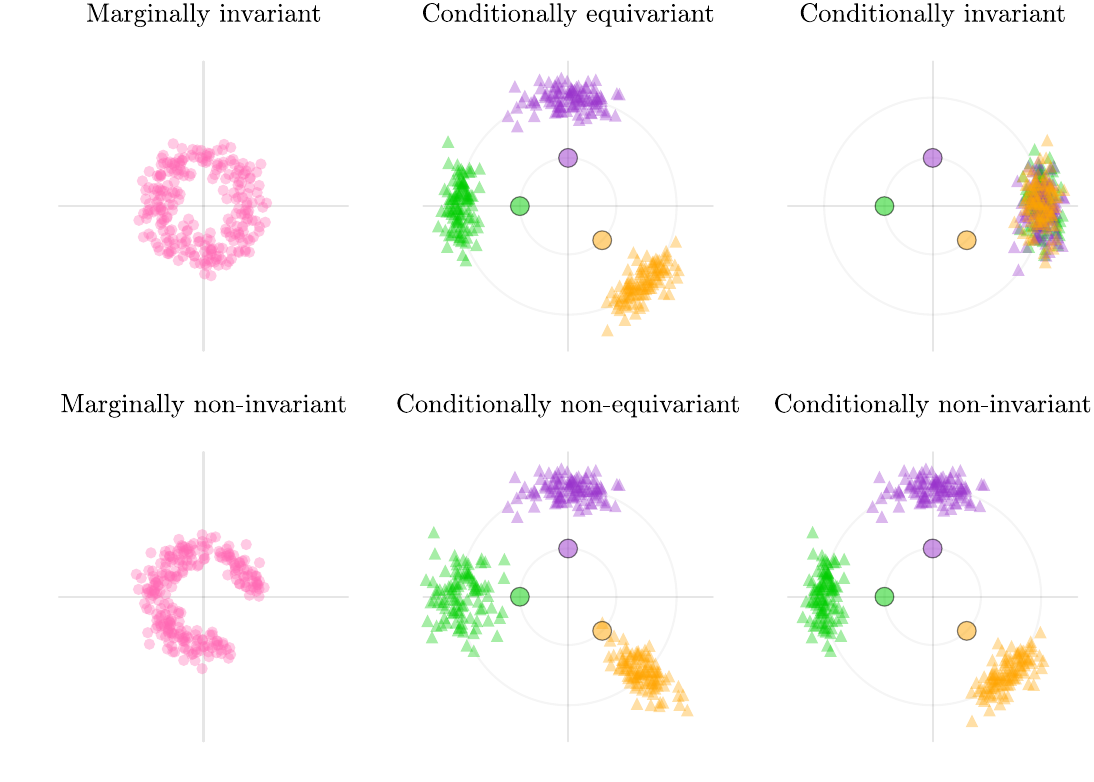}};
        \draw[opacity=0.4] (-2.8,5) -- (-2.8,-5);
        \draw[opacity=0.4] (2.55,5) -- (2.55,-5);
    \end{tikzpicture}
    \caption{Column~1: a random sample from a distribution on $\R^2$. A distribution is invariant with respect to rotations about the origin if rotating the distribution does not change it. Columns~2 and 3: conditional distributions $P_{Y|X}(\argdot|x)$ for $Y\in\R^2$ (triangle) given three different values of $x\in\R^2$ (circle), matched by color. A conditional distribution is rotationally-equivariant if rotations of $x$ transform the distribution of $Y| X=x$ with a corresponding rotation, and conditionally invariant if rotations of $x$ do not change the distribution of $Y| X=x$.}
    \label{fig:symmetry}
\end{figure}

\cref{fig:symmetry} illustrates the differences between marginal and conditional symmetries with a simple example.
Whereas symmetry of a probability distribution manifests as invariance under a group of transformations, a conditional distribution may exhibit a form of symmetry that only appears in relationships between variables, known as \emph{equivariance}: a transformation of $X$ induces a corresponding transformation of $Y$.
Under certain conditions, detecting equivariance can be reduced to a problem for which existing group-invariance tests are applicable, but those conditions are not satisfied in many settings.
Tests for (marginal) distributional group symmetry are based on the fact that an invariant distribution factors into the independent product of the group's ``uniform distribution'' and a distribution over group-induced equivalence classes known as orbits---a characterization that follows from the fact that marginal invariance can only hold with respect to compact groups. The independence between group transformation and orbit is key to the computational and statistical simplicity of tests for marginal invariance. 
In contrast, conditional symmetries do not have such a convenient characterization, in part because they may hold with respect to a non-compact group, for which a uniform distribution does not exist.
As a consequence, designing tests for conditional symmetry requires overcoming new technical and statistical challenges.

\paragraph{Contributions} 
The present work initiates the study of hypothesis testing for conditional group symmetry and formulates a general framework for such tests based on a single independent and identically distributed (i.i.d.) sample. 
In order to do so, we 
develop a number of novel characterizations of the invariance and equivariance of a conditional distribution under the action of a locally compact group.
The tests we develop can be used to detect the presence or absence of a particular conditional symmetry in data, with that symmetry specified by hypothesis.
The tests also may be used as model-checking criteria for models meant to exhibit a symmetry. 

In addition to a general formulation of symmetry-based conditional randomization tests for which we prove Type I error control properties, we provide specific kernel-based instantiations of our tests, which allow the tests to be used on any data type for which a kernel can be defined. The testing procedure approximates the conditional distribution of the estimated maximum mean discrepancy between two samples from a suitable null hypothesis distribution, conditioned on certain symmetry-induced statistics of the observations.  Using this test, we derive a finite-sample lower bound on the test's power against non-symmetric alternatives, which also implies that the test is asymptotically consistent as the number of observations $n\to\infty$. 
We explore the properties of our tests empirically on synthetic data and demonstrate their application to two problems from high-energy particle physics.

\paragraph{Related work} To our knowledge, general-purpose tests for conditional symmetry have not appeared in the existing literature. 
The focus of this paper is on tests for conditional {group} symmetry, which should not be confused with the literature that tests conditional distributions for symmetry about zero (or some other value) for every value of the conditioning variable \citep[e.g.,][]{Zheng:1998,Bai:2001,Su:2006}. 
In contrast to the present work, those tests do not consider symmetry in the relationship between the variables and are applicable to a disjoint set of problems from those considered here. 
All mentions of conditional symmetry in this work will refer to the group invariance or equivariance of a conditional distribution, as we define in \cref{sec:background:symmetry}.

Some recent methods in machine learning aim to estimate symmetry groups from data \citep{Krippendorf_2020,zhou2021metalearning,Desai:2021,dehmamy2021automatic,Yang:2023aa}.
They do not address the corresponding inferential questions, and the present work can be viewed as a complement to those estimation methods. 

The only work that we are aware of that addresses any notion of testing for conditional symmetry is that of \citet{Christie:2023}, who proposed tests for the invariance of the conditional expectation $f(x) = \E[Y|X=x]$, $f \colon \bfX \to \R$, under the action of a specified group.\footnote{We note that the primary aim of \citet{Christie:2023} is to estimate the \emph{largest} group under which $f$ is invariant, which amounts to performing a collection of tests over a subgroup lattice of some candidate largest group.
In principle, our tests could be substituted into their procedure, though we do not address that problem here.}
Testing for the invariance of the conditional expectation is not the same as testing the invariance of the conditional distribution, and those authors' focus on the former objective leads them to formulate tests that rely on certain regularity properties of $f$ (e.g., smoothness), for which there are no distributional analogues. 
In contrast, the tests that we develop here are applicable to a different set of problems: testing for distributional symmetry, and with theoretical properties that do not rely on regularity of the conditional moments. 
Moreover, the statistical properties of the tests in \citep{Christie:2023} are sensitive to the choice of a number of user-specified parameters. In particular, a sampling distribution on the group is chosen by heuristic. In the tests we develop, the required sampling distribution over the group and other test parameters are estimated from data in such a way that preserves the validity of the test while improving power against alternatives. 

The methods we develop rely on group-theoretic inversion probability kernels, which have not been used widely in statistics or machine learning, likely owing to their relatively recent appearance in probability \citep{Kallenberg2011:skew}. 
Since their appearance in a preliminary version of this work \citep{chiu2023hypothesis}, inversion kernels have been used to solve problems involving symmetry in machine learning \citep{cornish2024stochastic,lawrence2025improving,zhang2025symdiff}.
Using inversion kernels, we characterize conditional symmetry through a particular conditional independence relation, generalizing a result of \citet{BloemReddy:2020} that holds in a much more limited setting. Those authors used the result to study equivariant neural networks and did not consider testing for conditional symmetry.
Deterministic inversions have been used more widely in the machine learning literature \citep[e.g.,][]{BloemReddy:2020,winter2022unsupervised,Kaba:2023}, and have counterparts in classical statistical problems involving groups \citep{Fraser:1961:fiducial:invariance,Eaton:1989,eaton:sudderth:1999:consitency}.

\paragraph{Outline} In \cref{sec:background}, we review the necessary basics of groups (additional technical aspects are contained in \cref{apx:proper:action}), give an overview of randomization tests for group invariance,
and introduce the framework for conditional symmetry along with some examples. \cref{sec:condind} is the main technical section, in which conditional symmetry is characterized via conditional independence and certain distributional identities. We propose a general framework for conditional randomization tests of symmetry based on those results in \cref{sec:rand}. In \cref{sec:kernel}, we describe specific nonparametric tests based on kernel maximum mean discrepancy, along with finite-sample lower bounds on power against non-symmetric alternatives. \cref{sec:experiments} contains experiments that apply our tests to synthetic data and to particle physics data.

\section{\label{sec:background}Background}

We first introduce some technical necessities.
Throughout, we use $\bfX$ to denote a topological space and $\bfS_{\bfX}$ its Borel $\sigma$-algebra, so that $(\bfX, \bfS_{\bfX})$ is a standard Borel measurable space.
Unless stated otherwise, we refer to $\bfX$ as a measurable space, and likewise for $\bfY$, $\bfM$, and $\bfR$ (distinct from $\R$, the space of real numbers).
Let $\probs$ denote the set of all probability measures on $\bfX$.
For $P\in\probs$, we write $X\sim P$ for a random variable $X$ sampled from $P$.  A vector or tuple of random variables is denoted by $X_{1:n} \coloneq (X_1,\dotsc,X_n)$. 
For any measurable function $f$ on $\bfX$, let $f_*P$ denote the pushforward measure with $f_{*}P(A) = P(f^{-1}(A))$ for all measurable sets $A \in \bfS_{\bfX}$.
We use $\dirac_x$ to denote the Dirac measure at a point $x$, and $\equdist$ to denote equality in distribution. Let $[n]$ denote the set $\{1,2,\dotsc,n\}$.

\subsection{Groups and group actions} \label{sec:background:groups}

This work focuses on symmetries that can be encoded in the mathematics of group theory, which comprise many of the symmetries of interest in modern statistical, computational, and scientific problems. We assume a basic familiarity with groups and group actions; see \cref{apx:group:background} for background. 
Thorough treatments of classical group theoretic applications in statistics, along with relevant mathematical details, can be found in \citep{Eaton:1989,WIjsman:1990}.

We denote a group by $\grp$ and its elements by $g$. For convenience, group composition is written multiplicatively, i.e, $g_1g_2=g_1\cdot g_2$. 
We assume that $\grp$ has a topology that is locally compact, second countable, and Hausdorff (lcscH), and which makes the group operations continuous.
We may then take $\bfS_{\grp}$ as the Borel $\sigma$-algebra, making $\grp$ a standard Borel space. 
When $\grp$ is lcscH, there exist left- and right-invariant $\sigma$-finite measures $\haar_{\grp}$ and $\rhaar_{\grp}$, known as \textit{left-} and \textit{right-Haar measures}, respectively, that are unique up to scaling~\citep[Ch.~2.2]{Folland:2016}.
For convenience, we omit the subscript and use $\haar$ to denote left-Haar measure.
If $\grp$ is compact, then $\haar = \rhaar$, and the unique normalized Haar measure acts as the uniform probability measure over the group.

A group $\grp$ acts measurably on a set $\bfX$ if the group action $(g,x) \mapsto gx$ is a measurable function. 
Throughout, all actions are assumed to be continuous (and therefore measurable), which we imply by saying that $\grp$ acts on $\bfX$. 
For a set $A\subseteq\bfX$, $gA=\{gx : x\in A\}$.
The \textit{$\grp$-orbit} of $x \in \bfX$ is the subset of all points reachable through an action in $\grp$, denoted $\orbit(x)=\{gx : g\in\grp\}$.
If $\bfX$ has only one orbit, then the action is \textit{transitive}. 
On each orbit $\orbit(x)$, we can choose an element $\orbRep_x\in\orbit(x)$ as the \textit{orbit representative}, so that $\orbRep_x=gx$ for some $g\in\grp$.
The element $\orbRep_x$ can be thought of as an ``origin'' on the orbit. 
For a particular choice of representatives, the subset of $\bfX$ consisting of each orbit's representative is denoted by $\bfR$. A generic element of $\bfR$ is denoted by $r$, whereas $r_x$ indicates the element of $\bfR$ on the orbit of $x$, i.e., $\bfR \cap \orbit(x) = \{r_x\}$.
A function $\orbSel \colon \bfX \to \bfR$ that maps elements of $\bfX$ onto their corresponding orbit representatives is called an \textit{orbit selector}; that is, $\orbSel(x)=\orbRep_x$.
Any orbit selector makes for a useful choice of a \textit{maximal invariant}, which more generally is a $\grp$-invariant function $\maxInv \colon \bfX \to \bfM$ that takes a different value on each orbit, so that $\maxInv(x) = \maxInv(x')$ if and only if $gx = x'$ for some $g \in \grp$. 
Maximal invariants are not unique; we use $\maxInv$ to denote a generic maximal invariant.
\Cref{fig:group} depicts these quantities for $\SO{d}$---the compact group of $d$-dimensional rotations about the origin---acting on $\R^d$, for $d=2$ and $d=3$. (\Cref{fig:group} also depicts \textit{representative inversions} $\repInv$, which are defined in \cref{sec:condind}.) 
The \textit{stabilizer subgroup} of $x\in\bfX$ is the subset of group actions that leave $x$ invariant, denoted $\grp_x = \{ g \in \grp : gx = x \}$.
The group action is \textit{free} if $\grp_x = \{\id\}$ for all $x \in \bfX$.

In the analysis of probabilistic aspects of group actions, measurability issues can arise without regularity conditions. The key regularity condition that we assume in this work is that the group action is \textit{proper}. Properness is a common assumption in statistical applications of group theory \citep[e.g.,][]{Eaton:1989,WIjsman:1990,McCormack2023} and is satisfied in many settings of interest. We discuss the implications of this assumption in \cref{apx:proper:action}.

\begin{figure}[t]
\centering
{
\hfill
\begin{tikzpicture}
\draw[gray] (-1,0) -- (0,0);
\draw[gray] (0,-1) -- (0,3);
\draw[ForestGreen,line width=1.5] (0,0) -- (4,0);
\node[label={[font=\large,text=ForestGreen]above:{$\bfR$}}] (orbreps) at (4,0) {};
\filldraw[gray] (0,0) circle (3pt);
\node[label={[font=\footnotesize,text=gray]below left:{$(0,0)$}}] (origin) at (0.05,0) {};
\draw[Periwinkle,line width=1.5] (2.291,-1) arc(-23.578:113.578:2.5);
\node[label={[font=\large,text=Periwinkle]below:{$\orbit(x)$}}] (orbit) at (-1,2.291) {};
\node[circle,draw=black,fill=Thistle,minimum size=10pt,label={[font=\large,text=Thistle]above right:{$x=\repInv(x)\orbRep_x$}}] (x) at (1.25,2.165) {};
\node[circle,draw=black,fill=BurntOrange,minimum size=10pt,label={[font=\large,text=BurntOrange]below right:{$\orbRep_x$}}] (rho) at (2.5,0) {};
\draw[-{Stealth[length=3mm]},dotted,OrangeRed,line width=1.5] (2.730,0.335) arc(7:53:2.75);
\node[label={[font=\large,text=OrangeRed]below right:{$\repInv(x)$}}] (orbsel) at (2.45,1.85) {};
\draw[VioletRed,{Bar[width=3mm]}-{Bar[width=3mm]},dashed] (0.2,-0.3) -- (2.3,-0.3);
\node[label={[font=\large,text=VioletRed]below:{$\maxInv(x)$}}] (maxinv) at (1.25,-0.3) {};
\end{tikzpicture}
\hfill
\begin{tikzpicture}[scale=2.5,tdplot_main_coords]
\draw[gray,dotted,opacity=0.75] (-0.5,0,0) -- (0,0,0); %
\draw[gray,dotted,opacity=0.75] (0,-0.5,0) -- (0,0,0); %
\draw[gray,dotted,opacity=0.75] (0,0,-0.5) -- (0,0,0); %
\draw[gray,opacity=0.25] (0,0,0) -- (0.75,0,0); %
\draw[gray,opacity=0.25] (0,0,0) -- (0,0.75,0); %
\draw[gray,opacity=0.25] (0,0,0) -- (0,0,0.75); %
\draw[ForestGreen,line width=1.5,opacity=0.25] (0,0,0) -- (0,0.75,0);;

\tdplotsetthetaplanecoords{\phivec}
\draw[dashed,Periwinkle] (0.130,-0.739,0) arc (-80:100:\radius);
\draw[dashed,Periwinkle,opacity=0.25] (-0.130,0.739,0) arc (100:280:\radius);
\shade[ball color=Periwinkle!5!white,opacity=0.1,draw=Periwinkle,line width=1.5,draw opacity=1] (0.75cm,0) arc (360:0:0.75cm and 0.75cm);
\node[label={[font=\large,text=Periwinkle]above left:{$\orbit(x)$}}] (orbit) at (0,-0.5,0.5) {};

\draw[gray] (0,0.75,0) -- (0,1,0); %
\draw[gray] (0.75,0,0) -- (1,0,0); %
\draw[gray] (0,0,0.75) -- (0,0,1); %

\node[label={[font=\large,text=ForestGreen]below:{$\bfR$}}] (orbreps) at (0,1,0) {};
\draw[ForestGreen,line width=1.5] (0,0.75,0) -- (0,1,0);
\filldraw[gray] (0,0,0) circle (1pt);
\node[label={[font=\footnotesize,text=gray]below left:{$(0,0,0)$}}] (origin) at (-0.125,-0.075,0) {};

\node[circle,draw=black,fill=BurntOrange,minimum size=4pt,label={[font=\large,text=BurntOrange]above right:{$\orbRep_x$}}] (rho) at (0,0.78,0) {};
\node[circle,draw=black,fill=Thistle,minimum size=4pt,label={[font=\large,text=Thistle]above:{$x$}}] (x) at (0.257,0.257,0.705) {};

\draw[-{Stealth[length=3mm]},dotted,OrangeRed,line width=1.5] (0,\radius,0.14) arc (110:160:1.15);
\node[label={[font=\large,text=OrangeRed]below:{$g_1$}}] (orbsel) at (0,0.55,0.45) {};
\draw[-{Stealth[length=3mm]},dotted,OrangeRed,line width=1.5] (0,0.6,-0.05) arc (-20:-70:1.15);
\node[label={[font=\large,text=OrangeRed]below:{$g_2$}}] (orbsel) at (0,0.3,0.05) {};
\end{tikzpicture}
\hfill
}
\caption{Left: a graphical depiction of a point $x\in\R^2$, its orbit $\orbit(x)=\{x'\in\R^2:\|x'\|_2=\maxInv(x)\}$, its orbit representative $\orbRep_x\in\bfR=\{(z,0): z\geq0\}$, and the role of the orbit selector $\orbSel$, maximal invariant $\maxInv$, and representative inversion $\repInv$ under the action of $\SO{2}$.
Right: a graphical depiction of a point $x\in\R^3$, its orbit $\orbit(x)$ and orbit representative $\orbRep_x\in\bfR$ under the action of $\SO{3}$. Depicted are two rotations $g_1$ and $g_2$ that take $\orbRep_x$ to $x$. The inversion kernel $\repInvKern(\argdot|x)$ is the uniform distribution over $\{g\in\SO{3}:x=g\orbRep_x\}$.
}
\label{fig:group}
\end{figure}

\subsection{Testing for distributional invariance}
\label{sec:bg:testing:invariance}

The literature on hypothesis tests for group invariance of a probability measure is extensive, going back at least to \cite{Hoeffding_1952}, and even earlier for special cases. A textbook treatment of such tests when $\grp$ is finite can be found in \cite[][Ch.~15.2]{Lehmann:Romano:2005}; although the technical details for infinite compact groups are more involved, the basic structure is similar. 
As a reference point for the methods we will develop for testing conditional symmetry, we briefly review the basics of testing for distributional invariance, highlighting particular aspects that are important in the sequel. 

A function $f$ with domain $\bfX$ is $\grp$-invariant if $f(gx) = f(x)$, all $x \in \bfX, g \in \grp$. Observe that a $\grp$-invariant function is constant on each orbit. 
This specializes to probability measures.

\begin{definition} \label{def:invariance}
    A probability measure $P$ on $\bfX$ is \emph{$\grp$-invariant} if $g_*P(A)  = P(A)$ for all $g\in\grp$, $A \in \bfS_{\bfX}$.
\end{definition}

When the group action is continuous, $\grp$-invariance of a probability measure is possible only when $\grp$ is compact (a non-compact group acting continuously would induce non-compact orbits, which would not admit invariant probability measures). 
For a specified compact group $\grp$, consider the hypotheses
\begin{align} \label{hyp:invariance}
    \nullHyp\colon P \text{ is $\grp$-invariant} \quad \text{versus} \quad \altHyp\colon P \text{ is not $\grp$-invariant.}
\end{align}
The standard level-$\sig$ test based on statistic $\testStat$ is
\begin{align*}
    \test_{\sig}(X) = \indi\{ \testStat(X) > q_{\testStat(GX)}(1-\alpha) \} \;,
\end{align*}
where $q_{\testStat(GX)}(1-\sig)$ is the quantile function of $\testStat(GX)$ at $1-\sig$, with $X \sim P$ and $G\sim \haar$, a random element of $\grp$ sampled from Haar measure (i.e., the uniform distribution on $\grp$). The usual way to show that this test has level $\sig$ uses the fact that $P$ is $\grp$-invariant if and only if $X \equdist GX$. 
Hence $\E[\test_{\sig}(X)] = \E[\test_{\sig}(GX)] \leq \sig$. 
For small and finite $\grp$, the quantile $q_{T(GX)}(1-\sig)$ can be computed exactly by applying each element of $\grp$ to $X$. However, in the case that $\grp$ is large or uncountable, randomization tests (also called ``stochastic approximation tests,'' \citealp{Lehmann:Romano:2005}) may be used to obtain a Monte Carlo estimate of the quantile. \cref{alg:group:test:new} shows the basic recipe when the elements of $X_{1:n}$ are  i.i.d.\footnote{If, instead, $X_{1:n}$ is treated as a single observation of a multivariate random variable \citep[e.g.,][]{Hemerik2018}, then randomization is performed by applying a single random transformation $G^{(b)}\sim\haar$ to the entire sample $X_{1:n}$.} 
In the basic scheme, a fixed number $B$ of random group elements are sampled independently from Haar measure. 
As with all randomization tests, in order to avoid potential issues related to increasing the Type I error rate due to the randomness of the test, $B$ should be set sufficiently large. 
When sampling is computationally expensive, there are methods for selecting $B$ adaptively in order to reduce computation while maintaining statistical validity \citep[see][for a recent review]{stoepker2024inference}, and recent work has investigated more statistically efficient sampling schemes for permutation tests \citep{Ramdas:2023,Koning:2022}.

\begin{algorithm}[bt]
\caption{Randomization test for $\grp$-invariance}\label{alg:group:test:new}
\begin{algorithmic}[1]
\State \textbf{inputs}: $X_{1:n}$, $B$, $\testStat$, $\sig$
\State compute $\testStat(X_{1:n})$
\For{$b$ in $1,\dotsc,B$}
\State sample $G_i^{(b)} \simiid \haar$
\State set $X_{1:n}^{(b)} \define (G_1^{(b)} X_1,\dotsc, G_n^{(b)} X_n)$
\State compute $\testStat(X_{1:n}^{(b)})$
\EndFor
\State compute $p$-value $p_B$ as
\begin{align*}
    p_B \define \frac{1 + \sum_{b=1}^B \indi\{\testStat(X_{1:n}^{(b)}) \geq \testStat(X_{1:n})\} }{1 + B}
\end{align*}
\State reject $\nullHyp$ if $p_B \leq \sig$
\end{algorithmic}
\end{algorithm}

Somewhat less conventionally, it can be shown that \cref{alg:group:test:new} instantiates a test that is level-$\sig$ {conditionally on  $\maxInv(X)_{1:n}$,  the element-wise maximal invariants of $X_{1:n}$, 
which is a sufficient statistic for the set of $\grp$-invariant probability measures \citep{farrell:1962,dawid:1985,BloemReddy:2020}.} 
Hence, if $X_{1:n}\simiid P$ and $P$ is $\grp$-invariant, then for any $\sig \in [0,1]$, 
\begin{align} \label{eq:mc:test:p:valid}
    \bbP\left(p_B \leq \sig \mid \maxInv(X)_{1:n}\right)  = \frac{\lfloor \sig (B+1) \rfloor}{B+1} \leq \sig \;, \quad P\text{-a.s.}
\end{align}
An orbit selector is especially useful in this setting. 
With an orbit selector $\orbSel$ fixed, it is known that every $\grp$-invariant probability measure factorizes via the distributional identity $X \equdist G\orbSel(X)$, with $G \condind X$ and $G\sim\haar$ (see \cref{apx:bg:testing}). Once the orbit is known, all remaining variation in $X$ is produced by $G\sim\haar$.  Hence, the sufficiency of $\orbSel(X)$ and the  simplicity of \cref{alg:group:test:new}. 

Conditioning the test on $\orbSel(X)_{1:n}$ has important practical implications for increasing test power. 
Because \eqref{eq:mc:test:p:valid} holds conditioned on $\orbSel(X)_{1:n}$, the statistic $\testStat$ may depend on the observed data through $\orbSel(X)_{1:n}$, meaning that functional parameters of $\testStat$ may be selected in a data-adaptive way, as long as the selection only uses $\orbSel(X)_{1:n}$ (or any other $\grp$-invariant function of $X_{1:n}$). 
For example, the power of nonparametric tests based on kernel maximum mean discrepancy is known to depend greatly on the kernel function parameters such as bandwidth \citep{Ramdas:2015}, so power may improve when parameters are selected using the observed maximal invariants. This provides a complementary justification for an idea introduced by \citet{Biggs:2024} in the context of nonparametric permutation tests, which they allowed to depend on permutation-invariant functions of the data. 
We will build on these ideas when constructing tests for conditional symmetry. 

Although \eqref{eq:mc:test:p:valid} has been proven for special cases 
\citep{Romano:1989,Hemerik2018}, the literature seems to be missing the sufficiency perspective, as well as a full proof of the level of the conditional test for a compact group that may not act freely (e.g., the group of continuous rotations for $\mathbb{R}^d$, $d \geq 3$), though \citet[][Example 2]{Romano:1989} outlines such an argument for rotations. For completeness, we provide a formal statement and proof in \cref{apx:bg:testing}.

\subsection{Symmetry of conditional distributions} \label{sec:background:symmetry}

The primary focus of this paper is on a different notion of probabilistic symmetry, pertaining to conditional distributions. Here, we review the basic concepts needed to describe the problem. 
Suppose $\grp$ acts on $\bfX$ and also on another set $\bfY$; the group action may be different on each. A function $f\colon \bfX \to \bfY$ is \emph{$\grp$-equivariant} if $f(gx) = gf(x)$, for all $x \in \bfX,\ g \in \grp$. In the case that $\grp$ acts trivially on $\bfY$, equivariance specializes to invariance.  

We say that $P_{X,Y}$ is jointly $\grp$-invariant if it is invariant in the sense of \cref{def:invariance} extended to $\grp$ acting on $\bfX \times \bfY$.
In addition to joint invariance, we may define symmetry for a conditional distribution. We write the disintegration of $P_{X,Y}$ as $P_{X,Y}=P_X \otimes P_{Y|X}$, where $P_{Y|X}$ denotes a regular conditional probability of $Y$ given $X$ (i.e., a Markov probability kernel from $\bfX$ to $\bfY$).

\begin{definition} \label{def:cond:equiv}
    The conditional distribution $P_{Y|X}$ is \emph{$\grp$-equivariant} if
    \begin{align*}
        P_{Y|X}(B\mid x) = P_{Y|X}(gB\mid gx) \;, \quad x \in \bfX, \ B \in \bfS_{\bfY}, \ g \in \grp \;.
    \end{align*}
    $P_{Y|X}$ is \emph{$\grp$-invariant} if
    \begin{align*}
    P_{Y|X}(B\mid gx) = P_{Y|X}(B\mid x) \;, \quad x \in \bfX, \ B \in \bfS_{\bfY}, \ g \in \grp \;.
    \end{align*}
\end{definition} 

We use the term ``$\grp$-symmetric'' as a catch-all for either of the conditional symmetries in \cref{def:cond:equiv}. 
As a simple example of equivariance, consider a case in which $Y \in \R^d$ is conditionally Gaussian-distributed and is acted on by $\Trans{d}$, the locally compact group of $d$-dimensional translations.

\begin{example}
\label{expl:shift:equivariance:intro}
    For random vectors $X$, $Y$ in $\R^d$, suppose that $Y | X \sim \gaussian(X, \Sigma)$, where $\Sigma$ is some fixed covariance matrix.
    It is straightforward to check that $P_{Y|X}$ is $\Trans{d}$-equivariant.
    If $\eps \sim \gaussian(0, \Sigma)$ with $\eps\condind X$ so that $Y \equdist X + \eps$, then for any $g_v \in \Trans{d}$ where $g_vx \define v+x$ for $v\in\R^d$, $A\subseteq \R^d$,
    \begin{align*}
    P_{Y|X}(g_vY \in A \mid X) = P_{Y|X}(v + X + \eps \in A \mid X) = P_{Y|X}(Y \in A \mid g_vX) \;.
    \end{align*}
    If, instead, $Y| X \sim \gaussian(\onevec_d, \Sigma)$, where $\onevec_d$ is the $d$-dimensional vector of ones, then $P_{Y|X}$ is $\Trans{d}$-invariant. In this example, conditional invariance is equivalent to independence. That is not a coincidence, but as we establish in \cref{thm:equivariance:cond:ind}, is a consequence of the transitivity of $\Trans{d}$ acting on $\bfX = \R^d$.
\end{example}

Similar to the hypotheses \eqref{hyp:invariance} in testing distributional invariance, the primary aim of this paper is to test the hypotheses
\begin{align} \label{eqn:hypotheses}
    \begin{split}
        \nullHyp \colon P_{Y|X} \text{ is $\grp$-symmetric} \qquad \text{versus}  \qquad
        \altHyp \colon P_{Y|X} \text{ is not $\grp$-symmetric.}
    \end{split}
\end{align}
In the setting of distributional invariance from \cref{sec:bg:testing:invariance}, the test properties stem from the decomposition $X \equdist G\orbSel(X)$ with $G\condind X$ and $G\sim\haar$. To test for conditional symmetry in a similar fashion, we require an analogous decomposition of the joint distribution of $(X,Y)$ in terms of a random element $\G \in \grp$, to be shared between $X$ and $Y$. 
The conditional setting presents new challenges, however.
Conditional symmetries can hold with respect to non-compact groups, as in \cref{expl:shift:equivariance:intro}. Hence, in contrast to tests for distributional invariance, tests for conditional symmetry cannot simply rely on the sampling from a normalized Haar measure, which does not exist for non-compact groups. 
Moreover, because equivariance encodes systematic regularities in the relationship between variables, the required random group element $\G$ cannot be independent from $X$ and $Y$.
In short, whereas the test for distributional invariance required only a maximal invariant and a uniformly random group element, tests for conditional symmetry must make use of additional structure arising from the group actions.

\section{Characterizations of conditional symmetry}
\label{sec:condind}

In addition to orbit selectors, the tests we will develop rely on a technique known as \emph{group inversion}.\footnote{For technical details on measurability issues that can arise, the interested reader may refer to \cref{apx:proper:action}. In particular, measurable versions of the objects defined below may not exist. Throughout, we assume that the group action is proper so that conditions needed for measurability are satisfied.} 
Recall from \cref{sec:bg:testing:invariance} that an orbit selector $\orbSel$ maps an element $x$ to its orbit representative, $\orbSel(x) = r_x$. This is well-defined even for non-compact groups, subject to the regularity condition of proper action. 
The properties that we describe below do not depend on which elements are chosen as the representatives, but the specific subsequent functions and kernels will be relative to that choice.

Recall that the {stabilizer subgroup} of $x$ is $\grp_x = \{ g \in \grp : gx = x \}$, and that the action is {free} if $\grp_x = \{\id\}$ for all $x \in \bfX$. 
A $\grp$-equivariant function $\repInv\colon\bfX\to\grp$ is called a \textit{representative inversion} if $\repInv(x)\orbSel(x) = x$ for all $x \in \bfX$. 
As the name suggests, $\repInv$ returns the element of $\grp$ that moves, or ``inverts'', the representative $\orbSel(x)$ to $x$; the inverse group element $\repInv(x)^{-1}$ moves $x$ to $\orbSel(x)$.
In order for $\repInv$ to be uniquely defined, the group action must be free.
If it is not, an equivariant \textit{inversion probability kernel} (inversion kernel for short) $\repInvKern \colon \bfX \times \bfS_{\grp} \to [0,1]$ can be used in place of $\repInv$.
If $\repInvRand|X \sim \repInvKern$, then $X \equas \repInvRand \orbSel(X)$.
At a high level, one may think of the inversion kernel $\repInvKern$ as the uniform distribution on the left coset $g \grp_{\orbSel(X)}$, where $g\orbSel(X) = X$. Assuming that the group action is proper (see \cref{apx:proper:action}) ensures that the left coset is compact and a uniform distribution is well-defined. 
When the action is free, the inversion kernel simplifies to $\repInvKern(\argdot|x) = \dirac_{\repInv(x)}(\argdot)$.
\Cref{fig:group} visualizes these structures for $\SO{2}$ and $\SO{3}$. The main difference is that the stabilizer subgroup of  $x \in \R^2$ is trivial while that of $x \in \R^3$ is not. In some cases, a representative inversion can still be (non-uniquely) defined when the action is not free (see \cref{expl:SOd:equivariance}), in which case an equivalent inversion kernel can be defined as $\repInvKern'(\argdot|x) \define \repInvKern(\repInv(x)^{-1}\argdot\mid\orbSel(x))$.

The following examples illustrate the main ideas and provide background for our experiments in \cref{sec:experiments}.

\begin{example} \label{expl:SOd:equivariance}
Let $X$ be a random vector in $\R^d$, acted on by $\text{SO}(d)$.  
A convenient maximal invariant is $\maxInv(X)=\|X\|$.
The set of orbit representatives can be chosen to be the points on the axis with unit basis vector $\e_1 = [1, 0, \dotsc, 0]^{\top}$, i.e., $\orbSel(x)=\|x\|\e_1$ for $x\in\R^d$.
For $d = 2$, the action is free; for $d > 2$, the stabilizer subgroup $\grp_{\orbSel(x)}$ is the set of $d$-dimensional rotations around the axis corresponding to $\e_1$.
One may construct a representative inversion function corresponding to $\orbSel(x) = \|x\|\e_1$ by, for example, rotating $\|x\|\e_1$ to $x$ in the 2D subspace spanned by $\nicefrac{x}{\|x\|}$ and $\e_1$.
That is, let $\tx \define \nicefrac{(x - \lip \e_1, x\rip \e_1)}{\| x - \lip \e_1, x\rip \e_1\|}$, so that $[\e_1,\ \tx]$ is a matrix in $\R^{d\times 2}$ whose columns form an orthonormal basis for the 2D subspace spanned by $\nicefrac{x}{\|x\|}$ and $\e_1$.
Let $\Theta_x$ be the 2D rotation matrix of angle $\cos^{-1}(\lip \e_1, \nicefrac{x}{\|x\|}\rip)$.
Then the $d$-dimensional rotation defined by
\begin{align} \label{eq:SOd:rep:inv}
    \repInv(x) = \I_d - \e_1 \e_1^{\top} - \tx \tx^{\top} + [\e_1,\ \tx] \Theta_x [\e_1,\ \tx]^{\top}
\end{align}
satisfies $\repInv(x) \orbSel(x) = x$.
A sample from the corresponding inversion kernel is generated by taking a uniform random $(d-1)$-dimensional rotation $H$ and extending it to a $d$-dimensional rotation $H'$ that fixes $\e_1$, so that $\repInvRand\equdist\repInv(x)H'$ has distribution $\repInvKern$.

\end{example}

\begin{example} \label{expl:lorentz}
    The Lorentz group, or indefinite orthogonal group $\text{O}(1,3)$, is a fundamental symmetry group in high-energy physics, and is central to the physical description of particles moving in laboratories at velocities approaching the speed of light \citep{bogatskiy20a}. We give only the relevant details here, largely following \citep{bogatskiy20a}; more details can be found in standard textbook references \citep[e.g.,][]{Sexl_2001,Cottingham_Greenwood_2007}. 
    $\text{O}(1,3)$ is the group of linear isometries of Minkowski spacetime, on which the theory of special relativity is constructed. 
    For a particle's \emph{four-momentum}, $x = (E, p_1, p_2, p_3)\in\R^4$, the group preserves the quadratic form
    \begin{align*}
        Q(x) = E^2 - p_1^2 - p_2^2 - p_3^2 \;,
    \end{align*}
    and so by definition, $Q$ is a maximal invariant. The subgroup that preserves the orientation of time and space (hence viewed as corresponding to physically realistic transformations) is known as the restricted Lorentz group, denoted $\lorentz$. 
    The restricted Lorentz group is a non-compact, non-commutative, connected real Lie group (a differentiable manifold) that consists of transformations $x\mapsto g_O g_\beta x$ composed of \emph{boosts} $g_\beta$ and three-dimensional rotations $g_O$ that act on the directional entries $(p_1,p_2,p_3)$.
    For $\beta\in(-1,1)$ and $O\in\textrm{SO}(3)$, $\lorentz$ transformations of a particle's four-momentum $x$ are encoded by
    \begin{align*}
        g_\beta
        \begin{bmatrix}
        E \\ p_1 \\ p_2 \\ p_3
        \end{bmatrix}
        &=
        \begin{bmatrix}
        (1-\beta^2)^{-\frac{1}{2}}(E-\beta p_1) \\ (1-\beta^2)^{-\frac{1}{2}}(-\beta E+p_1) \\ p_2 \\ p_3
        \end{bmatrix} \;, &
        g_O x &=
        \begin{bmatrix}
            1 & 0 & 0 & 0\\
            0 & O_{11} & O_{12} & O_{13} \\
            0 & O_{21} & O_{22} & O_{23} \\
            0 & O_{31} & O_{32} & O_{33}
        \end{bmatrix}x \;.
    \end{align*}
    These transformations are sometimes referred to as ``hyperbolic rotations'' that give rise to orbits whose one-dimensional projections are hyperbolas.  
    The set of orbit representatives can be chosen as $\orbSel(x) = (E_{\orbSel(x)},1,0,0)$, with $E_{\orbSel(x)}^2 - 1 = Q(x)$, so that $\bfR=\{(E,1,0,0): E\in\R_+\}$. For any $x = (E,p_1,p_2,p_3)$, a sample from the inversion kernel $P_{\repInvRand|X}(\argdot|x)$ can be obtained in two steps: first, calculate $\beta = \frac{E_{\orbSel(x)} - E \sqrt{p_1^2 + p_2^2 + p_3^2}}{1 + E^2}$; second, sample a 3D rotation $O$ uniformly at random from $\grp_{(1,0,0)}$ and compute $\repInv((p_1,p_2,p_3))$ as in \eqref{eq:SOd:rep:inv}. Then $x = g_{\repInv((p_1,p_2,p_3))} G_O g_{\beta} \orbSel(x)$.
\end{example}

\subsection{Equivariance as conditional independence}

Equivariant Markov kernels arise naturally in the disintegration of jointly invariant measures. 
Under general conditions, a measure is jointly invariant if and only if it disintegrates into an invariant marginal measure and an equivariant Markov kernel. 
A measure-theoretic proof of this fact is provided by \citet[][Theorem 7.6]{Kallenberg:2017}. 
\citet{BloemReddy:2020} connected this to conditional independence, showing that for compact $\grp$ acting on $\bfX$ such that a measurable representative inversion $\repInv$ exists, and $\grp$-invariant $P_X$, $P_{Y|X}$ is $\grp$-equivariant if and only if
\begin{align} \label{eqn:cond:ind}
    X \condind \repInv(X)^{-1} Y \mid \maxInv(X) \;.
\end{align}
Informally, the relation in \eqref{eqn:cond:ind} says that given knowledge of the orbit through a maximal invariant $\maxInv(X)$, the distribution of the transformed variable $\repInv(X)^{-1} Y$ does not depend on the particular value $X$ takes on its orbit.
\cref{fig:equivariance} illustrates this intuition for the simple example of 2D rotations: the aligned conditional distributions of $Y|X=x$ for different values of $x$ on the same orbit should be the same. If they are not, additional dependence between $Y$ and $X$ is indicated.

\begin{figure}[t]
    \centering
    \includegraphics[trim={0 0 0 0},clip,width=0.75\textwidth]{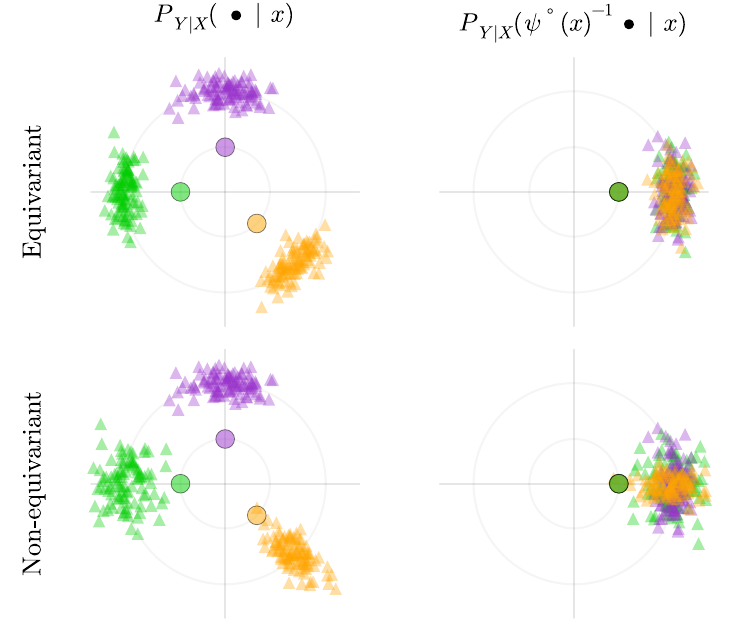}
    \caption{Three random samples of $Y\in\R^2$ (triangle) drawn conditionally on $X\in\R^2$ (circle). If the conditional distribution $P_{Y|X}$ is $\SO{2}$-equivariant, then $P_{Y|X}(\repInv(x)^{-1}\argdot|x)$ is the same for all samples.}
    \label{fig:equivariance}
\end{figure}

The assumptions for \eqref{eqn:cond:ind} in \citep{BloemReddy:2020} 
are somewhat restrictive, as they require $\grp$ to be compact {and} $P_X$ to be invariant. In many situations, one or both of those assumptions will not hold. 
The following theorem substantially generalizes \eqref{eqn:cond:ind}, making it applicable to non-compact $\grp$ and non-invariant $P_X$. Moreover, it does not require the existence of a representative inversion $\repInv(x)$, so that the action of $\grp$ on $\bfX$ may not be free.

\begin{theorem} \label{thm:equivariance:cond:ind}
    Let $\grp$ be a lcscH group acting on each of $\bfX$ and $\bfY$, the action on $\bfX$ proper, with measurable inversion kernel $\repInvKern$.
    The conditional distribution $P_{Y|X}$ is $\grp$-equivariant if and only if,
    \begin{align} \label{eqn:cond:ind:rand}
        (\repInvRand, X) \condind \repInvRand^{-1} Y \mid \maxInv(X) \;, \quad \text{with} \quad \repInvRand\mid X \sim \repInvKern \;.
    \end{align}
\end{theorem}

The proof can be found in \cref{apx:proofs:thm:4:1}. We note the following special cases. If there exists a representative inversion function $\repInv \colon \bfX \to \grp$, then condition \eqref{eqn:cond:ind:rand} reduces to \eqref{eqn:cond:ind}.
If the action of $\grp$ on $\bfY$ is trivial, then \eqref{eqn:cond:ind:rand} reduces to $X \condind Y \mid \maxInv(X)$. 
Finally, if the action of $\grp$ on $\bfX$ is transitive, as in \cref{expl:shift:equivariance:intro}, then \eqref{eqn:cond:ind:rand} specializes to the (unconditional) independence.

\Cref{thm:equivariance:cond:ind} implies that a test for the hypotheses \eqref{eqn:hypotheses} 
can be performed using a test for the conditional independence in \eqref{eqn:cond:ind:rand}. 
If the action of $\grp$ on $\bfX$ is transitive, then $\bfX$ has a single orbit and the test simplifies to a test for (unconditional) independence. 
In principle, any general-purpose test for conditional independence may be used. 
However, it is known that any test of level $\sig$ over the entire null hypothesis class has no power against any alternative \citep{Shah_2020}. 
The high-level solution is to restrict the null hypothesis class, which is commonly accomplished through problem-specific assumptions \citep{Shah_2020, Kim:2022} or modeling some parts of the relevant conditional distributions \citep{Zhang:2011, Pogodin:2024}.

\subsection{Distributional identities for conditional randomization}
\label{sec:condind:distid}

One broadly useful set of tools for testing conditional independence are conditional randomization tests \citep{Candes:2018,Berrett:2019,Liu:2022}, in which a suitable randomization of the data is used to generate a test that is similar in spirit to \cref{alg:group:test:new}. 
This motivates a refinement of the characterization of conditional symmetry in \cref{thm:equivariance:cond:ind}, on which we will build a conditional randomization test in the next section. The basic idea is that for a pair $(X,Y)$, the ``centered'' variables $\orbSel(X)$ and $\tY\define \repInvRand^{-1}Y$, 
with $\repInvRand|X\sim\repInvKern$,
can be used to obtain a randomized pair $(X',Y') \equdist (X,Y)$. The randomization involves sampling one or both of
\begin{align} \label{eq:equiv:randomization}
    \tY\mid \orbSel(X) \sim \tYKern\;, \quad \G\mid\orbSel(X) \sim \GKern \;, 
\end{align}
with $\tY \condind \G \mid \orbSel(X)$. 
Here, $\G$ denotes the random variable with conditional distribution
\begin{align} \label{eq:GKern:def}
    \GKern(\argdot \mid \orbRep) = \int_\bfX \repInvKern(\argdot\mid x')\XorbRepKern(dx'\mid \orbRep) \;,
\end{align}
where $\XorbRepKern$ is the conditional distribution of $X$ given an orbit, represented by the random representative $R$ (or more generally, any maximal invariant $M$). 
We refer to $\GKern$ as the \emph{group-orbit kernel} as it is the distribution on $\grp$ induced by restricting $P_X$ to the orbit of $R$. 
The group-orbit kernel has the property that if $\G|\orbSel(X) \sim\GKern$, then $(\orbSel(X),X) \equdist (\orbSel(X), \G\orbSel(X))$. (See \cref{lem:GKern} for a proof.)

In \cref{apx:proofs:equivariance}, we prove a number of distributional identities involving these randomizations and that underpin the tests we construct in \cref{sec:rand}. 
In particular, with $\G$ sampled as in \eqref{eq:equiv:randomization}, $P_{Y|X}$ is $\grp$-equivariant if and only if
\begin{align} \label{eq:equiv:iden:main}
    (\orbSel(X),\tY,X,Y) \equdist (\orbSel(X),\tY, \G\orbSel(X),\G\tY) \;.
\end{align}
This identity says that if we are able to sample $\G$ from the correct distribution conditional on $\orbSel(X)$, then we can test for $\grp$-equivariance by testing for equality in distribution between the data and the randomization $(X', Y') \coloneq (\G \orbSel(X), \G \tY)$. 
Crucially, the randomization (and therefore the test) can be conditional on part of the observed data: $\orbSel(X)$ and $\tY$. 
Conditioning restricts the null hypothesis set to agree with that part of the data, and therefore sidesteps the result of \cite{Shah_2020}. 
Moreover, the identity \eqref{eq:equiv:iden:main} specializes to 
\begin{align} \label{eq:equiv:iden:marginal}
    (\orbSel(X),\tY,Y) \equdist (\orbSel(X),\tY, \G\tY) \;,
\end{align}
which implies that as long as $\grp$ acts non-trivially on $\bfY$, the distributional comparison can be carried out using only observed $Y$ and randomized $Y' = \G \tY$. 

In \cref{sec:rand}, we describe how valid tests for conditional symmetry can be implemented for each of these cases.

\begin{example} \label{expl:SOd:equivariance:cont}
    Consider random vectors $X$ and $Y$ in $\R^d$, acted on by $\SO{d}$. As described in \cref{expl:SOd:equivariance}, a convenient maximal invariant is $\maxInv(X)=\|X\|$, with $\orbSel(X)=\|X\|\e_1$. Then \eqref{eq:equiv:iden:marginal} says that $P_{Y|X}$ is $\SO{d}$-equivariant if and only if $Y \equdist \G\tY$, 
    where $\G|\orbSel(X)\sim \GKern$ is a random $d$-dimensional rotation matrix that satisfies $\G \orbSel(X) \equdist X$. 
\end{example}

\begin{example}[Continuation of \cref{expl:lorentz}] \label{expl:lorentz:cont}
    The Standard Model of physics is invariant under the Lorentz group, and violations of $\lorentz$-invariance could indicate phenomena beyond the Standard Model. (More modestly, violations of $\lorentz$-invariance could indicate issues with high-energy physics simulators.) As an example, consider the four-momenta of particles produced by high-energy colliders such as the Large Hadron Collider (LHC). 
    High-energy quarks produced in collisions such as those at the LHC quickly decay through a cascading process of gluon emission into collimated sprays of stable hadrons, which are subatomic particles that can be detected and measured \citep{Salam2010,bogatskiy20a}. This spray is known as a \textit{jet}. According to the Standard Model, the four-momenta of the constituent particles should respect the symmetries of the Lorentz group in the following sense. Let $X$ and $Y$ be the four-momenta relative to the lab frame of any two particles produced by the same collision; hence, they are stochastically dependent. The conditional distribution $P_{Y|X}$ should be $\lorentz$-equivariant, as a Lorentz transformation of $X$ is equivalent to a change of the lab reference frame, which would then also apply to $Y$. Observe that this is not distributional invariance, because $gX \nequdist X$.
    \cref{expl:lorentz} describes how to sample from the inversion kernel, so if samples can be generated from $X'|\orbSel(X)\sim P_{X|R}$ (with $\orbSel$ as described in \cref{expl:lorentz}), then $\G = \repInvRand'$ with $\repInvRand'|X'\sim \repInvKern$ and $\tY = \repInvRand^{-1} Y$ with $\repInvRand|X \sim \repInvKern$ satisfy the distributional identities above. 
\end{example}

\section{\label{sec:rand}Randomization tests for conditional group symmetry}

The distributional identities in \cref{sec:condind:distid} provide the foundation for tests of conditional symmetry, in the form of  conditional randomization tests (CRTs). 
The basic idea is as follows. With a set of observations $(X,Y)_{1:n}$, generate $B$ randomization sets $(X,Y)_{1:n}^{(b)}$ as
\begin{align*}
    (X,Y)_i^{(b)} \coloneq (\GRand_i^{(b)}\orbSel(X_i),\GRand_i^{(b)}\tYRand_i^{(b)}) \;, \quad b = 1,\dotsc,B,\ \  i = 1,\dotsc, n .
\end{align*}
The distributional identities in 
\cref{thm:equiv:dist:id} indicate that the randomization can be implemented in three distinct ways by conditioning on observed data differently:
\begin{tests}[label=R\arabic*.,ref=R\arabic*]
    \item \label{test:g}
    Condition on $\orbSel(X)_{1:n}$ and $\tY_{1:n}$, with $\tY_i = \repInvRand^{-1}_i Y_i$, and where $\repInvRand_i$ is a representative inversion if one exists or is otherwise sampled from the inversion kernel and fixed. Set $\tYRand_i^{(b)} = \tY_i$ and sample $\GRand_i^{(b)}|\orbSel(X_i) \sim \GKern$. 
    
    \item \label{test:ty}
    Condition on $\orbSel(X)_{1:n}$ and $\repInvRand_{1:n}$, where the latter is a representative inversion if one exists or is otherwise sampled from the inversion kernel and fixed. Set $\GRand_i^{(b)} = \repInvRand_i$ and sample $\tYRand_i^{(b)}|\orbSel(X_i) \sim \tYKern$.

    \item \label{test:g:ty}
    Condition on $\orbSel(X)_{1:n}$. Sample both $\GRand_i^{(b)}|\orbSel(X_i) \sim \GKern$ and $\tYRand_i^{(b)}|\orbSel(X_i) \sim \tYKern$.
\end{tests}
In \cref{apx:rand:inv:kern}, we show that in settings where a representative inversion $\repInv$ can be defined despite the group action not being free, randomizing to $\repInvRand$ is not necessary in order for the tests to have the appropriate level. 

Each of the randomization procedures above yields a different test; we differentiate between them by indicating what statistic of $(X,Y)_{1:n}$ is conditioned on, denoted $S_{1:n} \coloneq S(X,Y)_{1:n}$. For example, in \cref{test:g}, $S_i = (\orbSel(X_i), \tY_i)$. Carrying out the appropriate randomization for each $i = 1,\dotsc,n$ and repeating $B$ times, the CRT with test statistic $T$ estimates a conditional $p$-value as
\begin{align} \label{eqn:crt:pval}
    p_B \define \frac{1+\sum_{b=1}^B\indi\left\{\testStat\left((X,Y)_{1:n}^{(b)}\right)\geq \testStat\left((X,Y)_{1:n}\right)\right\}}{B+1} \;.
\end{align}
The decision to reject the null hypothesis is taken if $p_B$ is less than or equal to a specified significance $\sig\in [0,1]$. The test statistic $T$ should be some signal of distributional equality, with small values constituting evidence in favor of the null hypothesis.
For example, we implement our tests in the subsequent sections using kernel maximum mean discrepancy (MMD). 

It is straightforward to show via \cref{thm:equiv:dist:id} and its corollaries that the resulting CRT based on $B$ randomization sets is a  level-$\sig$ test. 
Moreover, as $B\to\infty$, the CRT produces a consistent conditional $p$-value, in the following sense. Under the null hypothesis that $P_{Y|X}$ is $\grp$-equivariant, if $(X',Y') \equdist (X,Y)$, then for $P$-almost every sequence $s_{1:n}$, as $B\to\infty$,
\begin{align} \label{eqn:crt:consistency}
    p_B\; \toas\; \bbP\left\{ \testStat\left((X',Y')_{1:n}\right)) \geq \testStat\left((X,Y)_{1:n}\right) \mid S'_{1:n} = S_{1:n} = s_{1:n} \right\} \;.
\end{align}
\cref{prp:crt:validity} formalizes the result generally in terms of a test based on the identity \eqref{eq:equiv:iden:main} (from \cref{thm:equiv:dist:id}); the same holds for a test based on the identity \eqref{eq:equiv:iden:marginal} (from \cref{cor:eq:xy:cor}). 
The proof is straightforward and provided in \cref{apx:proofs:crt:validity} for completeness. As with the test of distributional invariance in \cref{sec:bg:testing:invariance}, conditioning on $S_{1:n}$ means that $S_{1:n}$ can be used as input to the test statistic. 

\begin{proposition} \label{prp:crt:validity}
    Assume that $P_{Y|X}$ is $\grp$-equivariant. For data $(X,Y)_{1:n}$, generate $B$ randomized versions $(X,Y)^{(b)}_{1:n}$ using any of the randomization procedures \cref{test:g}--\cref{test:g:ty}. Assume that $\bbP\{\testStat((X,Y)_{1:n}^{(b)}) = \testStat((X,Y)_{1:n}^{(b')})\} = 0$ for $b \neq b'$.   
    For arbitrary fixed $B\in\mathbb{N}$ and $\sig\in[0,1]$, the statistic $p_B$ defined in \eqref{eqn:crt:pval} satisfies
    \begin{align} \label{eq:crt:pval:level}
        \bbP\left(p_B\leq\sig \mid S_{1:n} \right) = \frac{\lfloor \sig (B+1)\rfloor}{B+1} \leq \sig \;, \quad  P\text{-a.s.}
    \end{align}
    Moreover, the conditional strong consistency of $p_B$, as in  \eqref{eqn:crt:consistency}, holds. 
\end{proposition}

As mentioned in \cref{sec:bg:testing:invariance}, one sets $B$ before conducting the test.
However, there are approaches that allow $B$ to be chosen adaptively, which may have benefits in terms of improving computational efficiency while retaining test validity; see \citep{stoepker2024inference} for a recent survey.
Our tests may also benefit from these approaches, though because simulating samples from the relevant conditional distributions is not overly expensive, we do not pursue that here.

\subsection{\label{sec:rand:approx}Approximate conditional randomization}

In \cref{apx:rand:cond:exact}, we describe certain settings that allow for exact sampling from $\GKern$. One particular case of interest described there is when $\grp$ acts transitively on $\bfX$, in which case sampling from $\GKern$ amounts to permuting the $\repInvRand_i$'s. If $\grp$ induces a finite number of orbits in $\bfX$, then an orbit-stratified permutation can be used. 
In general, we expect that sampling from approximations of $\GKern$ and/or $\tYKern$ will be the more common approach in practice.

To that end, let $(X,Y)_{i,m}^{(b)}$ be a randomization of $(X,Y)_i$ generated as in \cref{prp:crt:validity}, except that $\GRand_i^{(b)}$ and/or $\tYRand_i^{(b)}$ is sampled from an approximation to the appropriate conditional distribution. For example, this might be a sample from a kernel density estimate constructed from an independent dataset of $m$ observations. Under the null hypothesis, if $T$ is a continuous function and if as $m\to\infty$, $(X,Y)_{i,m}^{(b)} \todist (X,Y)$ conditioned on $\orbSel(X_i)$, then the approximate CRT will be asymptotically level-$\sig$. 
We state and prove this formally in \cref{apx:proofs:crt:validity}.
We note that the sequence of approximations does not necessarily require an independent dataset to estimate the conditional distributions, but can depend on $S_{1:n}$. 
Hence, any asymptotically correct method for approximating the required conditional distributions can be used for the conditional randomization procedure to yield an asymptotically level-$\sig$ test.

Here, we describe a kernel estimation approach for sampling from approximations of the conditional distributions required for \cref{test:g,test:ty,test:g:ty}.
Let $\maxInv(X)$ be a continuous maximal invariant, such as $\orbSel(X)$ chosen so that $\orbSel$ is a continuous function, and let $\kde\colon\bfM\times\bfM\to\R_+$ be a kernel function defined on $\bfM$. Assume that $\kde$ is bounded above by $\nu<\infty$, with $\kde(z,z)=\nu$ for all $z\in\bfM$. Using a (possibly independent) dataset of size $m$ for estimation, assume that the bandwidth parameter $\bandwidth$ decreases to $0$ as $m\to\infty$, slowly enough to ensure asymptotically consistent density estimation. To generate a randomization conditioned on $M(X)$, we first assign to each index $i\in\{1,\dotsc,m\}$ a probability proportional to $\kde(\maxInv(X),\maxInv(X_i))$.
To sample from an approximation of $\tYKern$ for \cref{test:ty,test:g:ty}, we then sample an index $\ell$ from $\{1,\dotsc,m\}$ with those probabilities, and then take $\tYRand_i=\tY_\ell$.
Similarly, to sample from an approximation of $\GKern$ for \cref{test:g,test:g:ty}, we sample an index $j$ from $\{1,\dotsc,m\}$ (independent of $\ell$) with those probabilities, and then take $\GRand_i=\repInvRand$, where $\repInvRand|X_j\sim\repInvKern$.
It can be seen that as $m\to\infty$, the distributions from which this procedure samples approach the target conditional distributions. 
The resulting conditional randomization procedure (without an independent dataset) used in a test based on \cref{thm:equiv:dist:id} is summarized in \cref{alg:cond:rand}. For a test based on \cref{cor:eq:xy:cor}, only the second component of the output pair is retained.

\begin{algorithm}[!tb]
\caption{Approximate conditional randomization procedure}\label{alg:cond:rand}
\begin{algorithmic}[1]
\State \textbf{inputs} $(X,Y)$, $(X,Y)_{1:n}$
\State compute $p_i \coloneq \kde(\maxInv(X),\maxInv(X_i))$ for $i\in\{1,\dotsc,n\}$
\State independently sample $j,\ell\in\{1,\dotsc,n\}$ with probability proportional to $p_1,\dotsc,p_n$
\If{randomize $\G$}
\State sample $\repInvRand|X_j\sim\repInvKern$
\Else
\State sample $\repInvRand|X\sim\repInvKern$
\EndIf
\State set $\GRand \coloneq \repInvRand$
\If{randomize $\tY$}
\State set $\tYRand \coloneq \tY_\ell$
\Else
\State set $\tYRand \coloneq \tY$
\EndIf
\State return $(\GRand\orbSel(X),\GRand\tYRand)$
\end{algorithmic}
\end{algorithm}

In our experiments, we estimate the conditional distributions using the same data used for testing. This procedure remains valid because we base the tests on \cref{cor:eq:xy:cor}, using only $\tY_{1:n}$; the part of the observed data used for estimating conditional distributions, $(\repInvRand_{1:n},M(X)_{1:n})$, is conditionally independent of those variables under the null hypothesis.  With a continuous maximal invariant, we find that it is sufficient to use a Gaussian kernel with a diagonal covariance matrix given by Silverman's rule of thumb \citep[][Ch.~4]{Silverman:1986}.
Our experiments show that with finite samples, using Silverman's rule for approximate randomization tends to yield conservative tests with sizes less than $\sig$.
Such tests are likely to only achieve exact level-$\sig$ as $m\to\infty$ or if exact randomization is possible in settings like those described in \cref{apx:rand:cond:exact}. Using larger bandwidths can increase power, although it can also lead to inflated Type I error rates (see \cref{apx:exp:bw}).

\subsection{\label{sec:rand:perm}Baseline: two-sample permutation tests}

As a simple alternative to the conditional randomization tests described in the previous sections, consider generating $B=1$ randomized sample, then comparing it to the observed data using a two-sample test for equality in distribution. For example, a permutation test for equality in distribution provides a nonparametric baseline. Although these are valid level-$\sig$ tests, they do not take full advantage of the structure induced by the group symmetry and we do not expect them to be as powerful against alternatives as the CRT. We implement such tests for baseline comparison in our experiments.

\section{\label{sec:kernel}Kernel-based hypothesis tests for conditional group symmetry}

In the experiments reported in \cref{sec:experiments}, we test for conditional symmetry using a kernel-based test statistic of observations denoted for convenience by $Z_{1:n}$. For tests based on identity \eqref{eq:equiv:iden:main} (\cref{thm:equiv:dist:id}), $Z_i = (X,Y)_i$, with a kernel defined on the product space $\bfZ = \bfX \times \bfY$. For tests based on \eqref{eq:equiv:iden:marginal} (\cref{cor:eq:xy:cor}), $Z_i = Y_i$, with a kernel defined on $\bfZ = \bfY$.  We provide a brief overview of kernel-based hypothesis testing and other related ideas in this section.
See \cite{steinwart:2008:svm} for a thorough treatment of the theory of kernels, and \cite{Muandet:2017} for a review of their uses in statistics and machine learning. 

Let $\k\colon\bfZ\times\bfZ\rightarrow\R$ be a positive semidefinite kernel with parameter $\kparam$, corresponding to the Reproducing Kernel Hilbert Space (RKHS) $\hilbert$. If the kernel is measurable and bounded, then the \textit{kernel mean embedding} (KME) of a probability measure $P$ on $\bfZ$ is the unique element $\kme{P}(\argdot)\define\int_\bfZ \k(z,\argdot)P(dz)\in\hilbert$ such that $\E_P[f(Z)]=\lip f,\kme{P}\rip_\hilbert$ for all $f\in\hilbert$ \citep{Gretton:2012}.
A kernel is said to be \textit{characteristic} if the map $P\mapsto\kme{P}$ from $\mathcal{P}(\bfZ)$ into $\hilbert$ is injective, meaning that the embedding for each probability measure is unique \citep{Sriperumbudur:2010}. 
KME-based hypothesis tests compare distributions through their KMEs \citep{Gretton:2007,Gretton:2012}. 
The test statistic in KME-based tests is an estimator of the squared \textit{maximum mean discrepancy} (MMD),
$\MMD^2(P_1,P_2) \define \left\|\kme{P_1} - \kme{P_2}\right\|_\hilbert^2$. 
If the kernel $\k$ is characteristic, $\MMD^2$ is a metric on $\mathcal{P}(\bfZ)$; hence $\MMD(P,Q) = 0$ if and only if $P=Q$. 
If a non-characteristic kernel is used instead, $\MMD^2$ may be unable to separate some distinct elements of $\mathcal{P}(\bfZ)$, which would not affect the level of the test, but may reduce the power under some alternatives. 

\begin{algorithm}[bt]
\caption{Test for conditional symmetry using MMD-based statistics}\label{alg:fuse}
\begin{algorithmic}[1]
\State \textbf{inputs}: $(X,Y)_{1:n}$, $B$, $\sig$, \texttt{reuse} option, statistic $\testStat \in \{ \widehat{\mathrm{MMD}}^2_{\kparam}, \widehat{\mathrm{FUSE}}_\omega, \widehat{\textrm{SK}}\}$, $\{\sigma_1,\dotsc,\sigma_L\}$
\State resample $Z^{(0)}_{1:n}$ using \cref{alg:cond:rand}
\State compute $\testStat^{(0)} \coloneq\testStat(Z_{1:n},Z_{1:n}^{(0)})$
\If{\texttt{reuse}}
\State resample $Z'_{1:n}$ using \cref{alg:cond:rand}
\EndIf
\For{$b$ in $1\dotsc,B$}
\State resample $Z^{(b)}_{1:n}$ using \cref{alg:cond:rand}
\If{\texttt{!reuse}}
\State resample $Z'_{1:n}$ using \cref{alg:cond:rand}
\EndIf
\State compute $\testStat^{(b)}\coloneq\testStat(Z^{(b)}_{1:n},Z'_{1:n})$
\EndFor
\State compute $p_B$ as 
\begin{equation*} %
    p_B \define \frac{1 + \sum_{b=1}^B \indi\{\testStat^{(b)} \geq \testStat^{(0)}\} }{1 + B}
\end{equation*}
\State reject $\nullHyp$ if $p_B \leq \sig$
\end{algorithmic}
\end{algorithm}

The MMD can be estimated from independent samples $Z_{1:n}\simiid P_1$ and $Z_{1:n}'\simiid P_2$ via the unbiased U-statistic
\begin{align*}
\widehat{\mathrm{MMD}}_{\kparam}^2(Z_{1:{n}},Z_{1:{n}}') = \frac{1}{n(n-1)}\sum_{i\neq j}\k(Z_i,Z_j) + \k(Z_i',Z_j') - \frac{2}{n^2}\sum_{i,j}\k(Z_i,Z_j') \;,
\end{align*}
or the biased V-statistic
\begin{align*}
    \widehat{\mathrm{MMD}}_{v,\kparam}^2(Z_{1:{n}},Z_{1:{n}}') = \frac{1}{n^2} \sum_{i, j}\k(Z_i,Z_j) + \k(Z_i',Z_j') - 2\k(Z_i,Z_j') \;.
\end{align*}
Our CRT implementation, shown in \cref{alg:fuse}, is able to use either of these statistics or the adaptive MMD-based statistics described in the next section. The basic idea is to simulate conditionally independent randomizations $Z_{1:n}^{(b)}$ and ${Z'}_{1:n}^{(b)}$, then use them to estimate the MMD. This produces an estimate of the test statistic under the null distribution with which to compare the estimated MMD between $Z_{1:n}$ and yet another independent randomization ${Z}^{(0)}_{1:n}$. Optionally, a single set of comparison samples $Z'_{1:n}$ can be reused (option \texttt{reuse} in \cref{alg:fuse}) instead of independent comparison samples ${Z'}_{1:n}^{(b)}$, $b = 0,\dotsc,B$. The level of the test in each case is $\sig$, though we expect different power properties.

In the theoretical analysis of the MMD-based tests, not all of the usual statistical properties of MMD-based tests hold. Whereas in the standard setting, independent samples are compared, here the samples are dependent through conditioning on $S_{1:n}$ for the randomization. The two samples are, however, conditionally independent, which allows us to build on techniques for analyzing MMD estimators when independent samples are used \citep{Sriperumbudur_2012}. 
For ease of theoretical analysis of the MMD-based test, we assume that independent comparison sets (\texttt{reuse = false}) and the V-statistic are used.
We use a single comparison set (\texttt{reuse = true}) and the U-statistic in our experiments, but we do not expect the tests to behave very differently. Simulation experiments comparing reused and independent comparison sets (see \cref{apx:exp:mvn}) support this conclusion.

Conditioned on $S_{1:n}$, the expected value of the V-statistic is
\begin{align*}
    \Delta_n^2 &\coloneq \E\left[ \widehat{\mathrm{MMD}}_{v,\kparam}^2(Z_{1:{n}},{Z}_{1:{n}}^{(0)}) \;\bigg|\; S_{1:n} \right]  \\
    &= \frac{1}{n^2} \sum_{i,j} \E\left[k_{\kparam}(Z_i,Z_j) + k_{\kparam}(Z^{(0)}_i,Z^{(0)}_j) - 2 k_{\kparam}(Z_i,Z^{(0)}_j)\mid S_{1:n}\right] \nonumber \;.
\end{align*}
If $\Delta_n$ is sufficiently large, we expect the test to have power against non-equivariant alternatives. To state the formal result, let $\supk^2 = \sup_z k_{\kparam}(z,z) < \infty$, and define for $\eta > 0$,
\begin{align} \label{eq:Rn:def}
C_n(\eta) \define \supk\left(6\sqrt{\frac{2}{n}} + 3\sqrt{\eta}\sqrt{\frac{1}{n}+8\sqrt{\frac{2}{n^3}}} + \frac{4\eta}{n}\right) + \supk^2\left(\frac{9}{n} + 72\sqrt{\frac{2}{n^3}}\right) \;.
\end{align}
Because we analyze the test that conditions on $S_{1:n}$, probabilistic properties of the test can be said to hold conditionally (on $S_{1:n}$), almost surely (with respect to $S_{1:n}$). Analysis of bootstrap procedures relies on the same ideas \citep[e.g.,][Ch.\ 23]{vaart_1998}. 

\begin{theorem} \label{thm:mmd:power}
    Let $T$ be $\widehat{\mathrm{MMD}}_{v,\kparam}$, with kernel $k_{\kparam}$ and corresponding RKHS $\hilbert_{k_{\kparam}}$, using exact randomization in the CRT.
    For any $\eta>0$, if
    \begin{align} \label{eq:emmd:condition}
    \EMMD > 2\supk\left(2\sqrt{\frac{2}{n}} + \sqrt{\eta}\sqrt{\frac{1}{n}+8\sqrt{\frac{2}{n^3}}} + \frac{4\eta}{3n}\right) \;,
    \end{align}
    then for almost every sequence $S_{1:n}$, conditionally on $S_{1:n}$, with probability at least $1-e^{-\eta}$ over observations $Z_{1:n}$ and randomization ${Z}^{(0)}_{1:n}$,  
    \begin{align} \label{eq:power:cdf:lower:bound}
        \bbP(p_B \leq \sig \mid S_{1:n}) \geq F_{B, p_{\eta}(\Delta_n)}(\lfloor \sig (B+1) - 1\rfloor) \;,
    \end{align}
    where $F_{B, p(\Delta_n)}$ is the cumulative distribution function for the binomial distribution of $B$ trials with success probability 
    \begin{align} \label{eq:success:prob}
        p_\eta(\EMMD) = \min\left\{1, \exp\left(-\frac{3}{2}\EMMD + C_n(\eta)\right) \right\} \;.
    \end{align}
\end{theorem}

The proof can be found in \cref{apx:power:proofs}. 
Observe that for fixed $\alpha$, $B$, and $n$, the lower bound in \eqref{eq:power:cdf:lower:bound} is increasing in $\Delta_n$.  
Because multiple approximations were used for the sake of deriving explicit bounds, a tighter lower bound seems possible for \eqref{eq:success:prob}, but we do not pursue that here. 
A lower bound on values of $\Delta_n$ that correspond to alternatives for which the test with high probability rejects the null with conditional probability at least $1-\beta$ can be found in \cref{apx:lower:bound:alternative}.

\subsection{Adaptive kernel test statistics}

It is known that the power of kernel MMD-based tests is highly dependent on the chosen kernel parameters, such as the bandwidth $\kparam$ \citep{Ramdas:2015}.
Several solutions have been proposed, in the form of data-based heuristics \citep{Gretton:2012,Garreau:2017} and, more recently, adaptive methods \citep{Schrab:2023,Biggs:2024,Carcamo:2024}. 
In our experiments, when using a kernel with a parameter (e.g., the Gaussian kernel), we use two test statistics based on solutions proposed in the {adaptive kernel} literature.
The first adapts the unnormalized {FUSE} statistic
\begin{align*}
\widehat{\mathrm{FUSE}}_\omega(Z_{1:n},Z_{1:n}') = \frac{1}{\omega}\log\left(\frac{1}{L}\sum_{\ell=1}^L\exp\left(\omega\widehat{\text{MMD}}^2_{\kparam_\ell}(Z_{1:n},Z_{1:n}')\right)\right) \;,
\end{align*}
where, following \citet{Biggs:2024}, $\omega=\sqrt{n(n-1)}$. The statistic can be interpreted as taking a softmax across estimated MMD$^2$ values computed using different kernels $k_{\kparam_\ell}$.
The bandwidth values $\{\kparam_1,\dotsc,\kparam_L\}$ in the original work \citep{Biggs:2024} are quantiles of the pairwise Euclidean distances in the data. In the context of a permutation test, the empirical measure of the data is a maximal invariant, so any function of it (such as quantiles) can be used in the context of a test for permutation-invariance, as in \cref{sec:bg:testing:invariance}. 
Our adaptation to CRTs for conditional group symmetry, which satisfies the requirements of \cref{prp:crt:validity}, is obtained by replacing the quantiles of the observation with quantiles of variables contained in the conditioning set $S_{1:n}$, such as the pairwise Euclidean distances in $\tY_1,\dotsc,\tY_n$.

The second test statistic we use is an adaptation of the \emph{supremum kernel distance} (SK, \citealt{Carcamo:2024}) given by
\begin{align*}
\widehat{\mathrm{SK}}(Z_{1:n}, Z_{1:n}') = \sup_{\ell\in\{1,\dotsc,L\}}\widehat{\text{MMD}}^2_{\kparam_\ell}(Z_{1:n}, Z_{1:n}') \;.
\end{align*}
Similar to the FUSE statistic above, $\kparam_{\ell}$ can depend on $S_{1:n}$, and using this procedure remains valid due to the conditioning argument in \cref{prp:crt:validity}. 

Analogous versions of the power lower bound in \cref{thm:mmd:power} hold for the FUSE and SK statistics. For brevity, their statements and proofs can be found in \cref{apx:power:proofs}.

\section{\label{sec:experiments}Experiments}

We evaluate our proposed tests on synthetic examples and on realistic state-of-the-art simulated data used in particle physics. 
In all of our experiments, we sample $n$ data points from a distribution or dataset and perform a test for conditional symmetry with respect to a specified group.
We repeat this procedure over $N=1000$ simulations for each test setting and report the proportion of simulations in which the test rejected, which is an estimate of the test size (resp.\ power) when the distribution has (resp.\ does not have) the specified conditional symmetry.
With $N=1000$ simulations, estimates are precise up to approximately $\pm0.016$.
We set the test level to be $\sig=0.05$ in all of our experiments.
Except where specified otherwise, we use a single randomized comparison set (\texttt{reuse = true}) for estimating the $p$-value in our tests.

We default to aggregating Gaussian and Laplace kernels via the FUSE and SK test statistics for continuous data. 
To compute the bandwidths $\{\kparam_1,\dotsc,\kparam_L\}$, we follow a procedure similar to the one described by \citet{Biggs:2024}.
We first compute the distances $\{\|\tY_i-\tY_j'\|_2\}_{i,j}$ and $\{\|\tY_i-\tY_j'\|_1\}_{i,j}$ for the Gaussian and Laplace kernels, respectively.
We then compute half of the 5\% quantile and two times the 95\% quantile for each set, and take the $L=10$ bandwidth values to be the discrete uniform interpolation of those two values. All estimates of MMD use the U-statistic.

All experiments were implemented in Julia (version 1.6.1) and run on a high-performance computing allocation with 4 CPU cores and 16 GB of RAM.
Additional experiments and experimental details can be found in \cref{apx:exp}, including an experiment showing how our tests for conditional symmetry can be used for checking learned symmetry in predictive models (\cref{sec:experiments:MNIST}).
Associated computer code can be found on GitHub.\footnote{\url{https://github.com/chiukenny/Tests-for-Conditional-Symmetry}}

\subsection{\label{sec:exp:tests}Experimental configurations}

In our experiments involving tests for conditional symmetry, we consider up to four tests, each of which corresponds to a particular configuration of the following parameters:
\begin{enumerate}
    \item Permutation test as a baseline (\nBS; \cref{sec:rand:perm}) versus conditional randomization test (\nRN; \cref{sec:rand}).
    \item FUSE test statistic (\nFS) versus supremum kernel test statistic (\nSK).
\end{enumerate}
For example, the permutation test that uses the FUSE test statistic is denoted \nBSFS. 
We also consider a similar baseline when conducting a test for joint invariance as a test for conditional symmetry, where the samples used to estimate the $p$-value are generated via permutations rather than group actions from $\haar$.
In this setting, we evaluate two tests: the permutation test described here (\nBS) and the group-based randomization test (\nRN; \cref{sec:bg:testing:invariance}).

\subsection{Synthetic experiments}
\label{sec:exp:mvn}

We start with several synthetic experiments involving simple distributions for the marginal of $X$ and the conditional $Y|X$.
These experiments will be used to compare and demonstrate the properties of the different tests for conditional symmetry.
In all experiments described in this section, we use $B=100$ randomizations for estimating the $p$-value, chosen for computational convenience due to the large number of experimental configurations evaluated. 
In practice, we recommend using a larger number of randomizations to minimize the risks associated with performing random inference (see \cref{apx:exp:risk}).

\subsubsection{\label{sec:exp:gauss}Gaussian data and rotation-equivariance}

\begin{figure}[tbp]
    \centering
    \includegraphics[trim={29 26 0 0},clip,width=\textwidth]{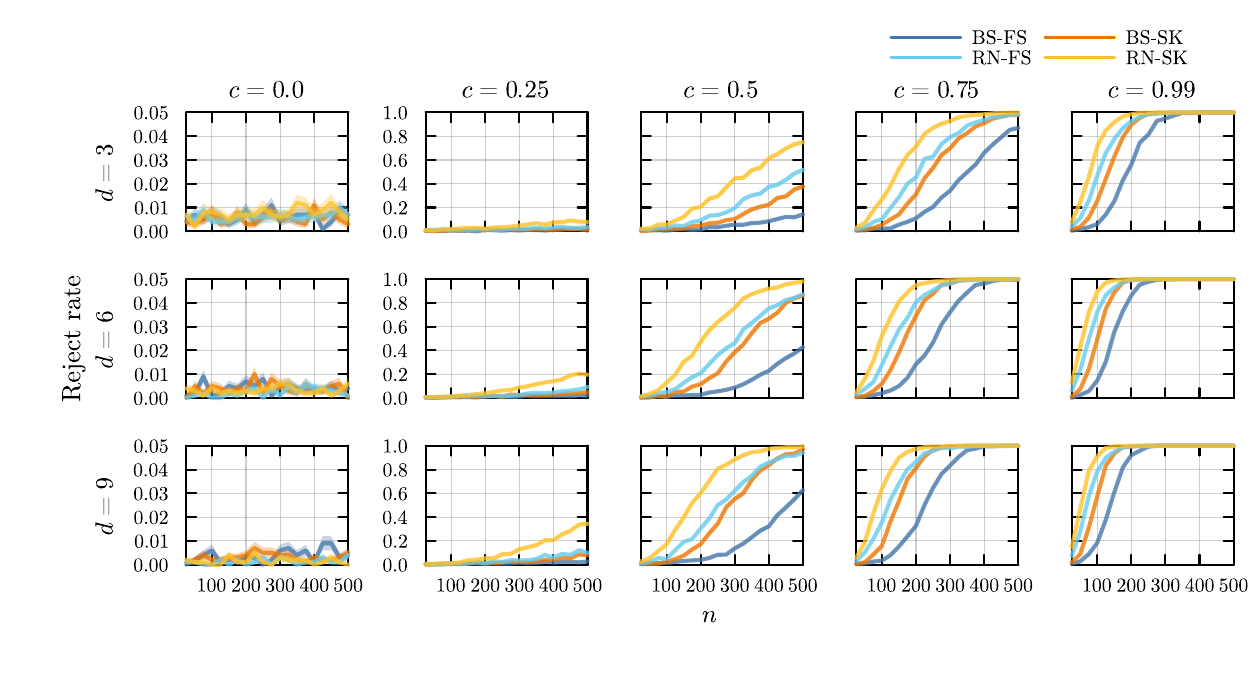}
    \caption{Rejection rates for tests that reuse a comparison set in the changing covariance experiment (\cref{sec:exp:mvn}) for varying values of the data dimension $d$, $Y|X$ conditional covariance matrix off-diagonal values $c$, and sample size $n$.
    The conditional distribution is $\SO{d}$-equivariant only for $c=0$. (Note the different vertical scale for $c = 0$ versus $c > 0$.)}
    \label{fig:gauss:equiv:cov}
\end{figure}

Consider random variables $X$ and $Y$ in $\R^d$. Let $X$ have a marginal distribution given by $P_X=\gaussian(0,\Sigma_{X,d})$, where $\Sigma_{X,d}$ is the $d\times d$ matrix with $1$'s on the diagonal and $0.75$ on the off-diagonals.
Let $Y|X$ have the conditional distribution $\gaussian(X,\Sigma_{Y,d}(c))$, where $\Sigma_{Y,d}(c)$ is the $d\times d$ matrix with $1$'s on the diagonals and $c \in [0,1)$ on the off-diagonals.
In words, the distribution of $Y|X$ is a multivariate normal centered on $X$.
For $c=0$, the distribution is spherical; as $c$ moves towards $1$, the distribution becomes more elliptical and eventually line-like.
Under the action of $\SO{d}$, the conditional distribution is $\SO{d}$-equivariant if and only if $c=0$.

\begin{figure}[tp]
    \centering
    \includegraphics[trim={26 25 0 0},clip,width=0.9\textwidth]{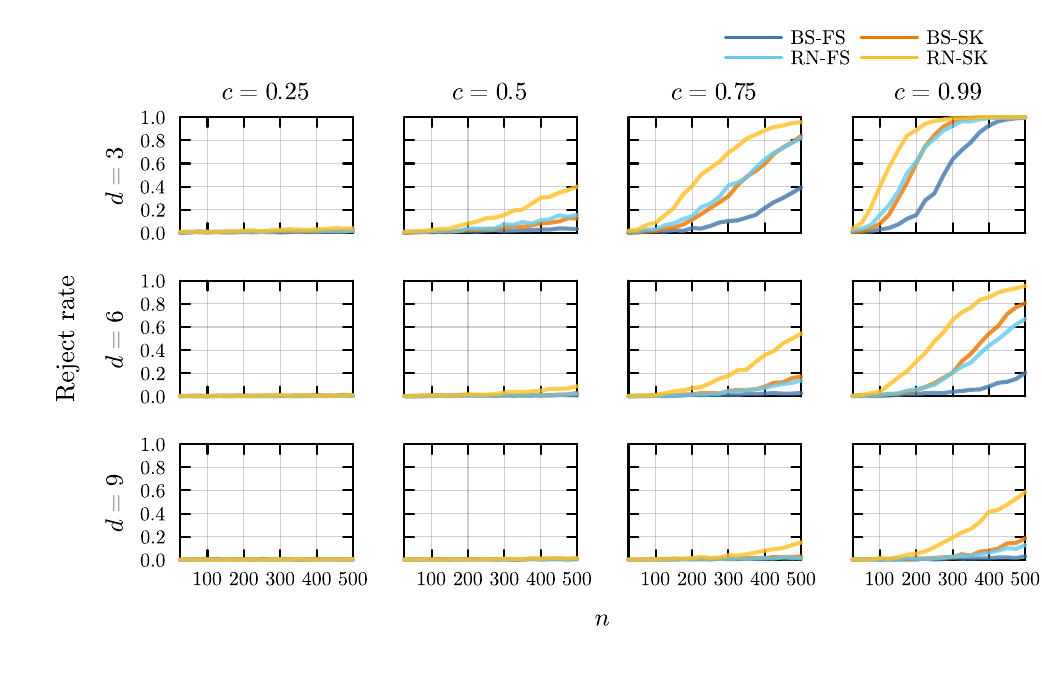}
    \caption{Rejection rates for tests that reuse a comparison set in the changing covariance experiment (\cref{sec:exp:mvn}), where the non-$\SO{d}$-equivariance comes from the covariance $c$ between the first two dimensions, for varying values of the data dimension $d$ and sample size $n$. (Note the different vertical scale for $c = 0$ versus $c > 0$.)}
    \label{fig:gauss:nonequiv:cov}
\end{figure}

To test for $\SO{d}$-equivariance on data generated from these distributions, we condition on the maximal invariant $\maxInv(X)=\|X\|$.
\Cref{fig:gauss:equiv:cov} shows the estimated test rejection rates for various values of $c$ and $d$ as the sample size $n$ increases from $25$ to $500$.
The randomization tests appear to be more powerful than the baseline tests in general.
We also observe that increasing the number of dimensions leads to an increase in power, because each dimension provides additional signal for departures from the null hypothesis. 
In contrast, consider a similar setup where the only non-zero off-diagonal entry of $\Sigma_{Y,d}$ is between the first and second dimensions.
That is, the non-equivariance of the distribution is isolated to the first two dimensions.
\Cref{fig:gauss:nonequiv:cov} shows the rejection rates: the power decreases with dimension, which suggests that the non-equivariance signal-to-noise ratio decreases with additional ``equivariant dimensions.''

\subsubsection{\label{sec:exp:synth:truth}Approximate versus exact conditional sampling}

Consider a setting where $X$ is generated in the following way. Let $Z\sim\chi_1$ be a chi random variable with $1$ degree of freedom, and take $X\coloneq Z\G\e_1$, where $\G|Z,\eps \sim \repInvKern(\argdot| Z\onevec_d+\eps)$ is a $d$-dimensional rotation matrix, with $\onevec_d$ being the vector of ones of length $d$ and $\eps\sim\gaussian(0,\I_d)$. Here, the conditional distribution $\GKern$ is known and can be sampled from exactly. 
Let $Y|X$ have the same conditional distribution described in the previous section, with $c$ being the covariance between any pair of dimensions of $Y$, so that the conditional distribution of $Y|X$ is $\SO{d}$-equivariant if and only if $c=0$.

For $d\in\{3,4,5\}$, we compare the performance of tests that sample from an approximation of $\GKern$ as per \cref{alg:cond:rand} with tests that sample from $\GKern$ exactly.
The test results are shown in \cref{fig:exp:truth}.
We see that the tests equipped with exact conditional sampling appear to have level-$\sig$ (within reasonable variation due to sampling).
Moreover, the tests with approximate sampling are comparable to those tests in terms of power, though the approximation also appears to make the tests slightly conservative.
This effect is also reflected in \cref{fig:exp:truth:reuse:p}, which shows the estimated $p$-value distributions for these tests.
In particular, in the null (equivariant) setting $(c=0)$,
the estimated $p$-value distribution for the tests with exact sampling appears to be approximately uniform, and the tests with approximate sampling appear to be less uniform with slightly more mass on higher values and correspondingly lower mass on lower values.
In \cref{apx:exp:bw}, we further investigate the connection between the approximation and the properties of the test; in particular, we find that Silverman's rule for the bandwidth of the kernel used in the conditional randomization produces a conservative test in this setting, and that larger bandwidths can increase power at the risk of inflating Type I error rates.

\begin{figure}[tbp]
    \centering
    \includegraphics[trim={29 26 0 14},clip,width=\textwidth]{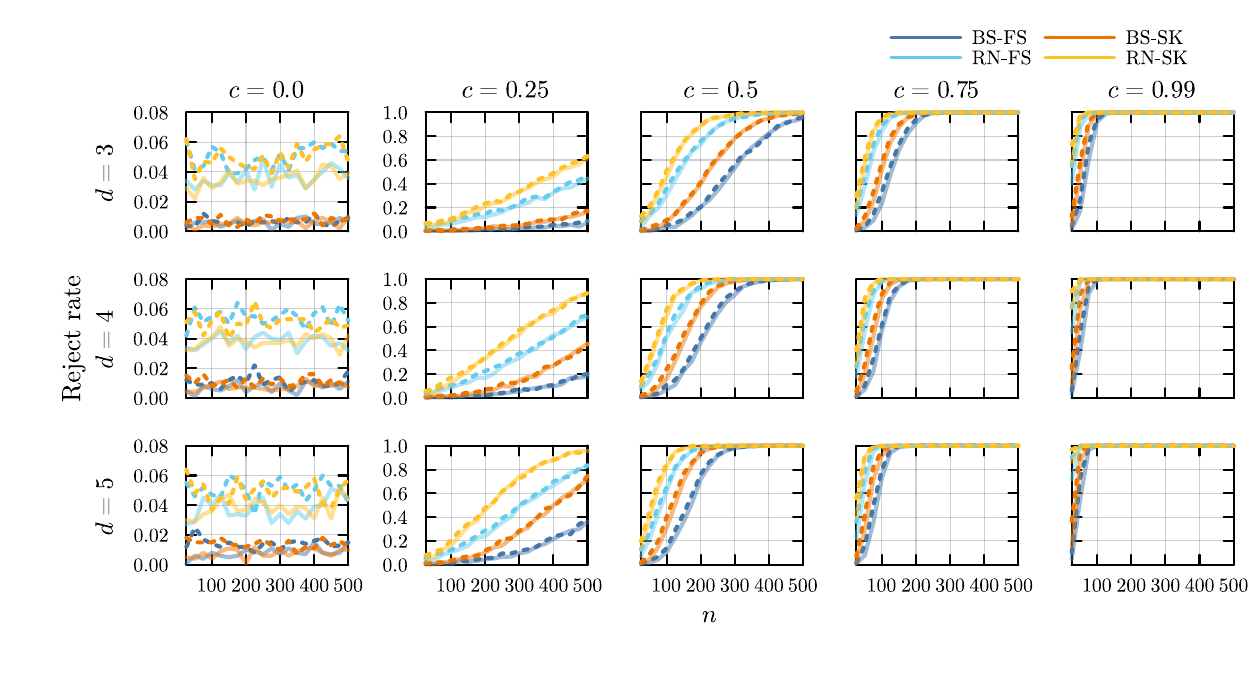}

    \caption{Rejection rates for $\SO{d}$-equivariance tests
    for varying values of the $Y|X$ covariance matrix off-diagonal values $c$ and sample size $n$. The dotted lines represent tests that sample from the true conditional $\GKern$, and the solid lines represent tests that sample from an approximation.  (Note the different vertical scale for $c = 0$ versus $c > 0$.)}
    \label{fig:exp:truth}
\end{figure}

\begin{figure}[tbp]
    \centering
    \includegraphics[trim={24 22 0 8},clip,width=0.9\textwidth]{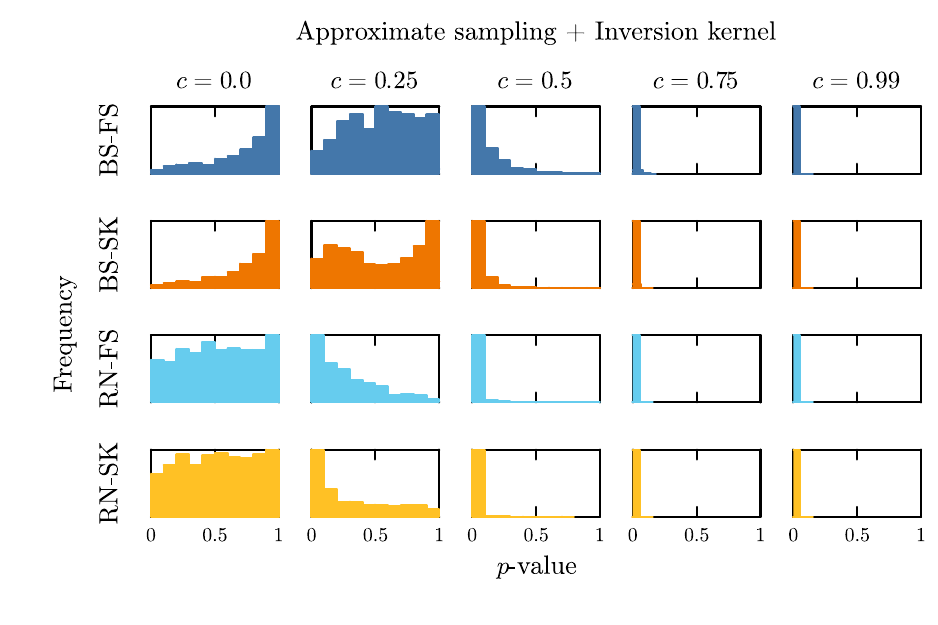}
    \vspace{1em}

    \includegraphics[trim={24 22 0 8},clip,width=0.9\textwidth]{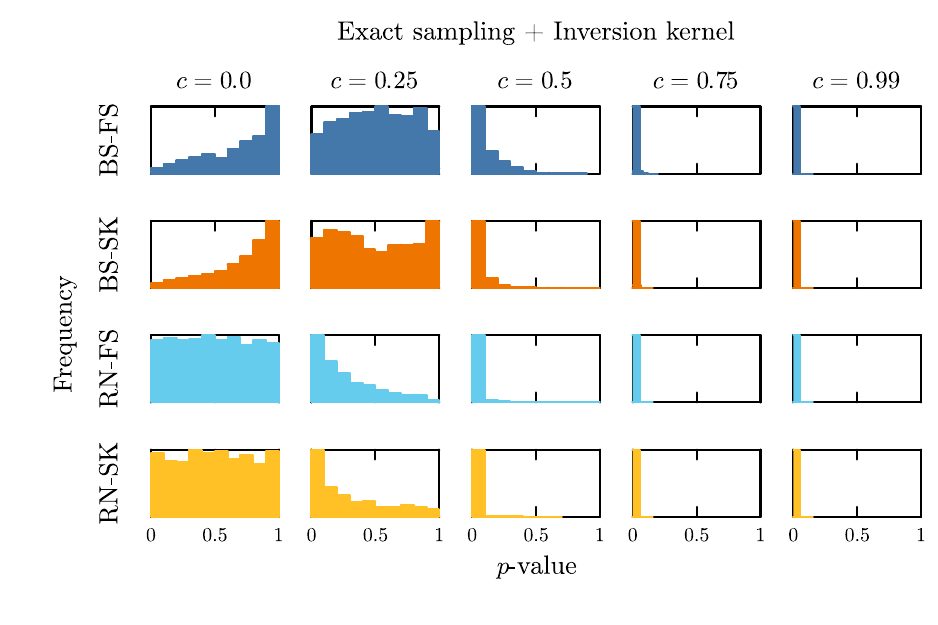}
    \caption{Histograms showing the distribution of $1000$ $p$-values of the $\SO{3}$-equivariant tests in the approximate versus exact sampling experiment at sample size $n=250$. Parameter $c$ is the covariance between any two dimensions.}
    \label{fig:exp:truth:reuse:p}
\end{figure}

\subsubsection{\label{sec:exp:other}Additional experiments}

Additional synthetic experiments are found in \cref{apx:exp}. 
\cref{apx:exp:mvn} compares tests that use independent comparison sets to those in \cref{sec:exp:gauss} that reuse comparison sets. It finds that in addition to being less computationally expensive, reused comparison sets appear to result in a slight increase in power. 
\Cref{apx:exp:risk} explores the risks of using randomization-based inference using measures proposed by \citet{stoepker2024inference}; there, we find that the Monte Carlo estimation of the $p$-value does not greatly interfere with the outcome of the test but does potentially result in $p$-values that are less interpretable.
\Cref{apx:exp:rep:inv} examines how using a deterministic representative inversion (if it exists) when the group action is not free can still yield a level-$\sig$ test.
\Cref{apx:exp:comp} investigates the runtimes of several variations of the FUSE test and shows that our tests are not prohibitively expensive; in particular, it suggests that setting $B$ to a value large enough to ensure accurate Monte Carlo estimates requires only a modest computational budget.
\Cref{apx:exp:mvn:sens,apx:exp:synth:perm} consider different experimental contexts where rotation-equivariance is lost due to changes in the mean as opposed to in the covariance, and where permutation-equivariance is the symmetry of interest, respectively.
In both experiments, our tests detect the alternative setting with large enough sample sizes.

\subsection{Large Hadron Collider dijet events} \label{sec:exp:LHC}

Our first application demonstrates how our tests could be used to verify symmetry presumed by theory and to detect possible instances of symmetry breaking for scientific discovery.
The Large Hadron Collider (LHC) Olympics 2020 dataset \citep{Kasieczka:2021,lhc:olympics:data} consists of 1.1 million simulated dijet events generated by PYTHIA \citep{Bierlich:2022aa}, a widely-used Monte Carlo simulator for high-energy physics processes.
A \emph{dijet} is two jets of constituent particles produced from the decay of high-energy quarks following a particle collision.  
The transverse momentum $p_T$ and polar angle $\phi$ 
for up to 200 constituents were recorded for each jet; the Cartesian momentum in the transverse plane for each constituent can be computed from these values as the pair $(p_1 = p_\text{T}\cos(\phi), p_2 = p_\text{T}\sin(\phi))$.
The \emph{leading constituents} of a jet are the particles with the largest transverse momenta in any direction.

\begin{table}[tb]
\centering
\caption{\label{tab:lhc:equiv}Rejection rates (over $1000$ repetitions) for tests for $\SO{2}$-equivariance on LHC dijet data samples of size $n=25$.}
\begin{tabular}{@{\extracolsep\fill}rrrrr@{\extracolsep\fill}}
\toprule
& \multicolumn{4}{@{}c@{}}{Test} \\
\cmidrule(l){2-5}
& \nBSFS & \nBSSK & \nRNFS & \nRNSK \\
\midrule
Original data & $0.003$ & $0.001$ & $0.000$ & $0.000$ \\
\midrule
Modified data & $0.694$ & $0.665$ & $0.796$ & $0.817$ \\
\bottomrule
\end{tabular}
\end{table}

By conservation of angular momentum, the Cartesian momenta of one leading constituent, $Y=(p_{1,Y},p_{2,Y})$, should be equivariant with respect to $\SO{2}$ rotations of the other, $X=(p_{1,X},p_{2,X})$.
We conduct tests for $\SO{2}$-equivariance on random samples of size $n=25$ drawn from the full data.
The estimated rejection rates based on $B=1000$ randomizations are shown in \cref{tab:lhc:equiv}.
We find that all tests achieve an average rejection rate of no more than $0.005$, suggesting the presence of $\SO{2}$-equivariance, which is consistent with physical theory.
As a check for power, we also consider tests for $\SO{2}$-equivariance on samples of the same size, where
the momenta of $Y$ are replaced by $Y'=(|p_{1,Y}|,p_{2,Y})$, breaking the conditional equivariance.
The tests all reject with a power greater than $0.65$, which suggests the absence of symmetry in this setting.
Here, note that an equivalent test for $\SO{2}$-equivariance is a test for joint $\SO{2}$-invariance.
In \cref{apx:exp:lhc}, we demonstrate what a test for joint invariance on these data would look like.

\subsection{Lorentz symmetry of quark decays}
\label{sec:experiments:top:quark}

We consider a second particle physics application based on the Top Quark Tagging Reference dataset \citep{Kasieczka:2019}, which also consists of jet events simulated by PYTHIA. 
The original dataset was constructed for the task of classifying jet events as having decayed from a top quark or not and consists of a training, validation, and test set.
The four-momenta $x=(E,p_1,p_2,p_3)$ of up to 200 jet constituents ordered by transverse momentum (i.e., highest magnitude momentum in the plane perpendicular to the particle beam line) are recorded for each event.
We only use the test set, which contains 404,000 simulated events, and ignore the binary label that indicates whether the jet decayed from a top quark or not.
According to the Standard Model, the four-momenta of the constituents in a jet should be equivariant with respect to the Lorentz group $\lorentz$, which was described in \cref{expl:lorentz,expl:lorentz:cont}, and preserves the quadratic form $Q(x) = E^2 - p_1^2 - p_2^2 - p_3^2$.

We take $X$ to be the four-momentum of the leading constituent in each jet, and $Y$ to be that of the second leading constituent \citep[as in][]{Yang:2023aa}. 
We conduct tests for $\lorentz$-equivariance of $P_{Y\mid X}$ based on samples of size $n=100$ using the maximal invariant $Q(x)$.
The estimated rejection rates based on $B=1000$ randomizations are shown in \cref{tab:tqt}.
The tests all reject $\lorentz$-equivariance at a rate of less than $0.01$, indicating symmetry consistent with the Standard Model.
We check for power by permuting the second leading constituents, breaking the dependence and thus the conditional equivariance.
In this setting, the conditional randomization tests are able to detect the breaking of $\lorentz$-equivariance with reasonable power. The baseline permutation tests are not as powerful in this setting, demonstrating the statistical benefit of using the group structure in the randomization scheme.

\begin{table}[tb]
\centering
\caption{\label{tab:tqt}Rejection rates (over $1000$ repetitions) for tests for $\lorentz$-equivariance on quark decay data samples of size $n=100$.}
\begin{tabular}{@{\extracolsep\fill}rrrrr@{\extracolsep\fill}}
\toprule
& \multicolumn{4}{@{}c@{}}{Test} \\
\cmidrule(l){2-5}
& \nBSFS & \nBSSK & \nRNFS & \nRNSK \\
\midrule
Original data & $0.000$ & $0.002$ & $0.001$ & $0.002$ \\
\midrule
Shuffled data & $0.292$ & $0.175$ & $0.812$ & $0.768$ \\
\bottomrule
\end{tabular}
\end{table}

\begin{appendix}

\section{Additional background on groups}
\label{apx:group:background}

A \textit{group} $\grp$ is a set with a binary operation~$\cdot$ that satisfies the associativity, identity, and inverse axioms. We denote the identity element by $\id$. 
For notational convenience, we write $g_1g_2=g_1\cdot g_2$ for $g_1,g_2\in\grp$. 
A group $\grp$ is said to be measurable if the group operations $g \mapsto g^{-1}$ and $(g_1,g_2) \mapsto g_1 g_2$ are $\bfS_{\grp}$-measurable, where $\bfS_{\grp}$ is a $\sigma$-algebra of subsets of $\grp$. 
For $A\subseteq \grp$ and $g \in \grp$, we write $gA = \{ gh : h \in A\}$ and $Ag = \{ hg : h \in A \}$. A measure $\nu$ on $\grp$ is said to be \textit{left-invariant} if $\nu(gA) = \nu(A)$ for all $A \in \bfS_{\grp}$, and \textit{right-invariant} if $\nu(Ag) = \nu(A)$. 
When $\grp$ is lcscH, there exist left- and right-invariant $\sigma$-finite measures $\haar_{\grp}$ and $\rhaar_{\grp}$, known as \textit{left-} and \textit{right-Haar measures}, respectively, that are unique up to scaling~\citep[Ch.~2.2]{Folland:2016}.

A \textit{group action} is a function $\grpAct\colon \grp\times\bfX \to \bfX$ that satisfies $\grpAct(\id,x) = x$ and $\grpAct(g_1 g_2, x) = \grpAct(g_1,\grpAct(g_2,x))$ for all $x \in \bfX, g_1,g_2\in \grp$. 
A group $\grp$ acts measurably on a set $\bfX$ if the group action $(g,x) \mapsto gx$ is measurable relative to $\bfS_{\grp} \otimes \bfS_{\bfX}$ and $\bfS_{\bfX}$. 
Throughout, all actions are assumed to be continuous (and therefore measurable), and for convenience we simply say that $\grp$ acts on $\bfX$, writing $gx=\grpAct(g,x)$. 
This standard notation hides the particular action $\grpAct$ associated with $gx$, and it is important not to lose track of that fact. 

The \textit{stabilizer subgroup} of $x\in\bfX$ is the subset of group actions that leave $x$ invariant, denoted $\grp_x = \{ g \in \grp : gx = x \}$.
The action is \textit{free} (or \textit{exact}) if $gx = x$ implies that $g = \id$, in which case $\grp_x = \{\id\}$ for all $x \in \bfX$. 
The \textit{$\grp$-orbit} of $x \in \bfX$ is the subset of all points in $\bfX$ reachable through a group action in $\grp$, denoted $\orbit(x)=\{gx : g\in\grp\}$.
Each orbit is an equivalence class of $\bfX$, where two points are equivalent if and only if they belong to the same orbit.
The stabilizer subgroups of the elements of an orbit are all conjugate; that is, if $gx = x'$ for $x \neq x'$, then $g \grp_x g^{-1} = \grp_{x'}$.
If $\bfX$ has only one orbit, then the action is said to be \textit{transitive}.

\subsection{Proper group actions} \label{apx:proper:action}

In the analysis of probabilistic aspects of group actions, measurability issues can arise without regularity conditions.
The key regularity condition that we assume is that the group action is \emph{proper}: there exists a strictly positive measurable function $h \colon \bfX \to \R_+$ such that $\int_\grp h(gx) \haar(dg) < \infty$ for every $x \in \bfX$ \citep{Kallenberg2007}.
This definition of proper group action is a non-topological version of the definition commonly used in previous work in the statistics literature \citep[e.g.,][]{Eaton:1989,WIjsman:1990,McCormack2023}, and only requires the existence of Haar measure. The common topological version defines $(g,x) \mapsto (gx, x)$ as a proper map, i.e., the inverse image of each compact set in $\bfX\times\bfX$ is a compact set in $\grp \times \bfX$. That definition implies the one we use; see \cite{Kallenberg2007} for details.  

A sufficient condition for proper group action is that $\grp$ is compact and acts continuously on $\bfX$.
When $\grp$ is non-compact, a group action can fail to be proper if $\grp$ is ``too large'' for $\bfX$ in the sense that the stabilizer subgroups are non-compact. A class of non-compact group actions known to be proper are those of the isometry group of a non-compact Riemannian manifold. 
In this work, we rely on the assumption of proper group actions to guarantee the existence of measurable orbit selectors and inversion kernels. The proposition below compiles some known useful properties of these objects.

To state the proposition, let $\nu$ be any bounded measure on $(\bfX,\bfS_{\bfX})$ and let $\bar{\bfS}^{\nu}_{\bfX}$ be the completion of $\bfS_{\bfX}$ to include all subsets of $\nu$-null sets, and denote by $\bar{\nu}$ the extension of $\nu$ to $\bar{\bfS}^{\nu}_{\bfX}$ \citep[see, e.g.][Proposition~1.3.10]{Cinlar_2011}.
All statements of $\bar{\nu}$-measurability in the following proposition are with respect to $\bar{\bfS}^{\nu}_{\bfX}$, so that a set $A\subseteq\bfX$ is $\bar{\nu}$-measurable if $A \in \bar{\bfS}^{\nu}_{\bfX}$.
Moreover, a function defined by a particular property is $\bar{\nu}$-measurable if it is measurable in the usual sense with respect to $\bar{\bfS}^{\nu}_{\bfX}$, and if the defining property holds with the possible exception of a $\bar{\nu}$-null set. 
Such a function would also be $\bar{\upsilon}$-measurable for any measure $\upsilon \ll \nu$. 

\begin{proposition} \label{prop:orbsel:inv}
    Let $\grp$ be a lcscH group acting continuously and properly on $\bfX$, and $\nu$ any bounded measure on $\bfX$.
    Then the following hold:
    \begin{enumerate}
        \item The canonical projection $\pi \colon \bfX \to \bfX/\grp$ is a measurable maximal invariant. Any measurable $\grp$-invariant function $f \colon \bfX \to \bfY$ can be written as $f = f^*\circ \pi$, for some measurable $f^* \colon \bfX/\grp \to \bfY$, and $f^*$ is bijective if and only if $f^*\circ\pi$ is a measurable maximal invariant. 

        \item There exists a $\bar{\nu}$-measurable orbit selector $\orbSel \colon \bfX \to \bfR$, which is a maximal invariant statistic, and it induces a $\bar{\nu}$-measurable cross-section $\bfR$. 

        \item For a fixed $\bar{\nu}$-measurable orbit selector $\orbSel$, there exists a unique $\bar{\nu}$-measurable inversion probability kernel $\repInvKern \colon \bfX \times \bfS_{\grp} \to [0,1]$ with the following properties:
            \begin{enumerate}
                \item $\repInvKern$ is $\grp$-equivariant: For all $g \in \grp, x \in \bfX, B \in \bfS_{\grp}$, $\repInvKern(B| gx) = \repInvKern(g^{-1} B| x)$. 

                \item For each $x \in \bfX$, $\repInvKern(\argdot| \orbSel(x))$ is normalized Haar measure on the stabilizer subgroup $\grp_{\orbSel(x)}$.

                \item For each $x \in \bfX$, if $\repInvRand|X=x\sim \repInvKern$, then $\repInvRand \orbSel(x) = x$ with probability one. 

                \item If there is a $\bar{\nu}$-measurable representative inversion $\repInv \colon \bfX \to \grp$ associated with $\orbSel$ such that it satisfies $\repInv(x)\orbSel(x) = x$  and $\repInv(gx) = g\repInv(x)$ for each $x \in \bfX, g \in \grp$, then $\repInvKern'(B| x) = \repInvKern(\repInv(x)^{-1}B| \orbSel(x))$ is an equivalent inversion kernel.
                In particular, this holds when the action of $\grp$ on $\bfX$ is free, in which case $\grp_{\orbSel(x)} = \{\id \}$ and the inversion kernel is $\delta_{\repInv(x)}$. 
            \end{enumerate}

    \end{enumerate}
\end{proposition}

The measurability of the canonical projection is a result from functional analysis; see \citet[][Theorem~5.4]{Eaton:1989} for an extended statement and references. One implication is that $\pi$ generates the invariant $\sigma$-algebra on $\bfX$, so that every invariant function can be written as a measurable function of it. 
Items 2--3c follow directly from results of \citet{Kallenberg2011:skew,Kallenberg:2017} on the existence of universally measurable versions of $\orbSel$ and $\repInvKern$.
Item 3d follows from 3a and 3b.

We rely heavily on the group-orbit kernel defined in \eqref{eq:GKern:def}, where the following is claimed. 

\begin{lemma} \label{lem:GKern}
    Let $\grp$ act properly on $\bfX$, $\GKern$ be as defined in \eqref{eq:GKern:def}, and $P_X$ be any probability measure on $\bfX$. If $X\sim P_X$, then  $(\orbSel(X),X) \equdist (\orbSel(X), \G \orbSel(X))$.
\end{lemma}
\begin{proof}
    For any measurable $f\colon \bfR \times \bfX \to \R$, by 3(c) of \cref{prop:orbsel:inv}, standard disintegration of measure, and Fubini's theorem, 
    \begin{align*}
        \int f(\orbSel(x),x) P_X(dx) &= \int f(\orbSel(x),\repInvRand \orbSel(x)) \repInvKern(d\repInvRand\mid x) P_X(dx) \\
        & = \int f(\orbRep,\repInvRand \orbRep) \repInvKern(d\repInvRand\mid x') P_{X \mid \orbRepRand }(dx'\mid\orbRep) (\orbSel_* P_X)(d\orbRep) \\
        & = \int f(\orbRep,\gamma \orbRep) \GKern(d\gamma\mid\orbRep) (\orbSel_* P_X)(d\orbRep) \\
        & = \int f(\orbSel(x),\gamma \orbSel(x)) \GKern(d\gamma\mid\orbSel(x)) P_X(dx) \;.
    \end{align*}
\end{proof}

\section{Further background on testing for distributional invariance} 
\label{apx:bg:testing}

As noted in \cref{sec:bg:testing:invariance}, the literature appears to be missing a proof of the conditional level for the randomization test in the most general setting of a compact group acting properly, with possible non-trivial stabilizer subgroups. 

For simplicity, we assume that the probability of ties is zero. If that is not the case (i.e., if $\testStat$ is supported on a discrete set), then a modified version with randomized tie-breaking yields a level-$\sig$ test \citep[e.g.,][]{Hemerik2018}.

\begin{theorem} \label{thm:mc:test:p}
    Let $\grp$ be a compact group that acts on $\bfX$. Let $\maxInv$ be any maximal invariant. Assume that $X_{1:n}\simiid P$. 
    For test statistic $\testStat$ and the randomization test in \cref{alg:group:test:new}, assume that $\bbP\{T(X^{(b)}_{1:n}) = T(X^{(b')}_{1:n})\}=0$ for $b \neq b'$. 
    For any fixed $B \in \mathbb{N}$ and $\sig \in [0,1]$, if $P$ is $\grp$-invariant then
    \begin{align*}
        \bbP\left(p_B \leq \sig \mid \maxInv(X)_{1:n}\right)  = \frac{\lfloor \sig (B+1) \rfloor}{B+1} \leq \sig \;, \quad P\text{-a.s.}
    \end{align*}
\end{theorem}

The high-level idea of the proof is as follows. 
It is known that any maximal invariant  is a sufficient statistic for $\invProbs$, the set of $\grp$-invariant probability distributions on $\bfX$ 
\citep{farrell:1962,dawid:1985,BloemReddy:2020}. The use of sufficiency greatly simplifies the subsequent technical details, allowing for the application of standard results for Monte Carlo tests, and also exposing randomization tests for invariance to be co-sufficient sampling methods \citep{Barber:2022}. 

The proof of \cref{thm:mc:test:p} relies on two auxiliary results. The first is a characterization of distributional invariance, which is well-studied in  probability \citep{Kallenberg:2005}. See \cite{Eaton:1989} for an accessible treatment of the basic results and some of their applications in classical statistics. We summarize two characterizations of invariance that are needed below. 

\begin{lemma} \label{thm:invariant_characterizations}
Let $\grp$ be a compact group acting on $\bfX$ and $P$ a probability measure on $\bfX$.
Let $\orbSel$ be a measurable orbit selector. 
With $X \sim P$, the following are equivalent. 
\begin{properties}[label=I\arabic*.,ref=I\arabic*]
    \setcounter{propertiesi}{-1}
    \item $P$ is $\grp$-invariant. \label{prp:inv:def}
    \item If $G \sim \haar$ with $G \condind X$, then $X \equdist GX$. \label{prp:inv:random:G}
    \item If $G \sim \haar$ and $R \sim \orbSel_{*}P$ with $G\condind R$, then $X \equdist GR$.
     \label{prp:inv:generate}
\end{properties}
\end{lemma}
\begin{proof}
    Suppose that $P$ is $\grp$-invariant. Then for any measurable $f \colon \bfX \to \R$,
    \begin{align*}
        \int_{\bfX\times\grp} f(gx)  P(dx) \haar(dg) = \int_{\bfX\times\grp} f(x)  P(dx) \haar(dg) = \int_{\bfX} f(x)  P(dx) \;,
    \end{align*}
    which shows that \cref{prp:inv:def} implies \cref{prp:inv:random:G}. For the converse, assume that $X \equdist GX$. Then for any $h \in \grp$,
    \begin{align*}
        \int_{\bfX} f(hx)  P(dx) & = \int_{\bfX\times\grp} f(hgx)  P(dx) \haar(dg) \\
        & = \int_{\bfX\times\grp} f(gx)  P(dx) \haar(dg) \\
        & = \int_{\bfX} f(x)  P(dx)  \;,
    \end{align*}
    where the second equality follows from the invariance of Haar measure. 

    The proof that \cref{prp:inv:generate} is equivalent to \cref{prp:inv:def} is more involved. An accessible proof can be found in \citet[][Theorems 4.3--4.4]{Eaton:1989}; see also \citet[][Theorem 7.15]{Kallenberg:2017}. 
\end{proof}

We note that \cref{prp:inv:generate} holds even conditionally on $\orbSel(X)$. That is,  
\begin{align*}
    (\orbSel(X), X, G) \equdist (\orbSel(X), G\orbSel(X), G) \;,
\end{align*}
which implies that ${X|\orbSel(X) \equdist G\orbSel(X)|\orbSel(X)}$.

One implication of this is that $\orbSel(X)$ is a sufficient statistic for the set of $\grp$-invariant probability measures on $\bfX$ \citep{farrell:1962,dawid:1985,BloemReddy:2020}. (In fact, any maximal invariant is sufficient.) To state the result precisely, for each $x \in \bfX$, let $\Phi_{\orbSel(x) *}\haar$ denote the pushforward of Haar measure by the group action applied to $\orbSel(x) \in \bfX$. 
\begin{corollary}
    Let $\grp$ be a compact group acting on $\bfX$, and $\invProbs$ be the set of $\grp$-invariant probability measures on $\bfX$. Then $\orbSel(X)$ is a sufficient statistic for $\invProbs$. That is, for each $P \in \invProbs$,
    \begin{align*}
        P(X \in A \mid \orbSel(X)) = \Phi_{\orbSel(X) *}\haar(A) \;, \quad A \in \bfS_{\bfX} \;,\quad P\text{-a.s.}
    \end{align*}
\end{corollary}

The second auxiliary result establishes a useful identity for exchangeable sequences of random variables, used extensively in the literature on randomization tests. 
A proof can be found in, e.g., \citet{Dufour:2006}. 
For simplicity, we assume that the probability of ties is zero, but that case can be handled with a randomized tie-breaking procedure described by \citet{Dufour:2006,Hemerik2018}. 

\begin{lemma}[\cite{Dufour:2006}, Proposition 2.2] \label{lem:dufour}
    Let $S_0, S_1,\dotsc,S_B$ be an exchangeable sequence of $\R$-valued random variables such that $\bbP(S_i = S_j) = 0$ for $i\neq j$, $i,j \in \{0,\dotsc,B\}$. Define
    \begin{align*}
        p_B = \frac{1 + \sum_{b = 1}^B \indi\{S_b \geq S_0\}}{B+1} \;.
    \end{align*}
    Then for any $\alpha \in [0,1]$,
    \begin{align*}
        \bbP(p_B \leq \alpha) = \frac{\lfloor \alpha (B+1) \rfloor}{B+1} \;.
    \end{align*}
\end{lemma}

\begin{proof}[Proof of \cref{thm:mc:test:p}]
    Due to the sufficiency of $\orbSel(X)$ for $\invProbs$, the samples $X^{(b)}_{1:n} = (G^{(b)}_1 X_1,\dotsc, G^{(b)}_n X_n)$ are conditionally i.i.d.\ (and therefore conditionally exchangeable) given $\orbSel(X)_{1:n}$, with the same conditional distribution as $X_{1:n}$ under $\nullHyp$.
Consequently,  $\left(T(X^{(b)}_{1:n})\right)_{b=0}^B$ are conditionally i.i.d.\ given $\orbSel(X)_{1:n}$, with the same conditional distribution as $T(X_{1:n})$ under the null.
It then follows from \cref{lem:dufour} that \eqref{eq:mc:test:p:valid} holds.
\end{proof}

\section{Proofs and additional theoretical results} \label{apx:proofs}

\subsection{Proof of \cref{thm:equivariance:cond:ind}}
\label{apx:proofs:thm:4:1}

\begin{proof}[Proof of \cref{thm:equivariance:cond:ind}]
    Define the random variable $\tY \define \repInvRand^{-1} Y$, where $\repInvRand|X \sim \repInvKern$, so that $P_{\tY|\repInvRand,X}(d\ty|\repInvRandCond,x) = \int \delta_{\repInvRandCond^{-1} y}(d\ty) P_{Y|X}(dy|x)$. 
    For any integrable function $f \colon \grp \times \bfX \times \bfY \to \R$, by Fubini's theorem,
    \begin{align*} %
        \int f(\repInvRandCond, x, \ty) & P_{\tY|\repInvRand,X}(d\ty\mid\repInvRandCond, x)\repInvKern(d\repInvRandCond\mid x)    P_X(dx) \\
        & = \int f(\repInvRandCond, x,\repInvRandCond^{-1} y) P_{Y|X}(dy\mid x)  \repInvKern(d\repInvRandCond\mid x)   P_X(dx)  \;.
    \end{align*}
    From this follows the identity $P_{\tY|\repInvRand,X}(B|\repInvRandCond,x) = P_{Y|X}(\repInvRandCond B | x)$, each $B \in \bfS_{\bfY}$. 

    Now assume that $P_{Y|X}$ is equivariant, so that for each $g \in \grp, x \in \bfX, B \in \bfS_{\bfY}$, $P_{Y|X}(B|gx) = P_{Y|X}(g^{-1}B|x)$ for $P_X$-almost all $x$.
    Then for any $\repInvRandCond \in \grp, x \in \bfX, g \in \grp$ and integrable $f \colon \bfY \to \R$,
    \begin{align*}
        \int  P_{\tY|\repInvRand,X}(d\ty\mid g\repInvRandCond, gx) f(\ty) & = \int f((g\repInvRandCond)^{-1} y) P_{Y|X}(dy\mid gx)  \\
          & = \int f(\repInvRandCond^{-1} g^{-1} g y) P_{Y|X}(dy\mid x)  \\
          & = \int f(\repInvRandCond^{-1}y) P_{Y|X}(dy\mid x)  \\
          & = \int f(\ty) P_{\tY|\repInvRand,X}(d\ty\mid  \repInvRandCond, x)  \;.
    \end{align*}
    This shows that the mapping $(\repInvRandCond, x) \mapsto P_{\tY|\repInvRand,X}(\argdot | \repInvRandCond, x)$ is $\grp$-invariant.
    Therefore, by \cref{prop:orbsel:inv}, for any measurable maximal invariant $\tM \colon \grp \times \bfX \to \bfM$, there is a unique Markov probability kernel $\tR \colon \bfM \times \bfS_{\bfY} \to [0,1]$ such that
    \begin{align*}
        P_{\tY|\repInvRand,X}(d\ty\mid \repInvRandCond, x) = P_{\tY|\tM}(B\mid  \tM(\repInvRandCond, x)) \;, \quad \repInvRandCond \in \grp, \ x \in \bfX, \ B \in \bfS_{\bfY} \;.
    \end{align*}
    Because the action of $\grp$ on itself is transitive (i.e., there is only one orbit in $\grp$), any maximal invariant $M$ for $\grp$ acting on $\bfX$ is also a maximal invariant for $\grp$ acting on $\grp \times \bfX$, and
    \begin{align} \label{eqn:invariant:kernel}
        P_{\tY|\repInvRand,X}(B\mid  \repInvRandCond, x) = \tR(B\mid \maxInv(x)) \;, \quad \repInvRandCond \in \grp, \  x \in \bfX, \ B \in \bfS_{\bfY} \;.
    \end{align}
    This is enough to establish the desired conditional independence in \eqref{eqn:cond:ind:rand}: For any integrable $f \colon \grp\times \bfX \times \bfY \to \R$,
    \begin{align*}
        \int & f(\repInvRandCond,x, \ty)   P_{\tY|\repInvRand,X}(d\ty\mid \repInvRandCond, x) \repInvKern(d\repInvRandCond\mid x) P_X(dx) \\
        & = \int f(\repInvRandCond,x, \ty)  \tR(d\ty \mid \maxInv(x)) \repInvKern(d\repInvRandCond \mid x)   P_X(dx) \;.
    \end{align*}
    
    Conversely, assume that $(\repInvRand,X) \condind \repInvRand^{-1} Y \mid \maxInv(X)$.
    Then \eqref{eqn:invariant:kernel} holds for $P_X\otimes \repInvKern$-almost all $(\repInvRandCond,x) \in \grp\times\bfX$.
    In particular, $P_{\tY|\repInvRand,X}$ is $\grp$-invariant in $(\repInvRandCond,x)$ for $P_X\otimes \repInvKern$-almost all $(\repInvRandCond,x)$.
    Recall also that the inversion kernel $\repInvKern$ is $\grp$-equivariant.
    Therefore, for any integrable $f \colon \grp \times \bfX \times \bfY \to \R$ and any $g \in \grp$,
    \begin{align*}
        &\int f(\repInvRandCond,x,y)  \repInvKern(d\repInvRandCond\mid x) P_{Y|X}(dy\mid x) P_X(dx)  \\
          &\qquad = \int f(\repInvRandCond,x,\repInvRandCond (\repInvRandCond^{-1} y))  \repInvKern(d\repInvRandCond\mid x) P_{Y|X}(dy\mid x) P_X(dx) \\
          &\qquad = \int f(\repInvRandCond,x,\repInvRandCond \ty) P_{\tY|\repInvRand,X} (d\ty\mid \repInvRandCond,x)  \repInvKern(d\repInvRandCond \mid x)  P_X(dx)  \\
          &\qquad = \int f(\repInvRandCond,x,\repInvRandCond \ty) P_{\tY|\repInvRand,X}(d\ty\mid g\repInvRandCond,gx)  \repInvKern(d\repInvRandCond\mid x)  P_X(dx) \\
          &\qquad = \int f(\repInvRandCond,g^{-1}x,\repInvRandCond \ty) P_{\tY|\repInvRand,X}(d\ty\mid g\repInvRandCond,x) \repInvKern(d\repInvRandCond\mid g^{-1} x)  (g_* P_X)(dx)   \\
          &\qquad = \int f(g^{-1}\repInvRandCond,g^{-1}x,g^{-1}\repInvRandCond \ty) P_{\tY|\repInvRand,X}(d\ty\mid \repInvRandCond,x) \repInvKern(d\repInvRandCond\mid x) (g_* P_X)(dx) \\
          &\qquad =\int f(g^{-1}\repInvRandCond,g^{-1}x,g^{-1}y)  \repInvKern(d\repInvRandCond\mid x) P_{Y|X}(dy\mid x) (g_* P_X)(dx) \\
          &\qquad = \int f(g^{-1}\repInvRandCond,x,g^{-1}y)  \repInvKern(d\repInvRandCond\mid gx) P_{Y|X}(dy\mid gx) P_X(dx) \\
          &\qquad =\int f(\repInvRandCond,x,g^{-1}y)  \repInvKern(d\repInvRandCond\mid x) P_{Y|X}(dy\mid gx) P_X(dx) \;.
    \end{align*}
    This implies that
    \begin{align} \label{eqn:almost:equiv}
        P_{Y|X}(B \mid  x) = P_{Y|X}(gB \mid  gx) \;, \quad B \in \bfS_{\bfY}, \ g \in \grp, \ P_X\text{-a.e.\ } x \in \bfX \;.
    \end{align}
    The subset of $\bfX$ for which \eqref{eqn:almost:equiv} holds is a $\grp$-invariant set \citep[][Lemma 7.7]{Kallenberg:2017}, and therefore the possible exceptional null set on which $P_{Y|X}$ is not equivariant does not depend on $g$.
    If there is such an exceptional null set on which $P_{Y|X}$ is not equivariant, denoted $N^{\times}$, define the Markov kernel $Q'$ as
    \begin{align*}
        Q'(B\mid x) \define
        \begin{cases}
            P_{Y|X}(B\mid x) & \text{if } x \notin N^{\times} \;, \\
            \int_{\grp} \repInvKern(d\repInvRandCond\mid x) P_{Y|X}(\repInvRandCond^{-1} B \mid  \repInvRandCond^{-1}x) & \text{if } x \in N^{\times} \;.
        \end{cases} 
    \end{align*}
    Since $\repInvKern$ and $P_{Y|X}$ are probability kernels, so too is $Q'$.
    It is also straightforward to show that $Q'$ is $\grp$-equivariant, so that $Q'$ is a regular version of $P_{Y|X}$ that is $\grp$-equivariant for all $x \in \bfX$, and equivalent to $P_{Y|X}$ up to the null set $N^{\times}$.

    If there exists a measurable representative inversion (function) $\repInv(\argdot)$, then the same proof holds with the inversion kernel $\repInvKern(\argdot|x)$ substituted by $\delta_{\repInv(x)}$, resulting in the simplified conditional independence statement in \eqref{eqn:cond:ind}. 

    If the action of $\grp$ on $\bfY$ is trivial, then $\tY = Y$. Moreover, $\repInvRand \condind Y \mid X$ by construction, and therefore $(\repInvRand, X) \condind Y \mid \maxInv(X)$ is implied by $X \condind Y \mid \maxInv(X)$. 
\end{proof}

\subsection{Further characterizations of conditional equivariance}
\label{apx:proofs:equivariance}

\cref{thm:equivariance:cond:ind} can be refined to obtain distributional identities that underpin the randomization tests of \cref{sec:rand}. Proofs follow \cref{cor:eq:xy:compact} below.

\begin{proposition} \label{thm:equiv:dist:id}
Let $P_{X,Y}=P_X\otimes P_{Y|X}$ be a probability measure on $\bfX\times\bfY$.
For a lcscH group $\grp$ acting on $\bfX$ and $\bfY$, with the action on $\bfX$ proper, let  $\GKern$ 
be the group-orbit kernel defined in \eqref{eq:GKern:def}. 
Then the following statements are equivalent to $P_{Y|X}$ being $\grp$-equivariant. 
\begin{properties}[label=E\arabic*.,ref=E\arabic*]
\setcounter{propertiesi}{0}
    \item \label{prp:eq:xy:all} 
    $(\orbSel(X),\tY,X,Y) \equdist (\orbSel(X),\tY, \G\orbSel(X),\G\tY)$, with $\G$ sampled as in \eqref{eq:equiv:randomization}.

    \item \label{prp:eq:xy:tau} 
    $(\orbSel(X),\repInvRand,X,Y) \equdist (\orbSel(X), \repInvRand, \repInvRand\orbSel(X),\repInvRand\tY)$, with $\tY$ sampled as in \eqref{eq:equiv:randomization}.

    \item \label{prp:eq:xy}
    $(\orbSel(X),X,Y) \equdist (\orbSel(X), \G\orbSel(X),\G\tY)$, with $\tY,\G$ sampled as in \eqref{eq:equiv:randomization}.

\end{properties}
\end{proposition}

\cref{thm:equiv:dist:id} has two special cases that are important for practical reasons. First, if $\grp$ acts non-trivially on $\bfY$, the data being used in the test can be reduced to $Y$ only. 

\begin{corollary} \label{cor:eq:xy:cor}
    Assume the conditions of \cref{thm:equiv:dist:id}. If $\grp$ acts non-trivially on $\bfY$ then the following statements are equivalent to $P_{Y|X}$ being $\grp$-equivariant.
    \begin{properties}[label=E\arabic*.,ref=E\arabic*]
        \setcounter{propertiesi}{3}
        \item \label{prp:eq:ty:y}
        $(\orbSel(X),\tY,Y) \equdist (\orbSel(X),\tY, \G\tY)$, with $\G$ sampled as in \eqref{eq:equiv:randomization}.
    
        \item \label{prp:eq:tau:y}
        $(\orbSel(X),\repInvRand,Y) \equdist (\orbSel(X),\repInvRand, \repInvRand\tY)$, with $\tY$ sampled as in \eqref{eq:equiv:randomization}.
    
        \item \label{prp:eq:y}
        $(\orbSel(X),Y) \equdist (\orbSel(X), \G\tY)$, with $\tY,\G$ sampled as in \eqref{eq:equiv:randomization}.
    \end{properties}
\end{corollary}

The second useful special case occurs when $P_X$ is invariant under a compact group, in which case the test for conditional symmetry is equivalent to joint invariance of $P_{X,Y}$.  Therefore, in such situations we can test for $\grp$-equivariance of $P_{Y|X}$ by testing for joint $\grp$-invariance of the distribution $P_{X,Y}$ or by testing for marginal $\grp$-invariance of the distribution $P_Y$ (when $\grp$ acts non-trivially on $\bfY$), leveraging the established body of work on group randomization tests for distributional invariance (\cref{sec:bg:testing:invariance}).

\begin{corollary} \label{cor:eq:xy:compact}
    If $\grp$ is compact and $P_X$ is $\grp$-invariant, the following are equivalent to $P_{Y|X}$ being $\grp$-equivariant. 
    \begin{properties}[label=E\arabic*.,ref=E\arabic*]
    \setcounter{propertiesi}{6}
        \item \label{prp:eq:joint:inv}
        $(\orbSel(X), \tY, X,Y)\equdist (\orbSel(X), \tY,G\orbSel(X),G\tY)$ with $G\sim\haar$.
        
        \item \label{prp:eq:marg:inv} 
        $\grp$ acts non-trivially on $\bfY$, and 
        $Y\equdist GY \equdist G\tY$ for $G\sim\haar$.
    \end{properties}
\end{corollary}

For the proof of the distributional identities in \cref{thm:equiv:dist:id} that are equivalent to equivariance, the following lemma will be useful. 

\begin{lemma} \label{lem:equiv:xy}
    Under the conditions of \cref{thm:equiv:dist:id}, $P_{Y|X}$ is $\grp$-equivariant if and only if $(X,Y) \equdist (\G\orbSel(X'), \G \tY)$, where $X \equdist X'$ and $\G,\tY$ are sampled as in \eqref{eq:equiv:randomization}. 
\end{lemma}

\begin{proof}
    We will prove that the distributional identity implies that $P_{Y|X}$ is equivariant. The converse follows as a special case of the stronger statements in the proof of \cref{thm:equiv:dist:id} below. 
    
    Fix arbitrary $g \in \grp$. From (i) the definition of $\GKern$, (ii) Fubini's theorem, and (iii) the equivariance of the inversion kernel $\repInvKern$, for any integrable $f \colon \bfX \times \bfY \to \mathbb{R}$,
    \begin{align*}
        \int &f(x,y) P_{Y|X}(dy\mid x) P_X(dx)  \\
        & = \int f(\g \orbRep, \g \ty) \tYKern(d\ty \mid  \orbRep) \GKern(d\g \mid  \orbRep) \orbRepKern(d\orbRep) \\
        & \overset{\text{(i),(ii)}}{=} \int f(\repInvRandCond \orbRep, \repInvRandCond \ty) \tYKern(d\ty \mid  \orbRep) \repInvKern(d\repInvRandCond \mid  x') \XorbRepKern(dx' \mid  \orbRep) \orbRepKern(d\orbRep) \\
        & \overset{\text{(iii)}}{=} \int f(g \repInvRandCond \orbRep, g \repInvRandCond \ty) \tYKern(d\ty \mid  \orbRep) \repInvKern(d\repInvRandCond \mid  g^{-1} x') \XorbRepKern(dx' \mid  \orbRep) \orbRepKern(d\orbRep) \\
        & = \int f(gx, gy) P_{Y|X}(dy\mid x) (g^{-1}_* P_X)(dx) \\
        & = \int f(x, gy) P_{Y|X}(dy \mid g^{-1} x) P_X(dx) \;,
    \end{align*}
    which shows that $P_{Y|X}$ is equivariant up to a $P_X$-null set. By a similar argument as the one in the proof of \cref{thm:equivariance:cond:ind}, the exceptional null set is $\grp$-invariant and therefore there is a version of $P_{Y|X}$ that is equivariant for all $g \in \grp$, $x \in \bfX$. 
\end{proof}

\begin{proof}[Proof of \cref{thm:equiv:dist:id}]
    Let $(X,Y)\sim P_{X,Y}$. Assume that $P_{Y|X}$ is\\$\grp$-equivariant. To show \cref{prp:eq:xy:all},
    for any integrable $f \colon \bfR \times \bfY \times \bfX \times \bfY \to \mathbb{R}$, 
    \begin{align*}
         \int  &f(\orbSel(x), \ty,  x, y) P(dx,d\ty,dy)  \\
        & = \int f(\orbRep,\ty, x', y) P_{Y|X}(dy\mid\ty, x') \tYKern(d\ty\mid\orbRep)  \XorbRepKern(dx'\mid\orbRep) \orbRepKern(d\orbRep) \\
        & = \int f(\orbRep, \ty, \repInvRandCond \orbRep, \repInvRandCond \ty) \tYKern(d\ty\mid\orbRep) \repInvKern(d\repInvRandCond\mid x') \XorbRepKern(dx'\mid\orbRep) \orbRepKern(d\orbRep) \\
        & =  \int f(\orbRep, \ty, \g \orbRep, \g \ty) \tYKern(d\ty \mid  \orbRep) \GKern(d\g \mid \orbRep) \orbRepKern(d\orbRep) \\
        & =  \int f(\orbSel(x), \ty, \g \orbSel(x), \g \ty) \tYKern(d\ty \mid  \orbSel(x)) \GKern(d\g \mid  \orbSel(x)) P_{X}(dx) \;,
    \end{align*}
    where the second equality is implied by \cref{thm:equivariance:cond:ind}, and the third equality follows from Fubini's theorem. This shows that $\grp$-equivariance of $P_{Y|X}$ implies \cref{prp:eq:xy:all}.

    To show that \cref{prp:eq:xy:all} implies that $P_{Y|X}$ is $\grp$-equivariant, observe that $(\orbSel(X),\tY, X,Y) \equdist (\orbSel(X),\tY, \G\orbSel(X),\G\tY)$ implies that $(X,Y) \equdist (\G\orbSel(X),\G\tY)$, where $\G,\tY$ are sampled as in \eqref{eq:equiv:randomization}. By \cref{lem:equiv:xy}, that implies that $P_{Y|X}$ is $\grp$-equivariant. 
    
    \cref{prp:eq:xy:tau} follows directly from the proof of \cref{thm:equivariance:cond:ind}, and \cref{prp:eq:xy} is a special case of \cref{prp:eq:xy:all} with $\tY$ marginalized. 
\end{proof}

\begin{proof}[Proof of \cref{cor:eq:xy:cor}]
    \cref{prp:eq:ty:y}--\cref{prp:eq:y} follow as special cases of the corresponding properties in \cref{thm:equiv:dist:id}, marginalizing over $X$.
\end{proof}

\begin{proof}[Proof of \cref{cor:eq:xy:compact}]
    \cref{prp:eq:joint:inv,prp:eq:marg:inv} follow from the fact that \\$\GKern(\argdot|\orbRep) = \haar(\argdot)$ for each $\orbRep$ if and only if $P_X$ is $\grp$-invariant. 
\end{proof}

\subsection{Proofs for \cref{sec:rand}} \label{apx:proofs:crt:validity}

\begin{proof}[Proof of \cref{prp:crt:validity}]
    By \cref{thm:equiv:dist:id}, each of the randomization procedures generates conditionally i.i.d.\ samples, conditioned on $S_{1:n}$. By construction, the randomized samples have the same distribution as the observed samples under the null hypothesis. Then \eqref{eq:crt:pval:level} follows from \cref{lem:dufour}. 

    The asymptotic consistency follows from the strong consistency of the empirical CDF (by the Glivenko--Cantelli theorem \citealp[Thm.\ 19.1]{vaart_1998}), along with the fact that under the null hypothesis that $P_{Y|X}$ is $\grp$-equivariant, the observations have the same distribution as the randomized samples. That is, the Glivenko--Cantelli theorem implies that for fixed $S_{1:n}$, 
    \begin{align*}
        \frac{1}{B+1}\sum_{b=1}^B\indi\left\{\testStat\left((X,Y)_{1:n}^{(b)}\right)\geq \testStat\left((X,Y)_{1:n}\right)\right\} \toas 1 - F_{\testStat|S}(\testStat(X,Y)_{1:n} \mid S_{1:n}) \;,
    \end{align*}
    where $F_{\testStat|S}$ is the conditional CDF of the test statistic given $S$ under the randomization distribution. This holds for $P$-almost every sequence $S_{1:n}$. 
\end{proof}

\begin{proposition} \label{prop:appx:pval}
    Assume $(X,Y)\sim P_{X,Y}$, with $P_{Y|X}$ $\grp$-equivariant. Let $\testStat$ be a continuous $\mathbb{R}$-valued function, such that the distribution of $\testStat((X,Y)_{1:n})$ is atomless. 
    Suppose that $(X,Y)_{i,m}^{(b)}$ are sampled conditionally on $S(X,Y)_{i,m}^{(b)} = S(X,Y)_i$, and used as the randomization sets to generate CRT statistics as in \eqref{eqn:crt:pval}, denoted $p_{B,m}$. If, for each $i\leq n$, conditioned on $S(X,Y)_i$, $(X,Y)_{i,m}^{(b)} \todist (X,Y)_i$ as $m\to\infty$, and $\bbP\{\testStat((X,Y)_{1:n}^{(b)}) = \testStat((X,Y)_{1:n}^{(b')})\} = 0$ for $b \neq b'$, then $p_{B,m} \todist p_B$. Therefore, for fixed $B \in \mathbb{N}$ and $\alpha \in [0,1]$,
    \begin{align*}
        \lim_{m\to\infty} \bbP\left(p_{B,m}\leq\sig \mid S_{1:n} \right) \leq \sig \;, \quad  P\text{-a.s.}
    \end{align*}
\end{proposition}

\begin{proof}%
    For each $m$, $(X,Y)_{i,m}^{(b)}$ are independent across $i$ and conditionally independent across $b$, conditioned on $S_{1:n}$. Hence, the conditional convergence in distribution $(X,Y)_{i,m}^{(b)} \todist (X,Y)_{i}$ holds jointly across $i$ and $b$. By the continuous mapping theorem, $(\testStat((X,Y)^{(b)}_{1:n,m}))_{1:B}$ converges jointly in distribution (across $b$) to $(\testStat((X,Y)^{(b)}_{1:n}))_{1:B}$. By assumption, ties occur with limiting probability zero and the distribution of $\testStat((X,Y)_{1:n})$ is atomless. Therefore, again by the continuous mapping theorem, it follows that $p_{B,m} \todist p_B$. 
\end{proof}

\subsection{Additional results and proofs for \cref{sec:kernel}} \label{apx:power:proofs}

\subsubsection{Proof of \cref{thm:mmd:power}}

The proof of \cref{thm:mmd:power} builds on the analysis of \citet{Sriperumbudur_2012} on the empirical estimation of MMD.  In particular, for RKHS $\hilbert$ and $\hilbert|_1 = \{f \in \hilbert : \| f\|_{\hilbert} \leq 1 \}$, the MMD for probability measures $P,Q$ on $\bfX$  can be written as
\begin{align*}
    \MMD^2&(P,Q) = \left(\sup_{f \in \hilbert|_1} \left| \int f(x) P(dx) - \int f(x) Q(dx) \right| \right)^2 \\
    & = \E[\k(X,\tilde{X})] + \E[\k(X',\tilde{X}')] - 2\E[\k(X,X')] \;, \quad X,\tilde{X}\simiid P\;, X',\tilde{X'}\simiid Q \;,
\end{align*}
with empirical estimator
\begin{align*}
    & \MMD^2(X_{1:n},X_{1:n}') = \left(\sup_{f \in \hilbert|_1} \left| \frac{1}{n} \sum_{i=1}^n (f(X_i) - f(X'_i)) \right| \right)^2 \;, \quad X_i \simiid P\;, \ X'_i \simiid Q \;.
\end{align*}

The proof of \cref{thm:mmd:power} relies on the following lemma, which is a special case of Proposition B.1 from \citep{Sriperumbudur_2012}, which itself extends a version of Talagrand's inequality \citep[][Thm.\ 7.5 and A.9.1]{steinwart:2008:svm} to apply to sequences of independent but not necessarily identically distributed random variables.

\begin{lemma} \label{lem:talagrand}
    Let $M\geq 0$, $n\geq 1$. For each $i\in [n]$, let $(\bfX,\calA,P_i)$ be a probability space. Let $\calF$ be a set of $\calA$-measurable $\R$-valued functions. Let $\theta_i \colon \calF \times \bfX \to \R$ be bounded measurable functions, such that for each $f \in \calF$, $i \in [n]$,
    \begin{itemize}
        \item $\int_{\bfX} \theta_i(f,x) P_i(dx) = 0$;
        \item $\int_{\bfX} \theta_i(f,x)^2 P_i(dx) \leq \varphi^2_i < \infty$; and
        \item $\| \theta_i(f,\argdot)\|_{\infty} \leq M < \infty$.
    \end{itemize}
    Define $\bfW = \bfX^n$, $P = P_1\times \dotsb \times P_n$, and $H\colon \bfW \to \R$ by
    \begin{align*}
        H(w) = \sup_{f\in\calF} \left| \sum_{i=1}^n \theta_i(f, x_i) \right| \;, \quad w = (x_1,\dotsc,x_n) \;.
    \end{align*}
    Then for all $\epsilon > 0$, 
    \begin{align*}
        P\left( w \in \bfW : H(w) - \E[H(W)] \geq \epsilon \right) \leq e^{-\frac{3}{2}\frac{\epsilon^2}{3a + \epsilon}} \;,
    \end{align*}
    with $a = 2M \E[H(W)]  + \sum_{i=1}^n \varphi^2_i$. Written differently, for all $\eta > 0$,
    \begin{align*}
        P\left( w \in \bfW : H(w) \geq \E[H(W)] + \sqrt{2a\eta} + 2\eta M/3 \right) \leq e^{-\eta} \;.
    \end{align*}
\end{lemma}

In order to apply the lemma to our setting, observe that our MMD-based test statistic evaluated on two (conditionally) independent randomization samples $Z_i \sim P_i$ and $Z_{n+i} \sim Q_i$ can be written as, with $b_i = (-1)^{\indi\{n+1\leq i \leq 2n\}}$,
\begin{align*}
    \testStat(Z_{1:n},Z_{(n+1):2n}) = \sup_{f \in \hilbert|_1} \left| \frac{1}{n}\sum_{i=1}^{2n} b_i f(Z_i)\right| = \left( \frac{1}{n^2} \sum_{i=1,j=1}^{2n} b_i b_j \k(Z_i,Z_j) \right)^{1/2} \;.
\end{align*}
Since there is a different distribution for each $i$, consider the ``population'' version of the statistic to be
\begin{align*}
    \testStat(P_{1:n},Q_{1:n}) = \sup_{f \in \hilbert|_1} \left| \frac{1}{n}\sum_{i=1}^{2n} b_i \E[f(Z_i)]\right| = \left( \frac{1}{n^2} \sum_{i=1,j=1}^{2n} b_i b_j \E[\k(Z_i,Z_j)] \right)^{1/2} \;,
\end{align*}
which is equal to $0$ if $P_i = Q_i$ for each $i \in [n]$. 
Finally, observe that with
\begin{align} \label{eq:thetai:def}
    \theta_i(f,z) = (-1)^{\indi\{n+1\leq i \leq 2n\}} \frac{1}{n}\bigg(f(z) - \int f(z') P_i(dz') \bigg) \;, \quad f \in \hilbert|_1\;,
\end{align}
the conditions of \cref{lem:talagrand} are satisfied, with $\varphi_i^2 \leq \supk^2/n^2$ and $M \leq 4\supk/n$. 

We will use it to bound
\stepcounter{equation}
\begin{align*}
&\left|\testStat(Z_{1:n},Z_{(n+1):2n}) - \testStat(P_{1:n},Q_{1:n})\right| \\
&\quad= \frac{1}{n}\left|\sup_{f\in\hilbert\mid_1}\left|\sum_{i=1}^{2n}(-1)^{\indi\{n+1\leq i \leq 2n\}}f(Z_i)\right| - \sup_{f\in\hilbert\mid_1}\left|\sum_{i=1}^{2n}(-1)^{\indi\{n+1\leq i \leq 2n\}}\E[f(Z_i)]\right|\right| \\
&\quad\leq \sup_{f\in\hilbert\mid_1}\frac{1}{n}\left|\sum_{i=1}^{2n}(-1)^{\indi\{n+1\leq i \leq 2n\}}\left(f(Z_i) - \E[f(Z_i)]\right)\right| \tag{\theequation}\label{eq:H:bound} \\
&\quad= \sup_{f\in\hilbert\mid_1}\left|\sum_{i=1}^{2n}\theta_i(f,Z_i)\right| \;.
\end{align*}

We will make repeated use of the following upper bound on the expectation of $H(W)$, the proof of which is standard, for example, \citep[][Appendix D]{Sriperumbudur_2012}.

\begin{lemma} \label{lem:H:bound}
    Let $\theta_i$ be as defined in \eqref{eq:thetai:def} and $\calF = \hilbert|_1$. Then
    \begin{align*}
        \E[H(W)] \leq 2 \supk \sqrt{\frac{2}{n}} \;.
    \end{align*}
\end{lemma}
\begin{proof}[Proof of \cref{lem:H:bound}]
For independent $Z_i,Z_i'\sim P_i$ and $Z_{n+i},Z_{n+i}'\sim Q_i$, and independent Rademacher random variables $R_i$, $i\in[n]$,
\begin{align*}
\E[H(W)] &= \E\left[\sup_{f\in\hilbert|_1}\left|\sum_{i=1}^{2n}\theta_i(f,Z_i)\right|\right] \\
&= \E\left[\sup_{f\in\hilbert|_1}\left|\frac{1}{n}\sum_{i=1}^{2n}(-1)^{\indi\{n+1\leq i \leq 2n\}}\left(f(Z_i) - \E[f(Z_i)]\right)\right|\right] \\
&\leq \E\left[\sup_{f\in\hilbert|_1}\left|\frac{1}{n}\sum_{i=1}^{2n}(-1)^{\indi\{n+1\leq i \leq 2n\}}\left(f(Z_i) - f(Z_i')\right)\right|\right] \\
&= \E\left[\sup_{f\in\hilbert|_1}\left|\frac{1}{n}\sum_{i=1}^{2n}R_i\left(f(Z_i) - f(Z_i')\right)\right|\right] \\
&\leq 2\E\left[\sup_{f\in\hilbert|_1}\left|\frac{1}{n}\sum_{i=1}^{2n}R_if(Z_i)\right|\right] \;,
\end{align*}
where the first inequality follows from Jensen's inequality.
Using RKHS identities and  the boundedness of the kernel \citep[see Appendix~D,][]{Sriperumbudur_2012} yields
\begin{align*}
\E\left[\sup_{f\in\hilbert|_1}\left|\frac{1}{n}\sum_{i=1}^{2n}R_if(Z_i)\right|\right] \leq \supk\sqrt{\frac{2n}{n^2}} = \supk\sqrt{\frac{2}{n}} \;.
\end{align*}
\end{proof}

\begin{proof}[Proof of \cref{thm:mmd:power}]
The proof proceeds in two stages.
We first derive a high-probability upper bound (over observations $Z_{1:n}$ and randomization samples $\ZRand_{1:n}^{(0)}$) for the conditional probability that a randomized test statistic is greater than the observed test statistic, conditional on $\calA_{1:n}\define (S_{1:n},Z_{1:n},\ZRand_{1:n}^{(0)})$.
We then use the probability upper bound to lower bound the test power in terms of a binomial cumulative distribution function.

Let $\Pnull_{1:n}$ be the conditional distribution of $\ZRand_{1:n}^{(0)}$ conditioned on $S_{1:n}$ under $\nullHyp$.
For $b\in[B]$, let $\ZRand_{1:n}^{(b)}$ and $\ZRand_{1:n}'^{(b)}$ denote further independent randomization samples from $\Pnull_{1:n}$.
Let $\hatDelta\define|\testStat(Z_{1:n},\ZRand^{(0)}_{1:n})|$ denote the absolute value of the observed test statistic, and let $\EH\define \E[H(W)]$.
Observe that
\begin{align*}
p(\hatDelta) &\define \bbP\left(\left.\left|\testStat(\ZRand^{(b)}_{1:n},\ZRand'^{(b)}_{1:n})\right|\geq \hatDelta \;\right|\; \calA_{1:n}\right) \\
&= \bbP\left(\left.\left|\testStat(\ZRand^{(b)}_{1:n},\ZRand'^{(b)}_{1:n})\right|\geq \hatDelta \;\right|\; \calA_{1:n},\hatDelta>\EH\right)\indi\left\{\hatDelta>\EH\right\} \\
&\quad + \bbP\left(\left.\left|\testStat(\ZRand^{(b)}_{1:n},\ZRand'^{(b)}_{1:n})\right|\geq \hatDelta \;\right|\; \calA_{1:n},\hatDelta\leq\EH\right)\indi\left\{\hatDelta\leq\EH\right\} \\
&\leq \min\left\{1, \bbP\left(\left.\left|\testStat(\ZRand^{(b)}_{1:n},\ZRand'^{(b)}_{1:n})\right|\geq \hatDelta \;\right|\; \calA_{1:n},\hatDelta>\EH\right) + \indi\left\{\hatDelta\leq\EH\right\}\right\} \;.
\end{align*}
Focusing on the first term in the upper bound, we have
\begin{align*}
p^* &\define \bbP\left(\left.\left|\testStat(\ZRand^{(b)}_{1:n},\ZRand'^{(b)}_{1:n})\right|\geq \hatDelta \;\right|\; \calA_{1:n},\hatDelta>\EH\right) \\
&= \bbP\left(\left.\left|\testStat(\ZRand^{(b)}_{1:n},\ZRand'^{(b)}_{1:n})\right| - \EH \geq \hatDelta - \EH \;\right|\; \calA_{1:n},\hatDelta>\EH\right) \\
&\leq \exp\left(-\frac{3}{2}\frac{(\hatDelta-\EH)^2}{3a + \hatDelta - \EH}\right) \\
&= \exp\left(-\frac{3}{2}\left(\hatDelta - \EH - 3a + \frac{9a^2}{3a + \hatDelta -\EH}\right)\right) \\
&\leq \min\left\{1, \exp\left(-\frac{3}{2}\hatDelta + \frac{3}{2}\EH + \frac{9}{2}a\right)\right\} \stepcounter{equation}\tag{\theequation}\label{eq:factor:trick} \\
&\leq \min\left\{1, \exp\left(-\frac{3}{2}\hatDelta + 3\supk\sqrt{\frac{2}{n}} + 72\supk^2\sqrt{\frac{2}{n^3}} + \frac{9}{n}\supk^2\right)\right\} \;.
\end{align*}
The first inequality follows from \cref{lem:talagrand} and the fact that the conditional distributions of $\ZRand_{1:n}^{(b)}$ and $\ZRand_{1:n}'^{(b)}$ are the same.
The third inequality follows from the bound $a\leq \frac{2\supk^2}{n}\left(1 + 8\sqrt{\frac{2}{n}}\right)$, which comes from \cref{lem:H:bound} and the bounds on $\varphi_i^2$ and $M$.

From \cref{lem:talagrand} conditioned on $S_{1:n}$, for $\eta>0$, we have with probability over $Z_{1:n}$ and $\ZRand_{1:n}^{(0)}$,
\stepcounter{equation}
\begin{align*}
1 - e^{-\eta} &\leq \bbP\left(\left. H(Z_{1:2n}) \leq \EH+ \sqrt{2a\eta} + \frac{2M\eta}{3} \;\right|\; S_{1:n}\right) \\
&\leq \bbP\left(\left.\left|\testStat(P_{1:n},\Pnull_{1:n}) - \testStat(Z_{1:n},\ZRand^{(0)}_{1:n})\right| \leq \EH + \sqrt{2a\eta} + \frac{2M\eta}{3}\;\right|\; S_{1:n}\right) \\
&\leq \bbP\left(\left.\testStat(P_{1:n},\Pnull_{1:n}) - \left|\testStat(Z_{1:n},\ZRand^{(0)}_{1:n})\right| \leq \EH + \sqrt{2a\eta} + \frac{2M\eta}{3}\;\right|\; S_{1:n}\right) \\
&\leq \bbP\left(\left.\EMMD - \hatDelta \leq 2\supk\sqrt{\frac{2}{n}} + \sqrt{2\eta\left(\frac{2\supk^2}{n}\left(1+8\sqrt{\frac{2}{n}}\right)\right)} + \frac{8\supk\eta}{3n}\;\right|\; S_{1:n}\right) \\
&= \bbP\left(\left.\EMMD - 2\supk\left(\sqrt{\frac{2}{n}} + \sqrt{\eta}\sqrt{\frac{1}{n}+8\sqrt{\frac{2}{n^3}}} + \frac{4\eta}{3n}\right) \leq \hatDelta\;\right|\; S_{1:n}\right) \tag{\theequation}\label{eq:deltahat:bound} \;,
\end{align*}
where the second inequality follows from \eqref{eq:H:bound}, the third inequality from the fact $\testStat(P_{1:n},\Pnull_{1:n})\geq0$, and the fourth inequality from applying the previous bounds on $\EH$, $a$, and $M$.
Applying the probabilistic bound in \eqref{eq:deltahat:bound} to the upper bound on $p^*$, we find that with probability at least $1-e^{-\eta}$ over $Z_{1:n}$ and $\ZRand_{1:n}^{(0)}$,
\begin{align*}
p^* &\leq \min\left\{1, \exp\left(-\frac{3}{2}\left(\EMMD - 2\supk\left(\sqrt{\frac{2}{n}} + \sqrt{\eta}\sqrt{\frac{1}{n}+8\sqrt{\frac{2}{n^3}}} + \frac{4\eta}{3n}\right)\right) \right.\right. \\
&\qquad\qquad \left.\left.+\; 3\supk\sqrt{\frac{2}{n}} + 72\supk^2\sqrt{\frac{2}{n^3}} + \frac{9}{n}\supk^2\right)\right\} \\
&= \min\left\{1, \exp\left(-\frac{3}{2}\EMMD + C_n(\eta)\right)\right\} \\
&=: p_\eta(\EMMD) \;,
\end{align*}
where $C_n(\eta)$ is defined as in \eqref{eq:Rn:def}.
Therefore,
\begin{align*}
p(\hatDelta) &\leq \min\left\{1, p_\eta(\EMMD) + \indi\left\{\hatDelta\leq\EH\right\}\right\} \;.
\end{align*}
By \eqref{eq:deltahat:bound} and \cref{lem:H:bound}, with the same probability,
\begin{align*}
\indi\left\{\hatDelta\leq\EH\right\} &\leq \indi\left\{\EMMD - 2\supk\left(\sqrt{\frac{2}{n}} + \sqrt{\eta}\sqrt{\frac{1}{n}+8\sqrt{\frac{2}{n^3}}} + \frac{4\eta}{3n}\right)\leq 2\supk\sqrt{\frac{2}{n}}\right\} \\
&= \indi\left\{\EMMD \leq 2\supk\left(2\sqrt{\frac{2}{n}} + \sqrt{\eta}\sqrt{\frac{1}{n}+8\sqrt{\frac{2}{n^3}}} + \frac{4\eta}{3n}\right)\right\} \;,
\end{align*}
which equals zero under assumption \eqref{eq:emmd:condition}.
Hence, $p(\hatDelta) \leq p_\eta(\EMMD)$ with probability at least $1-e^{-\eta}$.

We now argue that the test power is lower bounded by a binomial CDF using the upper bound on $p(\hatDelta)$.
Conditioned on $S_{1:n}$, $N_B\define p_B(B+1)-1$ is a Binomial$(B,p(\hatDelta))$ random variable with
\begin{align*}
\bbP(p_B\leq \sig \mid S_{1:n}) &= \bbP(N_B\leq \sig(B+1)-1 \mid S_{1:n}) \\
&= F_{B,p(\hatDelta)}(\lfloor\sig(B+1)-1\rfloor) \;,
\end{align*}
where $F_{B,p(\hatDelta)}$ is the CDF of the Binomial$(B,p(\hatDelta))$ distribution.
It is known that for $q\geq p$, a Binomial$(B,q)$-distributed random variable $X_{B,q}$ stochastically dominates a Binomial$(B,p)$-distributed random variable $X_{B,p}$ in the sense that for all $\ell\in\{0,1,\dotsc,B\}$,
\[
\bbP(X_{B,q}\geq \ell) \geq \bbP(X_{B,p}\geq \ell) \;.
\]
The inequality is reversed for the CDFs, so that the bound $p(\hatDelta)\leq p_\eta(\EMMD)$ implies
\begin{align*}
\bbP(p_B \leq \sig \mid S_{1:n}) &\geq F_{B,p_\eta(\EMMD)}(\lfloor\sig(B+1)-1\rfloor) \;.
\end{align*}
Thus, we have showed that with probability at least $1-e^{-\eta}$,
\begin{align*}
\bbP(p_B\leq \sig\mid S_{1:n}) \geq \sum_{\ell=0}^{\lfloor\sig(B+1)-1\rfloor}\binom{B}{\ell}p_\eta(\EMMD)^\ell(1-p_\eta(\EMMD))^{B-\ell} \;.
\end{align*}

\end{proof}

\subsubsection{Lower bound for high-probability rejection of alternatives}
\label{apx:lower:bound:alternative}

Using the Chernoff bound on the upper tail of deviations of a binomial random variable, we can obtain an explicit lower bound on values of $\Delta_n$ that correspond to alternatives for which the test (with high probability) rejects the null with probability at least $1-\beta $ for $\beta\in(0,1)$. 
For convenience, let $B_{\sig} \coloneq \lfloor \sig (B+1) - 1\rfloor$.

\begin{corollary} \label{cor:power:lower:bound}
    Assume the conditions of \cref{thm:mmd:power} and that $B_{\sig} \geq B p_{\eta}(\Delta_n)$.  Fix any $\beta \in (0,1)$. For each $\eta > 0$, if
    \begin{align*}
         \Delta_n \geq \frac{2}{3}\left( C_n(\eta) - \ln \frac{B_{\sig}}{B} - \ln\left(1-\sqrt{\frac{-2\ln\beta}{B_{\sig}}} \right) \right) \;,
    \end{align*}
    then for almost every sequence $(S_i)_{i\geq 1}$, conditionally on $S_{1:n}$, we have that $\bbP(p_B \leq \sig | S_{1:n}) \geq 1-\beta$ with probability over $(Z_{1:n},{Z}^{(0)}_{1:n})$ at least $1-e^{-\eta}$.
\end{corollary}

In the limit of $B,n\to\infty$, we see that if $\text{MMD}(P,P^{(0)}) \geq -\frac{2}{3}\ln \sig$, then the test will reject with probability at least $1-\beta$. 

\begin{proof}[Proof of \cref{cor:power:lower:bound}]
Let $B_{\eta}(\Delta_n) \define B p_{\eta}(\Delta_n)$.
Applying a Chernoff bound on the upper tail of deviations of a binomial random variable to $N_B\sim\mathrm{Binomial}(B,p_\eta(\EMMD))$,
\begin{align*}
\bbP(N_B\geq B_\sig \mid S_{1:n}) &= \bbP\left(\left.N_B\geq \left(\frac{B_\sig}{B_\eta(\EMMD)}\right) B_\eta(\EMMD) \;\right|\; S_{1:n}\right) \\
&\leq \exp\left\{-B_\eta(\EMMD)\left(\frac{B_\sig}{B_\eta(\EMMD)}\ln\left(\frac{B_\sig}{B_\eta(\EMMD)}\right) - \frac{B_\sig}{B_\eta(\EMMD)} + 1\right)\right\} \\
&= \exp\left(-B_\sig\ln\left(\frac{B_\sig}{B_\eta(\EMMD)}\right) + B_\sig - B_\eta(\EMMD)\right) \;.
\end{align*}
This bound is at most $\beta$ if and only if
\begin{align*}
B_\sig\ln\left(\frac{B_\sig}{B_\eta(\EMMD)}\right) - (B_\sig - B_\eta(\EMMD)) \geq -\ln\beta \;.
\end{align*}
Note that
\begin{align*}
\ln\left(\frac{B_\sig}{B_\eta(\EMMD)}\right)
= -\ln\left(1 - \frac{B_\sig - B_\eta(\EMMD)}{B_\sig}\right)
= \sum_{\ell=1}^\infty\frac{1}{\ell}\left(\frac{B_\sig - B_\eta(\EMMD)}{B_\sig}\right)^\ell \;,
\end{align*}
where the last equality follows from the Taylor expansion of $-\ln(1-x)$ for $|x|<1$, and so an equivalent condition is
\begin{align*}
B_\sig\sum_{\ell=2}^\infty\frac{1}{\ell}\left(\frac{B_\sig - B_\eta(\EMMD)}{B_\sig}\right)^\ell \geq -\ln\beta \;.
\end{align*}
The condition is satisfied if
\begin{align*}
B_\sig\frac{1}{2}\left(\frac{B_\sig - B_\eta(\EMMD)}{B_\sig}\right)^2 \geq -\ln\beta \;,
\end{align*}
or equivalently,
\begin{align*}
B_\eta(\EMMD) \leq B_\sig\left(1 - \sqrt{-\frac{2\ln\beta}{B_\sig}}\right) \;.
\end{align*}
Substituting in the definition for $B_\eta(\EMMD)$ and rearranging for $\EMMD$ yields
\begin{align*}
\EMMD \geq \frac{2}{3}\left(C_n(\eta) - \ln\left(\frac{B_\sig}{B}\right) - \ln\left(1 - \sqrt{-\frac{2\ln\beta}{B_\sig}}\right)\right) \;.
\end{align*}
    
\end{proof}

\subsection{Power lower bounds for adaptive statistics}
\label{apx:proofs:adaptive}

As with \cref{thm:mmd:power}, the conditional probability of rejecting when using the FUSE statistic can be lower bounded in terms of
\begin{align*}
    \barDelta \define \frac{1}{L} \sum_{\ell = 1}^L \sqrt{\E \left[\left. \widehat{\textrm{MMD}}^2_{v,\kparam_\ell}(Z_{1:n}, \ZRand^{(0)}_{1:n} ) \;\right|\; S_{1:n} \right]} \;.
\end{align*}
For convenience, let $C_n^{\vee}(\eta)$ denote the maximum over $\ell \in [L]$ of $C_{n,\ell}(\eta)$, which is obtained by substituting $\supk_{\ell}$ for $\supk$ into \eqref{eq:Rn:def}.
Define $\nusup \define \sup_\ell \supk_\ell$.

\begin{proposition} \label{prop:fuse:power}
    Let $\testStat$ be the FUSE statistic using $\widehat{\textrm{MMD}}^2_{v,\kparam_{\ell}}$ and exact randomization in the CRT.
    For any $\eta>0$, if
\begin{align} \label{eq:barDelta:condition}
\barDelta > 2\nusup\left(2\sqrt{\frac{2}{n}} + \sqrt{\eta}\sqrt{\frac{1}{n}+8\sqrt{\frac{2}{n^3}}} + \frac{4\eta}{3n} \right) \;,
\end{align}
then for almost every sequence $S_{1:n}$, conditionally on $S_{1:n}$, with probability at least $\max\{0,1-Le^{-\eta}\}$ over $(Z_{1:n},{Z}^{(0)}_{1:n})$,
    \begin{align*}
        \bbP(p_B \leq \sig \mid S_{1:n}) \geq F_{B,p_{\eta}(\bar{\Delta}_n)}(\lfloor \sig (B+1) - 1 \rfloor) \;,
    \end{align*}
    where $F_{B,p_{\eta}(\Delta_n^{\vee})}$ is the distribution function for the binomial distribution of $B$ trials with success probability
    \begin{align*}
        p_{\eta}(\bar{\Delta}_n) = \min\left\{ 1, L\exp\left( -\frac{3}{2}\bar{\Delta}_n + C_n^{\vee}(\eta) \right) \right\} \;.
    \end{align*}
\end{proposition}

Using this, a lower bound for $\bar{\Delta}_n$ that guarantees power $1-\beta$ with high probability, analogous to the one in \cref{cor:power:lower:bound}, can be obtained. 

For the following proof, define
\begin{align} \label{eq:supa}
a^\lor \define \sup_{\ell\in[L]} a_\ell \leq \nusup^2\left(\frac{2}{n}+16\sqrt{\frac{2}{n^3}}\right)
\end{align}
and similarly 
\begin{align} \label{eq:supH}
\EH^\lor \define \sup_{\ell\in[L]} \EHell \leq 2\nusup\sqrt{\frac{2}{n}} \;.
\end{align}

\begin{proof}[Proof of \cref{prop:fuse:power}]
Let
\begin{align*}
\testStat(Z_{1:n},Z_{1:n}') &\define \widehat{\mathrm{FUSE}}_\omega(Z_{1:n},Z_{1:n}') \\
&= \frac{1}{\omega}\log\left(\frac{1}{L}\sum_{\ell=1}^L\exp\left(\omega\widehat{\mathrm{MMD}}^2_{v,\sigma_\ell}(Z_{1:n},Z_{1:n}')\right)\right) \;,
\end{align*}
and let $\hatTestStat\define \testStat(Z_{1:n},\ZRand_{1:n}^{(0)})$ denote the observed test statistic.
Denote $\calA_{1:n}\define (S_{1:n},Z_{1:n},\ZRand_{1:n}^{(0)})$.
Observe that
\begin{align*}
p(\hatTestStat) &\define \bbP\left(\left.\testStat(\ZRand_{1:n}^{(b)},\ZRand_{1:n}'^{(b)}) \geq \hatTestStat \;\right|\; \calA_{1:n} \right) \\
&\leq \bbP\left\{\left.\frac{1}{\omega}\log\left(\sup_{\ell\in[L]}\exp\left(\omega\widehat{\mathrm{MMD}}^2_{v,\kparam_\ell}(\ZRand_{1:n}^{(b)},\ZRand_{1:n}'^{(b)})\right)\right) \geq \hatTestStat \;\right|\; \calA_{1:n} \right\} \\
&= \bbP\left(\left.\sup_{\ell\in[L]} \widehat{\mathrm{MMD}}^2_{v,\kparam_\ell}(\ZRand_{1:n}^{(b)},\ZRand_{1:n}'^{(b)}) \geq \hatTestStat \;\right|\; \calA_{1:n} \right) \\
&\leq \bbP\left(\left.\sup_{\ell\in[L]} \left|\widehat{\mathrm{MMD}}_{v,\kparam_\ell}(\ZRand_{1:n}^{(b)},\ZRand_{1:n}'^{(b)})\right| \geq \sqrt{\hatTestStat} \;\right|\; \calA_{1:n} \right) \\
&= \bbP\left(\left.\sup_{\ell\in[L]} \left|\widehat{\mathrm{MMD}}_{v,\kparam_\ell}(\ZRand_{1:n}^{(b)},\ZRand_{1:n}'^{(b)})\right| - \EH^\lor \geq \sqrt{\hatTestStat}  - \EH^\lor \;\right|\; \calA_{1:n} \right) \\
&\leq \min\left\{1, \bbP\left(\left.\sup_{\ell\in[L]} \left|\widehat{\mathrm{MMD}}_{v,\kparam_\ell}(\ZRand_{1:n}^{(b)},\ZRand_{1:n}'^{(b)})\right| - \EH^\lor \geq \sqrt{\hatTestStat} - \EH^\lor \;\right|\; \calA_{1:n}, \sqrt{\hatTestStat} > \EH^\lor \right) \right. \\
&\qquad\qquad + \indi\left\{\sqrt{\hatTestStat} \leq \EH^\lor\right\} \Bigg\} \;,
\end{align*}
where the last inequality follows from the same argument as that in \cref{thm:mmd:power}.
We first derive a high-probability upper bound for the first term.
Applying a Fr\'echet inequality (first inequality), \cref{lem:talagrand} (third inequality), the same factoring trick as in \eqref{eq:factor:trick} (fourth inequality), and Jensen's inequality (last two inequalities), we obtain
\begin{align*}
p^* &\define \bbP\left(\left.\sup_{\ell\in[L]} \left|\widehat{\mathrm{MMD}}_{v,\kparam_\ell}(\ZRand_{1:n}^{(b)},\ZRand_{1:n}'^{(b)})\right| - \EH^\lor \geq \sqrt{\hatTestStat} - \EH^\lor \;\right|\; \calA_{1:n}, \sqrt{\hatTestStat} > \EH^\lor \right) \\
&\leq \min\left\{1, \sum_{\ell=1}^L \bbP\left(\left.\left|\widehat{\mathrm{MMD}}_{v,\sigma_\ell}(\ZRand_{1:n}^{(b)},\ZRand_{1:n}'^{(b)})\right| - \EH^\lor \geq \sqrt{\hatTestStat} - \EH^\lor \;\right|\; \calA_{1:n}, \sqrt{\hatTestStat} > \EH^\lor \right) \right\} \\
&\leq \min\left\{1, \sum_{\ell=1}^L \bbP\left(\left.\left|\widehat{\mathrm{MMD}}_{v,\sigma_\ell}(\ZRand_{1:n}^{(b)},\ZRand_{1:n}'^{(b)})\right| - \EHell \geq \sqrt{\hatTestStat} - \EH^\lor \;\right|\; \calA_{1:n}, \sqrt{\hatTestStat} > \EH^\lor \right) \right\} \\
&\leq \min\left\{1, \sum_{\ell=1}^L \exp\left(-\frac{3}{2} \frac{\left(\sqrt{\hatTestStat} - \EH^\lor\right)^2}{3a_\ell + \sqrt{\hatTestStat} - \EH^\lor}\right) \right\} \\
&\leq \min\left\{1, \sum_{\ell=1}^L \exp\left(-\frac{3}{2}\sqrt{\hatTestStat} + \frac{3}{2}\EH^\lor + \frac{9}{2}a_\ell\right)\right\} \\
&\leq \min\left\{1, L\exp\left(-\frac{3}{2}\sqrt{\hatTestStat} + \frac{3}{2}\EH^\lor + \frac{9}{2}a^\lor\right) \right\} \\
&\leq \min\left\{1, L\exp\left(-\frac{3}{2}\sqrt{\frac{1}{L}\sum_{\ell=1}^L \widehat{\mathrm{MMD}}^2_{v,\kparam_\ell}(Z_{1:n},\ZRand_{1:n}^{(0)})} + \frac{3}{2}\EH^\lor + \frac{9}{2}a^\lor\right) \right\} \\
&\leq \min\left\{1, L\exp\left(-\frac{3}{2L}\sum_{\ell=1}^L \left|\widehat{\mathrm{MMD}}_{v,\kparam_\ell}(Z_{1:n},\ZRand_{1:n}^{(0)})\right| + \frac{3}{2}\EH^\lor + \frac{9}{2}a^\lor\right) \right\} \;.
\end{align*}
Applying a union bound with \eqref{eq:deltahat:bound} and the bounds on $a^\lor$ \eqref{eq:supa} and $\EH^\lor$ \eqref{eq:supH}, for any $\eta>0$, with probability at least $\max\{0,1-Le^{-\eta}\}$ over $Z_{1:n}$ and $\ZRand_{1:n}^{(0)}$,
\begin{align*}
&\exp\left(-\frac{3}{2L}\sum_{\ell=1}^L \left|\widehat{\mathrm{MMD}}_{v,\kparam_\ell}(Z_{1:n},\ZRand_{1:n}^{(0)})\right| + \frac{3}{2}\EH^\lor + \frac{9}{2}a^\lor\right) \\
&\quad\leq \exp\left\{-\frac{3}{2L}\sum_{\ell=1}^L \left(\Delta_{n,\ell} - 2\supk_\ell\left(\sqrt{\frac{2}{n}} + \sqrt{\eta}\sqrt{\frac{1}{n} + 8\sqrt{\frac{2}{n^3}}} + \frac{4\eta}{3n}\right)\right) \right. \\
&\quad\qquad + \left.\frac{3}{2}\left(2\nusup\sqrt{\frac{2}{n}}\right) + \frac{9}{2}\nusup^2\left(\frac{2}{n}+16\sqrt{\frac{2}{n^3}}\right)\right\} \\
&\quad\leq \exp\left\{-\frac{3}{2}\barDelta + \nusup\left(6\sqrt{\frac{2}{n}} + 3\sqrt{\eta}\sqrt{\frac{1}{n} + 8\sqrt{\frac{2}{n^3}}} + \frac{4\eta}{n}\right) \right. \\
&\quad\qquad \left.+ \nusup^2\left(\frac{9}{n}+72\sqrt{\frac{2}{n^3}}\right)\right\} \\
&\quad= \exp\left(-\frac{3}{2}\barDelta + C_n^\lor(\eta)\right) \;.
\end{align*}
Hence,
\begin{align*}
p^* \leq \min\left\{1, L\exp\left(-\frac{3}{2}\barDelta + C_n^\lor(\eta)\right) \right\}  =: p_\eta(\barDelta) \;.
\end{align*}

For the indicator term in the upper bound of $p(\hatTestStat)$, by \eqref{eq:deltahat:bound} and \eqref{eq:supH}, with the same probability,
\begin{align*}
\indi\left\{ \sqrt{\hatTestStat} \leq \EH^\lor \right\} &\leq \indi\left\{ \sqrt{\frac{1}{L}\sum_{\ell=1}^L\widehat{\mathrm{MMD}}_{v,\kparam_\ell}^2(Z_{1:n},\ZRand_{1:n}^{(0)})} \leq \EH^\lor \right\} \\
&\leq \indi\left\{ \frac{1}{L}\sum_{\ell=1}^L\left|\widehat{\mathrm{MMD}}_{v,\kparam_\ell}(Z_{1:n},\ZRand_{1:n}^{(0)})\right| \leq \EH^\lor \right\} \\
&\leq \indi\left\{ \frac{1}{L}\sum_{\ell=1}^L\left(\Delta_{n,\ell} - 2\supk_\ell\left(\sqrt{\frac{2}{n}} + \sqrt{\eta}\sqrt{\frac{1}{n}+8\sqrt{\frac{2}{n^3}}} + \frac{4\eta}{3n}\right) \right) \leq \EH^\lor \right\} \\
&\leq \indi\left\{ \barDelta \leq 2\nusup\left(2\sqrt{\frac{2}{n}} + \sqrt{\eta}\sqrt{\frac{1}{n}+8\sqrt{\frac{2}{n^3}}} + \frac{4\eta}{3n} \right) \right\} \;,
\end{align*}
which is zero under assumption \eqref{eq:barDelta:condition}.
Thus, $p(\hatTestStat) \leq p_\eta(\barDelta)$ with probability at least $\max\{0,1-Le^{-\eta}\}$.

The rest of the proof is analogous to the second part of \cref{thm:mmd:power}.

\end{proof}

When using the supremum kernel as the test statistic, we expect that when 
\begin{align*}
\EMMD^\lor \define \sup_{\ell \in [L]} \Delta_{n,\ell} = \sup_{\ell \in [L]} \sqrt{\E\left[ \left.\widehat{\mathrm{MMD}}_{v,\kparam_\ell}^2(Z_{1:n},\ZRand_{1:n}^{(0)}) \;\right|\; S_{1:n} \right]}
\end{align*}
is large, then the test should reject with high probability. 

\begin{proposition} \label{prop:sk:power}
    Let $T$ be the supremum kernel distance using $\widehat{\textrm{MMD}}^2_{v,\kparam_{\ell}}$ and exact randomization in the CRT.
    For any $\eta>0$, if
\begin{align} \label{eq:supDelta:condition}
\supDelta > 2\nusup\left(2\sqrt{\frac{2}{n}} + \sqrt{\eta}\sqrt{\frac{1}{n}+8\sqrt{\frac{2}{n^3}}} + \frac{4\eta}{3n} \right) \;,
\end{align}
    then for almost every sequence $S_{1:n}$, conditionally on $S_{1:n}$, with probability at least $\max\{0,1-Le^{-\eta}\}$ over $(Z_{1:n},{Z}^{(0)}_{1:n})$,
    \begin{align*}
        \bbP(p_B \leq \sig \mid S_{1:n}) \geq F_{B,p_{\eta}(\supDelta)}(\lfloor \sig (B+1) - 1 \rfloor) \;,
    \end{align*}
    where $F_{B,p_{\eta}(\supDelta)}$ is the distribution function for the binomial distribution of $B$ trials with success probability
    \begin{align*}
        p_{\eta}(\supDelta) = \min\left\{ 1, L\exp\left( -\frac{3}{2}\supDelta + C_n^{\vee}(\eta) \right) \right\} \;.
    \end{align*}
\end{proposition}

Using this, a lower bound for $\Delta_n^{\vee}$ that guarantees power $1-\beta$ with high probability, analogous to the one in \cref{cor:power:lower:bound}, can be obtained. 

\begin{proof}[Proof of \cref{prop:sk:power}]
Let
\begin{align*}
\testStat(Z_{1:n},Z_{1:n}') &\define \widehat{\mathrm{SK}}(Z_{1:n},Z_{1:n}') \\
&= \sup_{\ell\in[L]} \widehat{\mathrm{MMD}}^2_{v,\kparam_\ell}(Z_{1:n},Z_{1:n}') \\
&= \sup_{\ell\in[L]} \left|\widehat{\mathrm{MMD}}_{v,\kparam_\ell}(Z_{1:n},Z_{1:n}')\right|^2 \;,
\end{align*}
and let $\hatTestStat\define \testStat(Z_{1:n},\ZRand_{1:n}^{(0)})$ denote the observed test statistic.
Denote $\calA_{1:n}\define (S_{1:n},Z_{1:n},\ZRand_{1:n}^{(0)})$.
Observe that
\begin{align*}
p(\hatTestStat) &\define \bbP\left(\left.\testStat(\ZRand_{1:n}^{(b)},\ZRand_{1:n}'^{(b)}) \geq \hatTestStat \;\right|\; \calA_{1:n} \right) \\
&= \bbP\left(\left.\sup_{\ell\in[L]} \left|\widehat{\mathrm{MMD}}_{v,\kparam_\ell}(\ZRand_{1:n}^{(b)},\ZRand_{1:n}'^{(b)})\right| \geq \sqrt{\hatTestStat} \;\right|\; \calA_{1:n} \right) \\
&= \bbP\left(\left.\sup_{\ell\in[L]} \left|\widehat{\mathrm{MMD}}_{v,\kparam_\ell}(\ZRand_{1:n}^{(b)},\ZRand_{1:n}'^{(b)})\right| - \EH^\lor \geq \sqrt{\hatTestStat}  - \EH^\lor \;\right|\; \calA_{1:n} \right) \\
&\leq \min\left\{1, \bbP\left(\left.\sup_{\ell\in[L]} \left|\widehat{\mathrm{MMD}}_{v,\kparam_\ell}(\ZRand_{1:n}^{(b)},\ZRand_{1:n}'^{(b)})\right| - \EH^\lor \geq \sqrt{\hatTestStat}  - \EH^\lor \;\right|\; \calA_{1:n}, \sqrt{\hatTestStat} > \EH^\lor \right) \right. \\
&\qquad\qquad + \indi\left\{\sqrt{\hatTestStat} \leq \EH^\lor\right\}\Bigg\} \;,
\end{align*}
where the last inequality follows from the same argument as that in \cref{thm:mmd:power}.
We first derive a high-probability upper bound for the first term.
Applying a Fr\'{e}chet inequality (first inequality), \cref{lem:talagrand} (third inequality), and the same factoring trick as in \eqref{eq:factor:trick} (fourth inequality),
\begin{align*}
p^* &\define \bbP\left(\left.\sup_{\ell\in[L]} \left|\widehat{\mathrm{MMD}}_{v,\kparam_\ell}(\ZRand_{1:n}^{(b)},\ZRand_{1:n}'^{(b)})\right| - \EH^\lor \geq \sqrt{\hatTestStat}  - \EH^\lor \;\right|\; \calA_{1:n}, \sqrt{\hatTestStat} > \EH^\lor \right) \\
&\leq \min\left\{1, \sum_{\ell=1}^L\bbP\left(\left.\left|\widehat{\mathrm{MMD}}_{v,\kparam_\ell}(\ZRand_{1:n}^{(b)},\ZRand_{1:n}'^{(b)})\right| - \EH^\lor \geq \sqrt{\hatTestStat} - \EH^\lor \;\right|\; \calA_{1:n}, \sqrt{\hatTestStat} > \EH^\lor\right) \right\} \\
&\leq \min\left\{1, \sum_{\ell=1}^L\bbP\left(\left.\left|\widehat{\mathrm{MMD}}_{v,\kparam_\ell}(\ZRand_{1:n}^{(b)},\ZRand_{1:n}'^{(b)})\right| - \EHell \geq \sqrt{\hatTestStat} - \EH^\lor \;\right|\; \calA_{1:n}, \sqrt{\hatTestStat} > \EH^\lor\right) \right\} \\
&\leq \min\left\{1, \sum_{\ell=1}^L\exp\left(-\frac{3}{2}\frac{\left(\sqrt{\hatTestStat} - \bar{H}_{n,\ell}\right)^2}{3a_\ell + \sqrt{\hatTestStat} - \EH^\lor}\right)\right\} \\
&\leq \min\left\{1, \sum_{\ell=1}^L\exp\left(-\frac{3}{2}\sqrt{\hatTestStat} + \frac{3}{2}\EH^\lor + \frac{9}{2}a_\ell \right)\right\} \\
&\leq \min\left\{1, L\exp\left(-\frac{3}{2}\sqrt{\hatTestStat} + \frac{3}{2}\EH^\lor + \frac{9}{2}a^\lor \right)\right\} \;.
\end{align*}
Applying a union bound with \eqref{eq:deltahat:bound} and the bounds on $a^\lor$ \eqref{eq:supa} and $\EH^\lor$ \eqref{eq:supH}, for any $\eta>0$, with probability at least $\max\{0,1-Le^{-\eta}\}$ over $Z_{1:n}$ and $\ZRand_{1:n}^{(0)}$,
\begin{align*}
&\exp\left(-\frac{3}{2}\sqrt{\hatTestStat}  + \frac{3}{2}\EH^\lor + \frac{9}{2}a^\lor \right) \\
&\quad\leq \exp\left\{-\frac{3}{2}\left(\supDelta - 2\nusup\left(\sqrt{\frac{2}{n}} + \sqrt{\eta}\sqrt{\frac{1}{n}+8\sqrt{\frac{2}{n^3}}} + \frac{4\eta}{3n}\right)\right) \right.\\
&\quad\qquad\qquad \left. +\frac{3}{2}\left(2\nusup\sqrt{\frac{2}{n}}\right) + \frac{9}{2}\nusup^2\left(\frac{2}{n}+16\sqrt{\frac{2}{n^3}}\right) \right\} \\
&\quad= \exp\left(-\frac{3}{2}\supDelta + C_n^\lor(\eta)\right) \;.
\end{align*}
Hence,
\begin{align*}
p^* \leq \min\left\{1, L\exp\left(-\frac{3}{2}\supDelta + C_n^\lor(\eta)\right) \right\}  =: p_\eta(\supDelta) \;.
\end{align*}

For the indicator term in the upper bound of $p(\hatTestStat)$, by \eqref{eq:deltahat:bound} and \eqref{eq:supH}, with the same probability,
\begin{align*}
\indi\left\{\sqrt{\hatTestStat} \leq \EH^\lor \right\} &= \indi\left\{ \sup_{\ell\in[L]}\left|\widehat{\mathrm{MMD}}_{v,\sigma_\ell}(Z_{1:n},\ZRand_{1:n}^{(0)})\right| \leq \EH^\lor \right\} \\
&\leq \indi\left\{\supDelta - 2\supk_\ell\left(\sqrt{\frac{2}{n}} + \sqrt{\eta}\sqrt{\frac{1}{n}+8\sqrt{\frac{2}{n^3}}} + \frac{4\eta}{3n}\right) \leq \EH^\lor \right\} \\
&\leq \indi\left\{\supDelta \leq 2\supk_\ell\left(2\sqrt{\frac{2}{n}} + \sqrt{\eta}\sqrt{\frac{1}{n}+8\sqrt{\frac{2}{n^3}}} + \frac{4\eta}{3n}\right) \right\}  \;,
\end{align*}
which is zero under assumption \eqref{eq:supDelta:condition}.
Thus, $p(\hatTestStat) \leq p_\eta(\supDelta)$ with probability at least $\max\{0,1-Le^{-\eta}\}$.

The rest of the proof is analogous to the second part of \cref{thm:mmd:power}.

\end{proof}

\section{Further considerations for randomization tests for conditional symmetry}

\subsection{Randomization using representative inversions with non-free group actions}\label{apx:rand:inv:kern}

As demonstrated in \cref{expl:SOd:equivariance}, a representative inversion $\repInv$ may be (non-uniquely) defined even when the group action is not free in some cases.
In these cases, a test for conditional symmetry that uses the deterministic representative inversion is still a valid level-$\sig$ test under the null.
This follows from the observation that for a random action $\repInvRand|X \equdist \repInv(X)H | X$, where $H$ is a random element of the stabilizer subgroup $\grp_X$ of $X$, assuming that $Y|X$ is $\grp$-equivariant,
\begin{align*}
P_{Y|X}(\repInvRand^{-1}Y \in B\mid x) &= P_{Y|X}((\repInv(x)H)^{-1}Y \in B\mid x) \\
&= P_{Y|X}(\repInv(x)^{-1}Y \in B\mid H^{-1}x) = P_{Y|X}(\repInv(x)^{-1}Y \in B \mid x) \;.
\end{align*}
That is, the conditional distributions of $\repInvRand^{-1}Y|X$ and $\repInv(X)^{-1}Y|X$ are the same.
A test that uses the representative inversion may be more computationally efficient than one that samples from the inversion kernel.
However, tests will generally have different properties under alternatives.
We explore such differences empirically in \cref{apx:exp:rep:inv}.

\subsection{Exact conditional randomization} \label{apx:rand:cond:exact}

Here, we describe two possible assumptions for $P_{\G|\maxInv}$ that each allow for sampling from the exact conditional distribution.

\begin{assumption}\label{asp:group:orbit:indep}
$P_{\G|\maxInv} = P_\G$ for all $\maxInv(X)$.
\end{assumption}

\Cref{asp:group:orbit:indep} says that the distribution over the group action is the same across orbits.
Generating a new sample $\GRand_1,\dotsc,\GRand_n$ can then be done by simply permuting $\repInvRand_1,\dotsc,\repInvRand_n$.
The resulting randomization procedure using \cref{prp:eq:xy} is outlined in \cref{alg:rand:g:asp1}. We note that this assumption is trivially true if $\grp$ acts transitively on $\bfX$. 

\begin{algorithm}[htb]
\caption{$\GRand$ randomization procedure under \cref{asp:group:orbit:indep}}\label{alg:rand:g:asp1}
\begin{algorithmic}[1]
\State sample $\repInvRand_1, \dotsc, \repInvRand_n$ with $\repInvRand_i|X_i\sim\repInvKern$
\State sample permutation $\repInvRand_{\pi(1)}, \dotsc, \repInvRand_{\pi(n)}$
\State compute $\left((\repInvRand_{\pi(1)}\orbSel(X_1),\repInvRand_{\pi(1)}\tY_1),\dotsc,(\repInvRand_{\pi(n)}\orbSel(X_n),\repInvRand_{\pi(n)}\tY_n)\right)$
\end{algorithmic}
\end{algorithm}

\begin{assumption}\label{asp:group:equiv}
Suppose there is a group $\bfH$ acting transitively on $\bfM$, and let $\phi\colon\bfH\to\grp$ be a group homomorphism. Then $P_{\G|\maxInv}$ is $\bfH/\grp$-equivariant in the sense that for all $h\in\bfH$,
\[
P_{\G|\maxInv}(\argdot\mid h\maxInv(X)) = P_{\G|\maxInv}(\phi(h)^{-1}\argdot\mid \maxInv(X)) \;.
\]
\end{assumption}

\Cref{asp:group:equiv} says that the conditional distribution over the group action is itself equivariant with respect to the action of $\bfH$ on $\bfM$.
Assuming that $\phi$ is known, new samples can be generated by permuting $\repInvRand_1,\dotsc,\repInvRand_n$ and applying a corresponding transform under $\phi$.
The corresponding randomization procedure using \cref{prp:eq:xy} is outlined in \cref{alg:rand:g:asp2}.

\begin{algorithm}[htb]
\caption{$\GRand$ randomization procedure under \cref{asp:group:equiv}}\label{alg:rand:g:asp2}
\begin{algorithmic}[1]
\State compute $\repInvRand_1, \dotsc, \repInvRand_n$ with $\repInvRand_i|X_i\sim\repInvKern$
\State sample permutation $\repInvRand_{\pi(1)}, \dotsc, \repInvRand_{\pi(n)}$
\State sample $\tilde{h}_{\pi(1)},\dotsc,\tilde{h}_{\pi(n)}$ where $\phi(\tilde{h}_{\pi(i)})^{-1}\orbSel(X_{\pi(i)})\equas \orbSel(X_i)$
\State compute $\left((\phi(\tilde{h}_{\pi(i)})^{-1}\repInvRand_{\pi(i)}\orbSel(X_i),\phi(\tilde{h}_{\pi(i)})^{-1}\repInvRand_{\pi(i)}\tY_i)\right)_{i=1}^n$
\end{algorithmic}
\end{algorithm}

\section{Additional experiments and experimental details} \label{apx:exp}

\subsection{Using independent comparison sets}
\label{apx:exp:mvn}

For the experiment discussed in \cref{sec:exp:gauss}, \cref{fig:gauss:equiv:cov:all} shows the estimated rejection rates for the tests that use independent comparison sets. 
We observe that the results for the tests that use independent comparison sets are similar to those of tests that reuse a single comparison set.

\begin{figure}[tbp]
    \centering
    \includegraphics[trim={29 26 0 0},clip,width=\textwidth]{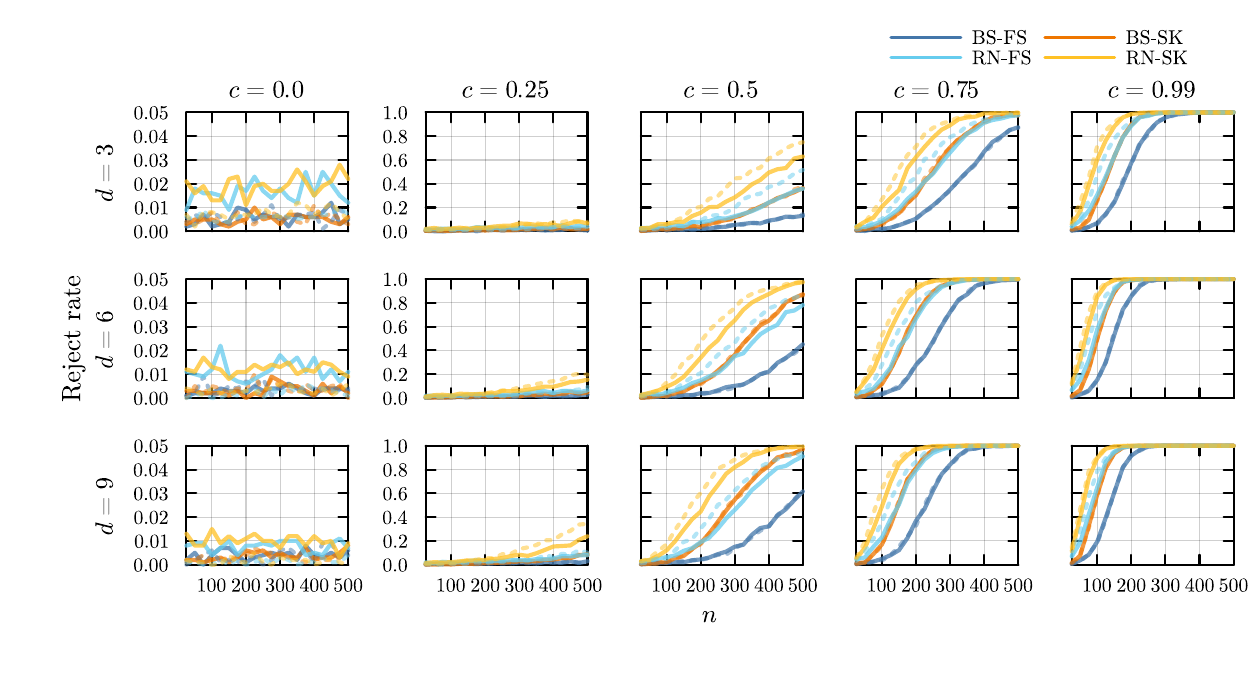}
    \caption{Rejection rates for tests for $\SO{d}$-equivariance that use independent comparison sets in the changing covariance experiment for varying values of the $Y|X$ covariance matrix off-diagonal values $c$ and sample size $n$.
    Dotted lines show the corresponding rates for tests that reuse a comparison set.  (Note the different vertical scale for $c = 0$ versus $c > 0$.)}
    \label{fig:gauss:equiv:cov:all}
\end{figure}

\subsection{Approximation via kernel density estimation}\label{apx:exp:bw}

\begin{figure}[tbp]
    \centering
    \includegraphics[trim={26 24 0 13},clip,width=\textwidth]{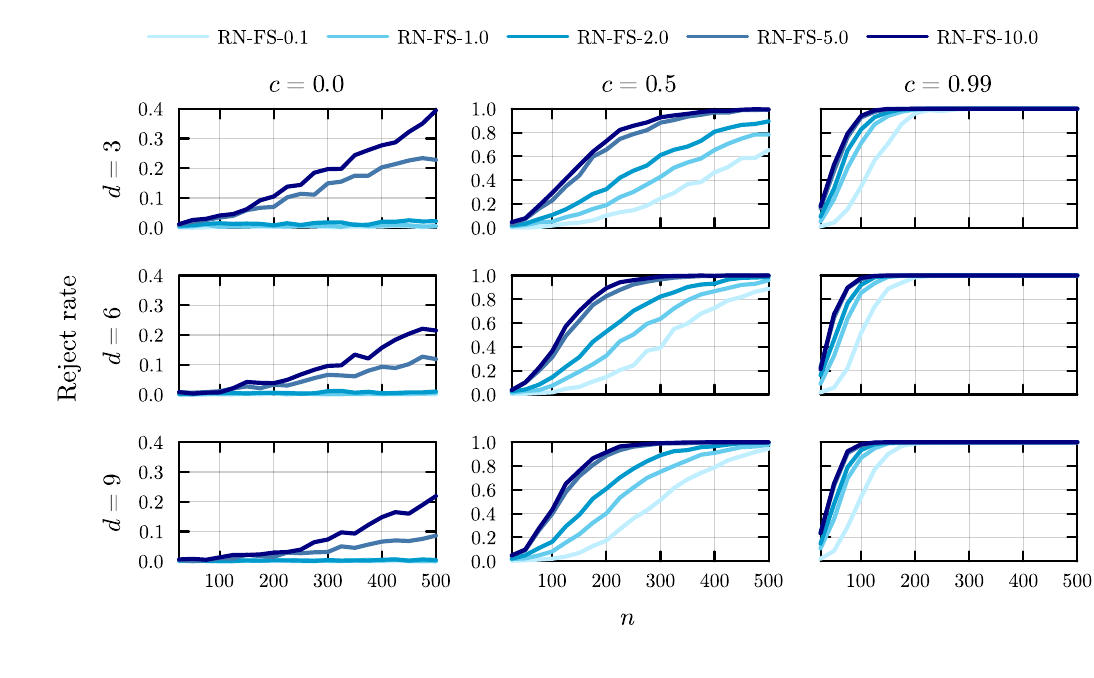}
    \vspace{1em}

    \includegraphics[trim={26 24 0 13},clip,width=\textwidth]{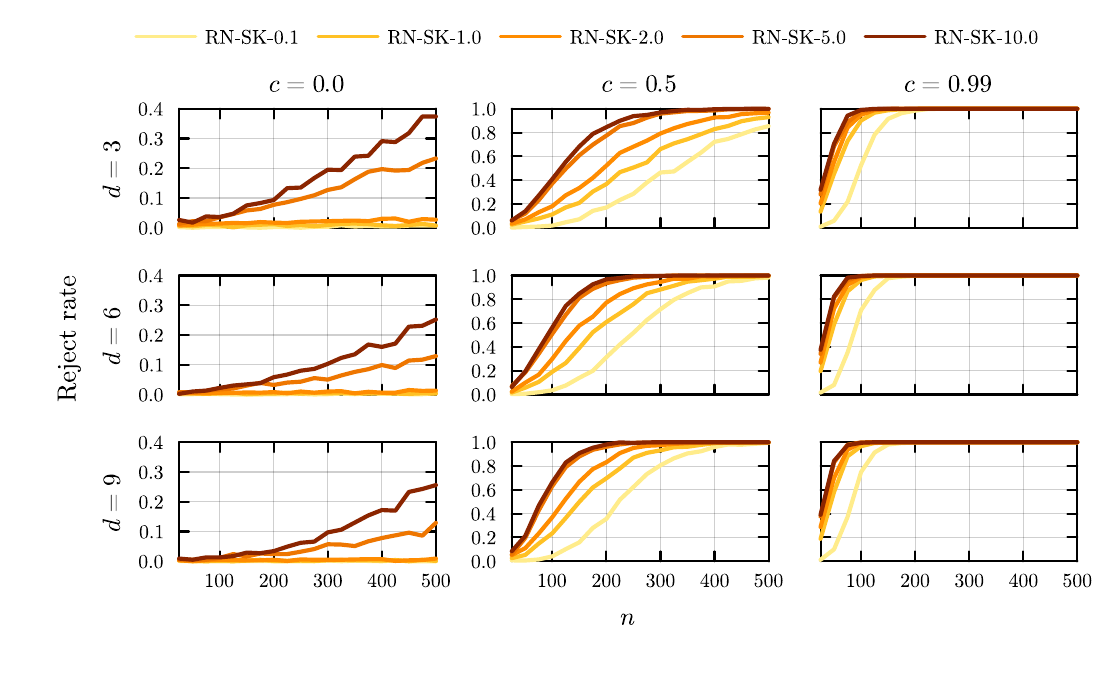}

    \caption{Rejection rates for $\SO{d}$-equivariance \nRNFS and \nRNSK tests for varying values of the $Y|X$ covariance matrix off-diagonal values $c$ and sample size $n$.
    The number suffix in the test labels indicates the constant scaling applied to the bandwidth on the maximal invariant kernel used for conditional randomization, where the test with a scaling of $1.0$ is the one used in other experiments.  (Note the different vertical scale for $c = 0$ versus $c > 0$.)}
    \label{fig:exp:bw}
\end{figure}

For the same setting as in \cref{sec:exp:gauss}, \cref{fig:exp:bw} shows how approximating the conditional randomization via kernel methods can change the size and power of the FUSE and SK tests.
Here, we apply different constant scalings to the bandwidth on the kernel used in the approximate conditional randomization method described in \cref{sec:rand:approx}, where the default bandwidth follows Silverman's rule of thumb for multivariate kernels \citep[][Ch.~4]{Silverman:1986}.
The rule says to take the bandwidth matrix $\bfH$ as a diagonal matrix with the diagonal elements being
\begin{align*}
\left(\frac{4}{d+2}\right)^{\frac{2}{d+4}}n^{\frac{-2}{d+4}}(s\hat\sigma_j)^2 \;,
\end{align*}
where $d$ is the dimension of $\maxInv$, $\hat\sigma^2$ is the sample variance in the $j$-th dimension, and $s=1$ is the default scaling.
In our setting, we observe that with less smoothing ($s\in\{0.1,1,2\}$), the test is conservative with a rejection rate less than $\sig=0.05$, while over-smoothing ($s\in\{5,10\}$) can increase power but also inflate Type I error.

\subsection{Risks in random inference}\label{apx:exp:risk}

\begin{figure}[tbp]
    \centering
    \includegraphics[trim={29 32 6 2},clip,width=\textwidth]{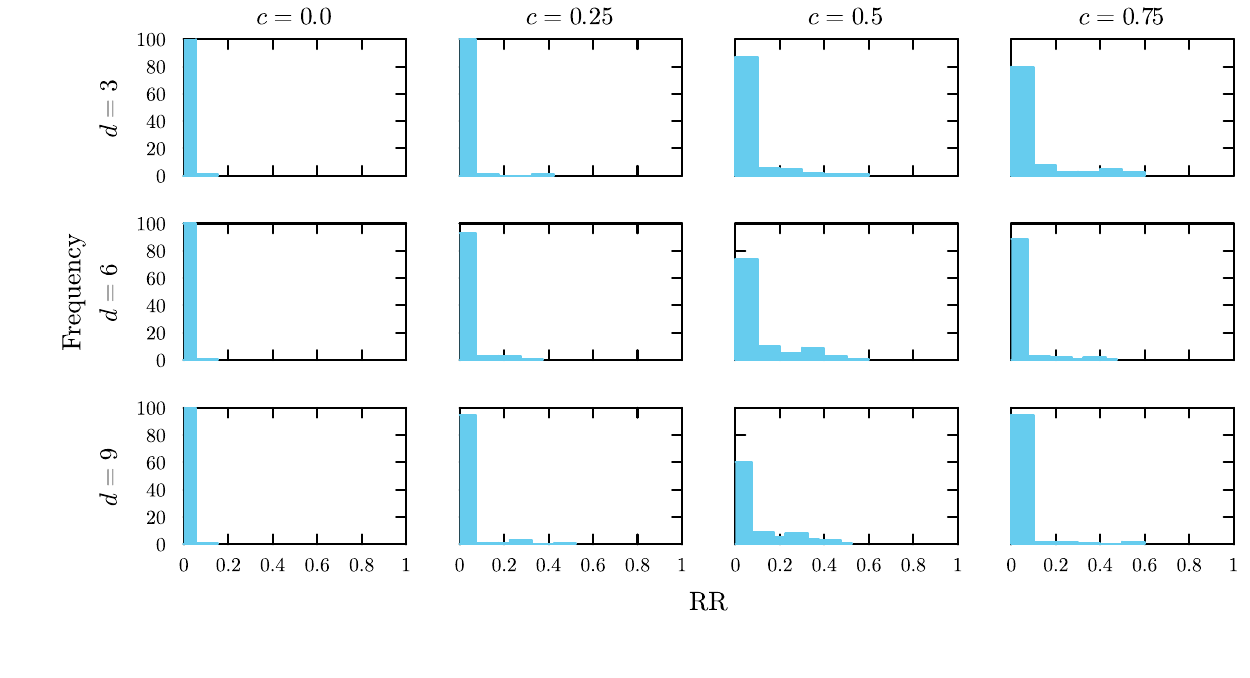}
    \vspace{1em}

    \includegraphics[trim={29 32 6 2},clip,width=\textwidth]{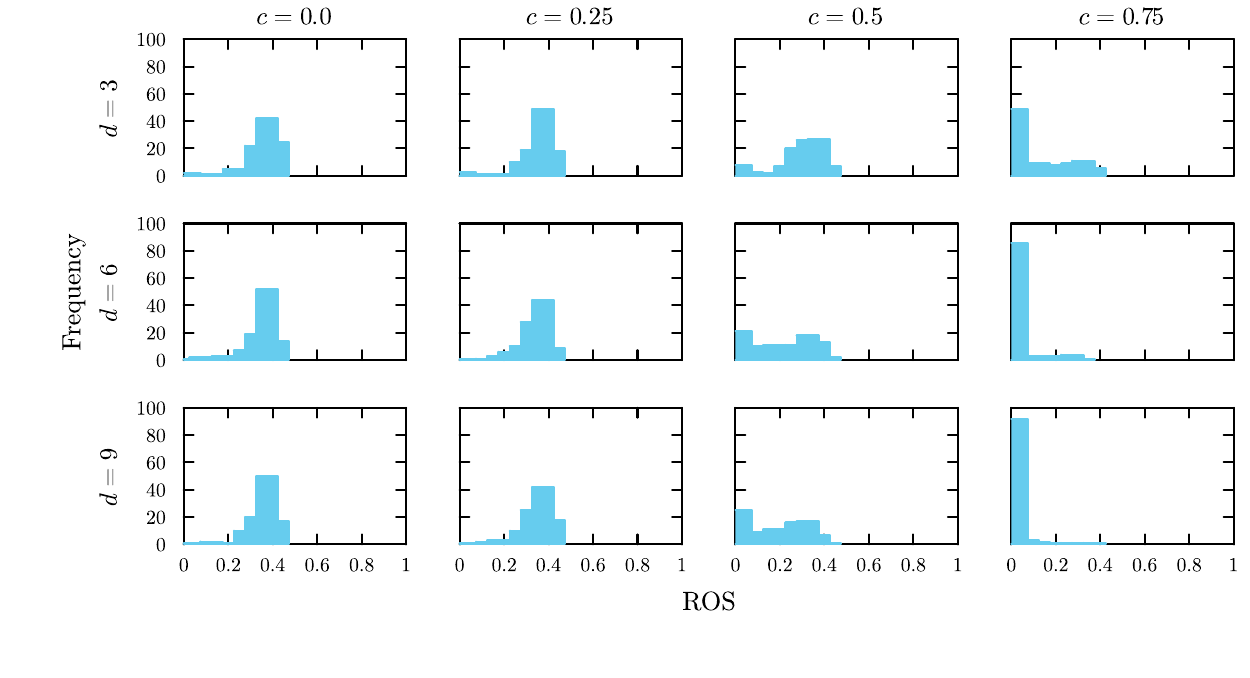}

    \caption{Histograms for the estimated resampling risk (RR) and risk of overestimated significance (ROS) for \nRNFS in testing for $\SO{d}$-equivariance with $B=100$ and $n=250$ across 100 simulations.}
    \label{fig:exp:risk}
\end{figure}

For the FUSE test for $\SO{d}$-equivariance in the setting described in \cref{sec:exp:gauss}, we estimate its resampling risk (RR) and risk of overestimated significance (ROS) \citep{stoepker2024inference}.
These quantities measure risks in inference that arise purely due to the Monte Carlo estimation of the $p$-value.
Given a fixed dataset $z_{1:n}$, let $p^*(z_{1:n})$ denote the ``true'' $p$-value (the value obtained from $p_B$ in \eqref{eqn:crt:pval} by taking the number of randomizations $B\to\infty$), and let $\test^*(z_{1:n})\in\{0,1\}$ denote the result of the test based on $p^*(z_{1:n})$.
The RR measures the probability of obtaining a test decision differing from the one based on $p^*(z_{1:n})$ as a consequence of the Monte Carlo estimation:
\begin{align*}
\mathrm{RR}(z_{1:n}) = \bbP\left\{\test_B(z_{1:n}) \neq \test^*(z_{1:n})\right\} \;.
\end{align*}
The ROS measures the probability of obtaining a $p$-value smaller than $p^*(z_{1:n})$:
\begin{align*}
\mathrm{ROS}(z_{1:n}) = \bbP\left\{p_B(z_{1:n}) < p^*(z_{1:n})\right\} \;.
\end{align*}
In this experiment, we estimate $p^*(z_{1:n})$ by running the test with $B^*=10,000$ and using that estimate to compute the above theoretical probabilities for a test with $B=100$.
\Cref{fig:exp:risk} shows histograms of the estimated RR and ROS across $N=100$ independent datasets.
Overall, RR remains low across datasets, indicating that random inference does not greatly affect the result of the test.
On the other hand, the ROS may be nontrivial at times, suggesting that the $p$-values obtained from the test may not be interpretable in terms of significance without using a larger $B$.

\subsection{Randomization using representative inversions}\label{apx:exp:rep:inv}

Consider the same experimental setup as in \cref{sec:exp:gauss} but with tests that directly use the representative inversion defined in \eqref{eq:SOd:rep:inv} as opposed to sampling actions from the inversion kernel $\repInvKern$.
The results are shown in \cref{fig:exp:rep:inv}.
We see that these tests retain their approximate level-$\sig$ size when the true conditional distribution is equivariant $(c=0)$.
We also observe that these tests have slightly different power properties under the alternative compared to those that randomize using the inversion kernel.
\Cref{fig:exp:rep:inv:p} show the corresponding $p$-value histograms for these tests.

\begin{figure}[tbp]
    \centering
    \includegraphics[trim={29 26 0 14},clip,width=\textwidth]{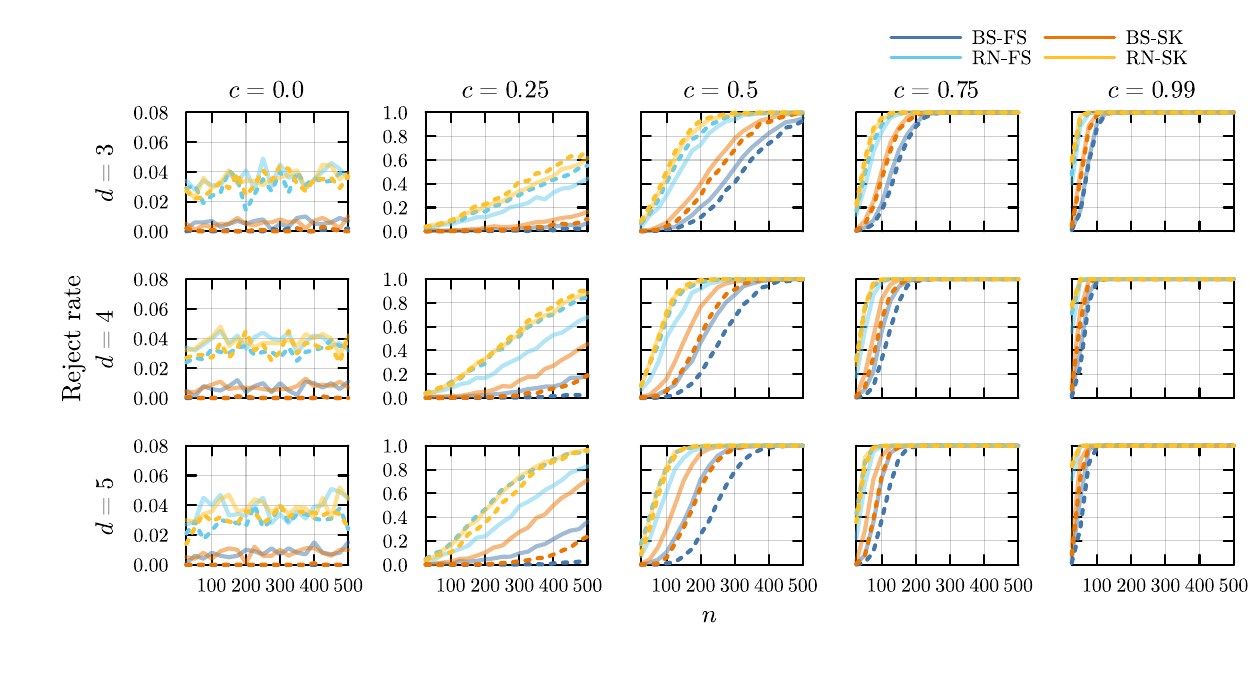}
    \caption{Rejection rates for $\SO{d}$-equivariance tests for varying values of the $Y|X$ covariance matrix off-diagonal values $c$ and sample size $n$.
    The dotted lines represent tests that randomize using a deterministic representative inversion, and the solid lines represent tests that randomize using the inversion kernel.  (Note the different vertical scale for $c = 0$ versus $c > 0$.)}
    \label{fig:exp:rep:inv}
\end{figure}

\begin{figure}[tbp]
    \centering
    \includegraphics[trim={24 22 0 8},clip,width=0.9\textwidth]{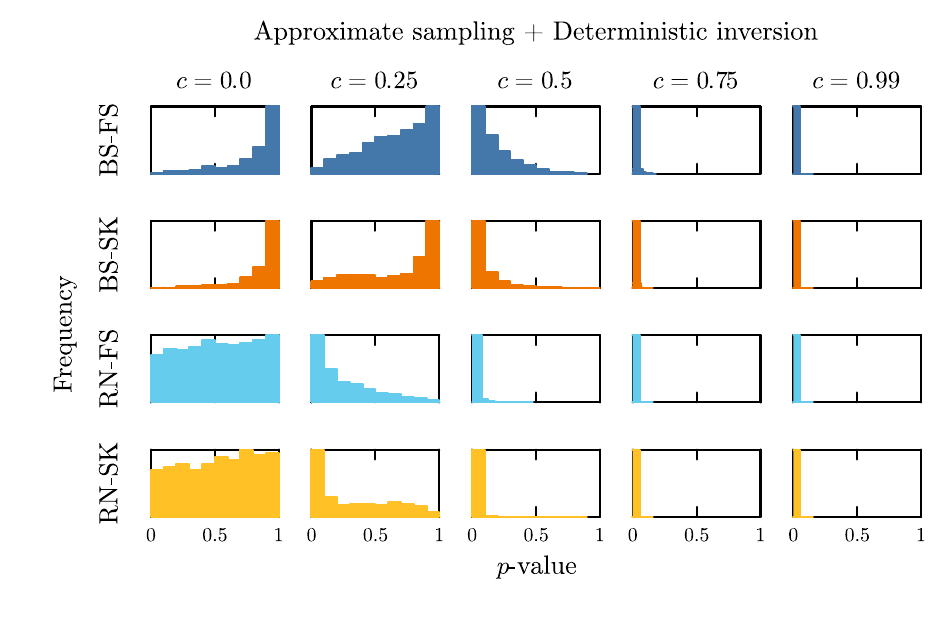}
    
    \caption{Histograms showing the distribution of $1000$ $p$-values of the $\SO{3}$-equivariant tests that use deterministic representative inversions at sample size $n=250$. Parameter $c$ is the covariance between any two dimensions.}
    \label{fig:exp:rep:inv:p}
\end{figure}

\subsection{Test runtime}\label{apx:exp:comp}

\Cref{tab:comp} shows the mean wall time (in seconds, computed across 1000 simulations) for the FUSE test for $\SO{d}$-equivariance on 100 observations generated the same way as in \cref{sec:exp:gauss}.
Observe that the runtime increases linearly with the number of randomizations $B$, and a single run of the test takes no longer than a minute even with 10,000 randomizations.
We also see that reusing the comparison set can noticeably decrease computation time, while the benefits of using a deterministic representative inversion are less significant in this setting.\footnote{The results suggest that using a deterministic inversion can sometimes \emph{slow down} the test, but this is likely due to other factors, such as variations across server allocations.}
It can also be observed that the effect of data dimension on runtime is minimal with our particular implementation for these tests.

Note that the FUSE test generally has a longer runtime compared to that of the baseline and SK tests, which we do not show here.

\begin{table}[tbp]
\centering
\caption{\label{tab:comp}Mean wall time (in seconds) across $1000$ simulations for variations of the FUSE test for $\SO{d}$-equivariance on samples of size $n=100$.}
\begin{tabular}{@{\extracolsep\fill}rccrrr@{\extracolsep\fill}}
\toprule
&&& \multicolumn{3}{@{}c@{}}{$B$} \\
\cmidrule(l){4-6}
& Reuse comparison set & Deterministic inversion & $10^2$ & $10^3$ & $10^4$ \\
\midrule
\multirow{4}{*}{$d=3$} & & & $0.32$ & $3.29$ & $29.66$ \\
\cmidrule(l){2-6}
& $\checkmark$ & & $0.24$ & $2.22$ & $22.52$ \\
\cmidrule(l){2-6}
& & $\checkmark$ & $0.30$ & $3.09$ & $28.28$ \\
\cmidrule(l){2-6}
& $\checkmark$ & $\checkmark$ & $0.25$ & $2.19$ & $24.41$ \\
\midrule
\multirow{4}{*}{$d=6$} & & & $0.32$ & $3.41$ & $34.31$ \\
\cmidrule(l){2-6}
& $\checkmark$ & & $0.24$ & $2.27$ & $25.13$ \\
\cmidrule(l){2-6}
& & $\checkmark$ & $0.32$ & $3.04$ & $32.91$ \\
\cmidrule(l){2-6}
& $\checkmark$ & $\checkmark$ & $0.23$ & $2.35$ & $21.14$ \\
\midrule
\multirow{4}{*}{$d=9$} & & & $0.35$ & $3.07$ & $33.02$ \\
\cmidrule(l){2-6}
& $\checkmark$ & & $0.25$ & $2.34$ & $24.97$ \\
\cmidrule(l){2-6}
& & $\checkmark$ & $0.32$ & $3.06$ & $31.48$ \\
\cmidrule(l){2-6}
& $\checkmark$ & $\checkmark$ & $0.24$ & $2.27$ & $25.10$ \\
\bottomrule
\end{tabular}
\end{table}

\subsection{Sensitivity to non-equivariance in the mean} \label{apx:exp:mvn:sens}

For $d=3$, we consider two other forms of non-equivariance under the action of $\SO{d}$ that could manifest in the conditional distribution.
Consider the conditional distribution
\[
Y\mid X \sim
\begin{cases}
\gaussian(X+s\onevec_d, \I_d) & \text{if }\cos^{-1}\left(\frac{X_1}{\sqrt{X_1^2+X_2^2}}\right) < q \\
\gaussian(X,\I_d) & \text{otherwise}
\end{cases}
\;.
\]
Non-equivariance under rotations about the origin is controlled by two parameters with this conditional distribution:
\begin{enumerate}
    \item A proportion $q$ of rotations that leads to a shift in the mean, determined by the 2D angle between the vector spanned by the first two entries of $X$ and the axis of the first dimension.
    The conditional distribution is equivariant at both $q=0$ and $q=1$, and most non-equivariant at $q=0.5$.
    \item The scaling factor $s$ in the shift in the mean, with the conditional distribution being equivariant only when $s=0$.
\end{enumerate}
For various combinations of $q$ and $s$ and sample sizes ranging from $25$ to $500$, the test rejection rates are shown in \cref{fig:gauss:equiv:sens}.
Our tests are sensitive to both parameters, with powers increasing as the conditional distribution moves away from equivariance in either parameter.

\begin{figure}[tbp]
    \centering
    \includegraphics[trim={29 33 0 13},clip,width=\textwidth]{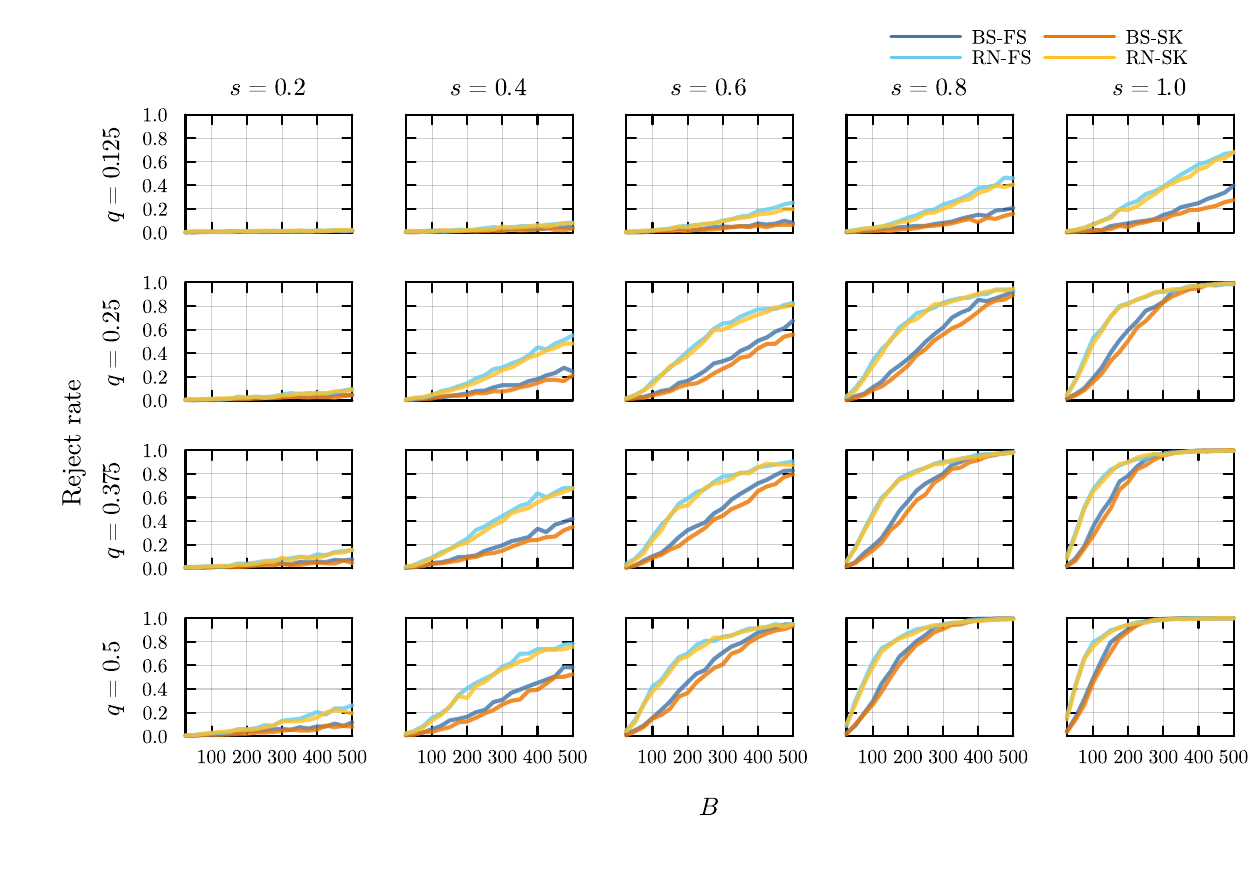}
    
    \caption{Rejection rates for tests
    in the non-equivariant mean experiment, controlled by proportion $q$ and shift $s$ parameters, for varying values of sample size $n$. The conditional distribution is not $\SO{3}$-equivariant in all settings shown.}
    \label{fig:gauss:equiv:sens}
\end{figure}

\subsection{Permutation equivariance}\label{apx:exp:synth:perm}

We consider a synthetic example where we test for equivariance with respect to the action of a different group.
Consider the same $X$ marginal distribution as in \cref{sec:exp:mvn}, but now consider the conditional distribution
\[
Y \mid X \sim
\begin{cases}
\gaussian(X + s\bfe_1, \I_d) & \text{if }X_1 > X_j \;,\; \forall j\in\{2,\dotsc,d\} \\
\gaussian(X, \I_d) & \text{otherwise}
\end{cases}
\;,
\]
where $s$ is some positive scaling factor.
In words, the conditional distribution has a shift in the mean if the first dimension of $X$ is its largest entry.
For data generated under this setup, we test for equivariance with respect to the symmetric group $\Sym{d}$---the group of $d$-permutations.
The conditional distribution is $\Sym{d}$-equivariant if and only if $s=0$.

In this setting, a convenient maximal invariant to condition on is the order statistic $\maxInv(X)=(X_{(1)},\dotsc,X_{(d)})$.
For $d\in\{3,4,5\}$ and various values of $s$, the test rejection rates are shown in \cref{fig:gauss:equiv:perm}.
The results are similar to those of the other synthetic experiments, though we note that the powers of the tests rapidly decrease with increasing dimensions in this particular setting.

\begin{figure}[tbp]
    \centering
    \includegraphics[trim={29 26 0 14},clip,width=\textwidth]{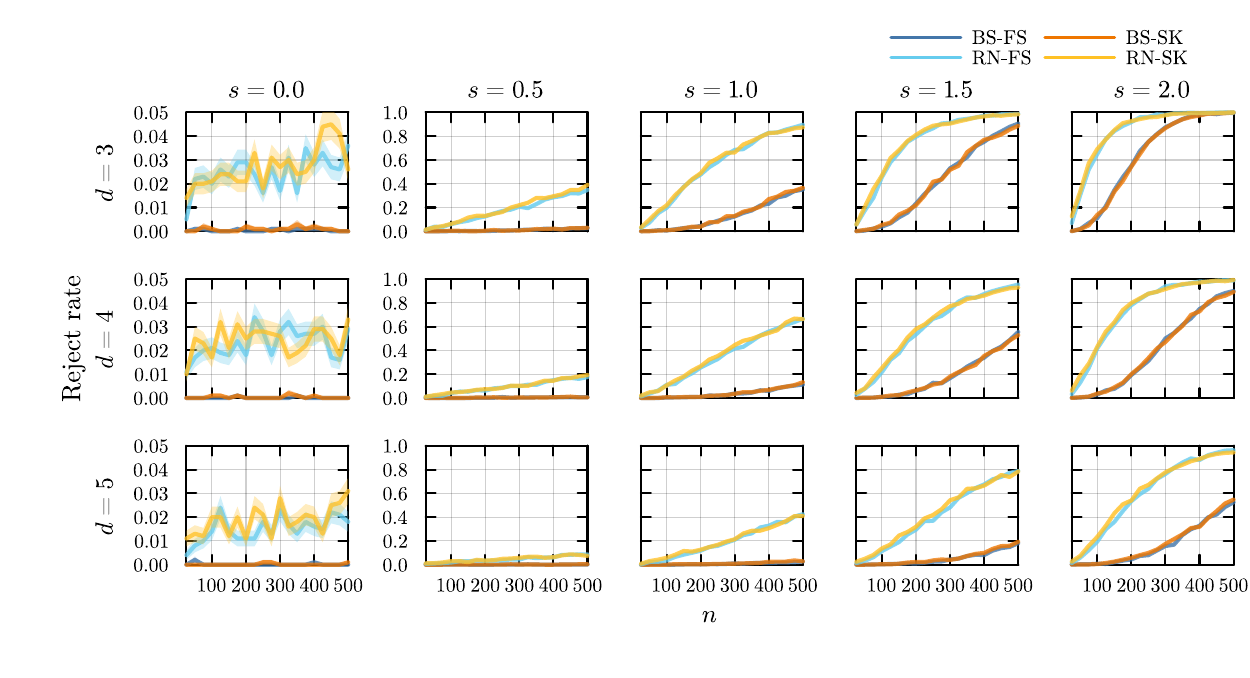}

    \caption{Rejection rates for tests
    for the permutation equivariance experiment for varying values of sample size $n$. The conditional distribution is $\Sym{d}$-equivariant only at $s=0$. (Note the different vertical scale for $s = 0$ versus $s > 0$.)}
    \label{fig:gauss:equiv:perm}
\end{figure}

\subsection{Equivariance as joint invariance in the LHC dijet data} \label{apx:exp:lhc}

\begin{figure}[tb]
    \centering
    \includegraphics[trim={12 0 0 0},clip,width=0.5\textwidth]{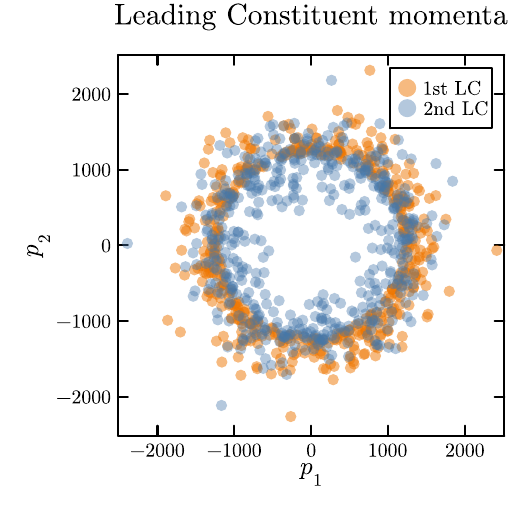}
    \caption{Transverse momentum of $500$ first and second leading constituents (LC) in the LHC dijet data.}
    \label{fig:lhc:data}
\end{figure}

\Cref{fig:lhc:data} shows the distribution of transverse momenta for a sample of $500$ first and second leading constituents from the LHC dijet data.
The momenta $X$ of the first leading constituents appear to be distributed on a ring centered at the origin, suggesting that the distribution of $X$ may be marginally invariant with respect to the action of $\SO{2}$.
As per \cref{prp:eq:joint:inv}, equivariance could then be tested via a test for joint invariance under the action of the group of paired rotations $\grp_0=\{(g_1,g_2)\in\SO{2}\times\SO{2}:g_1=g_2\}$.
We conduct a test for invariance with respect to this subgroup on samples of size $n=25$.
Furthermore, as the presence of unexpected symmetry would also be of scientific interest, we also test for invariance under the action of the full $\grp_1=\SO{2}\times\SO{2}$ group and of $\grp_2=\SO{4}$.
The results based on $B=1000$ randomizations are shown in \cref{tab:lhc}.
We see that our tests detect $\grp_0$-invariance and reject $\grp_1$- and $\grp_2$-invariance, which is again consistent with the Standard Model.

\begin{table}[tb]
\centering
\caption{\label{tab:lhc}Rejection rates (over $1000$ repetitions) for tests for joint invariance on LHC dijet data samples of size $n=25$.}
\begin{tabular}{@{\extracolsep\fill}rrrrr@{\extracolsep\fill}}
\toprule
& \multicolumn{4}{@{}c@{}}{Test} \\
\cmidrule(l){2-5}
Group to test & \nBSFS & \nBSSK & \nRNFS & \nRNSK \\
\midrule
$\text{Paired }\SO{2}\text{-rotations}$ & $0.044$ & $0.052$ & $0.050$ & $0.052$ \\
\midrule
$\text{Unpaired }\SO{2}\text{-rotations}$ & $0.481$ & $0.874$ & $0.958$ & $0.972$ \\
\midrule
$\SO{4}$ & $0.633$ & $0.933$ & $0.989$ & $0.987$ \\
\bottomrule
\end{tabular}
\end{table}

\subsection{Training an invariant MNIST classifier} \label{sec:experiments:MNIST}

As an application under a different context, we illustrate how our tests for symmetry of a conditional distribution $P_{Y|X}$ can be adapted as tests for symmetry of a function $f:\bfX\to\bfY$ by taking $P_{Y|X}=\dirac_{f(X)}$.
This means that our tests could be used as, for example, validation tools for assessing whether a prediction model has learned symmetry after training, requiring only input data and model-predicted output. 
To demonstrate this utility, consider a setting where an MNIST classifier is trained to be invariant with respect to image rotations via data augmentation \citep{Chen:2020,Huang:2022aa}.
Unlike methods that encode symmetry as inductive biases, data augmentation may not always succeed in producing a symmetric model.

The MNIST dataset is a collection of $28\times28$ pixel images of handwritten digits.
The training set consists of $20000$ images; the test set consists of $45000$ augmented images, from which we uniformly sample to use in our tests.
We train two identical LeNet convolutional neural networks for predicting the digit in the image, where one model is trained on images augmented with rotations and the other without.
Before and during training, we pause the training process after every three epochs to test the invariance properties of each model.
To generate the data used in each test, a sample of $n=100$ augmented images is sampled from the test set and fed into the model.
The data $Y$ are then taken as the digit probabilities predicted by the model, and so each observation is a vector on a simplex of dimension $d=10$ (or $d=9$ in the setting where the ``$9$'' images are removed from the data).
In our tests, we take the maximal invariant to be the unaugmented image to avoid complications with orienting images.

As augmentation ensures that the marginal distribution of the images will be invariant, we test for conditional invariance using tests for invariance via \cref{prp:eq:marg:inv}.
For the test statistic, we use the \textit{information diffusion kernel} for simplex data \citep{Lebanon:2002}, given by
\[
k(y,y') \approx (4\pi\sigma)^{-\frac{d+1}{2}}\exp\left(-\frac{1}{\sigma}\arccos^2\left(\textstyle\sum_{i=1}^d\sqrt{y_iy_i'}\right)\right) \;,
\]
where $\sigma$ is the kernel parameter.
To perform randomization in these tests, we apply random rotations to the images and then predict their label probabilities using the model trained until that point in the training process.
We repeat this procedure for $N=1000$ instances at each testing point to estimate the rejection rate of the test, which can be interpreted as a measure of how invariant the model is at that point during training.

\begin{figure}[tbp]
\centering
\includegraphics[trim={21 14 0 0},clip,width=0.65\textwidth]{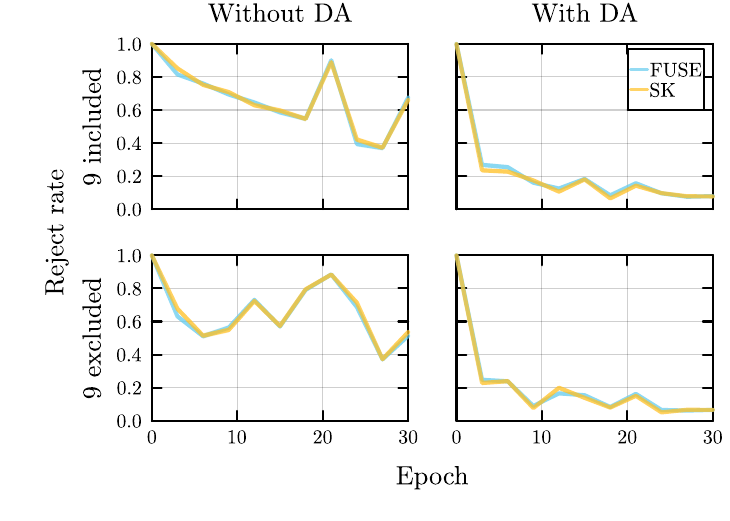}
\caption{\label{fig:mnist}Test rejection rates for image rotation-invariance of two MNIST models during training, where one is trained with data augmentation and the other without.
The top row shows results for when the models are trained on the full MNIST data, and the bottom row shows results for when images with true label $9$ are removed from the data.}
\end{figure}

The estimated rejection rates for the randomization tests are shown in \cref{fig:mnist} for two settings: one where the models are trained on the full MNIST data, and one where images with true label 9 are removed from the data.
In both settings, the rejection rate for the model trained with augmentations is approximately $\sig=0.05$ after the 30th epoch, whereas it is unclear if the rejection rate for the model trained without augmentations is converging to any value.

Here, we used the estimated test rejection rate as a measure of model invariance.
It may be tempting to determine whether symmetry has been learned at a particular time point by whether a single test rejects or not, but conducting a suite of tests at multiple points during training would be subject to multiple testing considerations.
A potential direction for future work is anytime-valid tests for conditional symmetry, which would allow for the continual evaluation of symmetry-learning based on a single test at each time point during training.
Anytime-valid tests for distributional invariance have recently appeared in the literature \citep{Koning:2023,Ramdas:2023:b,Lardy:2024}.

\end{appendix}

\begin{acks}[Acknowledgments]
The authors would like to thank Arthur Gretton and Antonin Schrab for a helpful conversation about randomization tests and MMD estimators. They are also grateful to the editor, associate editor, and two anonymous referees, whose comments and questions helped improve the paper. 
\end{acks}

\begin{funding}
This research was supported in part through computational resources and services provided by Advanced Research Computing at the University of British Columbia.
BBR and KC are supported by the Natural Sciences and Engineering Research Council of Canada (NSERC): RGPIN2020-04995, RGPAS-2020-00095, DGECR-2020-00343. AS is supported by a Canadian Statistical Sciences Institute (CANSSI) Distinguished Postdoctoral Fellowship. 
\end{funding}

\bibliographystyle{imsart-nameyear} %
\bibliography{references}       %

\end{document}